\documentclass[submission]{eptcs}
\providecommand{\event}{COS 2013} % Name of the event you are submitting to
\usepackage{amsmath,amssymb, amsthm}
\usepackage{mathpartir}
\usepackage{url}
\usepackage{hyperref}

\newcommand{\LamDeltadrule}[4][]{{\displaystyle\frac{\begin{array}{l}#2\end{array}}{#3}\quad\LamDeltadrulename{#4}}}

\newcommand{\LamDeltapremise}[1]{ #1 \\}
\newenvironment{LamDeltadefnblock}[3][]{ \framebox{\mbox{#2}} \quad #3 \\[0pt]}{}

\newcommand{\LamDeltant}[1]{\mathit{#1}}
\newcommand{\LamDeltamv}[1]{\mathit{#1}}

\newcommand{\LamDeltasym}[1]{#1}

\newcommand{\LamDeltadrulename}[1]{\textsc{#1}}

\newcommand{\LamDeltadruleAx}[1]{\LamDeltadrule[#1]{%
}{
 \Gamma  \LamDeltasym{,}  \LamDeltamv{x}  \LamDeltasym{:}  \LamDeltant{A}  \vdash  \LamDeltamv{x}  :  \LamDeltant{A} }{%
{\LamDeltadrulename{Ax}}{}%
}}

\newcommand{\LamDeltadruleLam}[1]{\LamDeltadrule[#1]{%
\LamDeltapremise{ \Gamma  \LamDeltasym{,}  \LamDeltamv{x}  \LamDeltasym{:}  \LamDeltant{A}  \vdash  \LamDeltant{t}  :  \LamDeltant{B} }%
}{
 \Gamma  \vdash   \lambda  \LamDeltamv{x} : \LamDeltant{A}  .  \LamDeltant{t}   :   \LamDeltant{A}  \to  \LamDeltant{B}  }{%
{\LamDeltadrulename{Lam}}{}%
}}

\newcommand{\LamDeltadruleApp}[1]{\LamDeltadrule[#1]{%
\LamDeltapremise{ \Gamma  \vdash  \LamDeltant{t_{{\mathrm{2}}}}  :  \LamDeltant{A} }%
\LamDeltapremise{ \Gamma  \vdash  \LamDeltant{t_{{\mathrm{1}}}}  :   \LamDeltant{A}  \to  \LamDeltant{B}  }%
}{
 \Gamma  \vdash  \LamDeltant{t_{{\mathrm{1}}}} \, \LamDeltant{t_{{\mathrm{2}}}}  :  \LamDeltant{B} }{%
{\LamDeltadrulename{App}}{}%
}}

\newcommand{\LamDeltadruleDelta}[1]{\LamDeltadrule[#1]{%
\LamDeltapremise{ \Gamma  \LamDeltasym{,}  \LamDeltamv{x}  \LamDeltasym{:}   \neg  \LamDeltant{A}   \vdash  \LamDeltant{t}  :   \perp  }%
}{
 \Gamma  \vdash   \Delta  \LamDeltamv{x}  :   \neg  \LamDeltant{A}   .  \LamDeltant{t}   :  \LamDeltant{A} }{%
{\LamDeltadrulename{Delta}}{}%
}}

\newcommand{\LamDeltadruleBeta}[1]{\LamDeltadrule[#1]{%
}{
 \LamDeltasym{(}   \lambda  \LamDeltamv{x} : \LamDeltant{T}  .  \LamDeltant{t}   \LamDeltasym{)} \, \LamDeltant{t'}  \redto  \LamDeltasym{[}  \LamDeltant{t'}  \LamDeltasym{/}  \LamDeltamv{x}  \LamDeltasym{]}  \LamDeltant{t} }{%
{\LamDeltadrulename{Beta}}{}%
}}

\newcommand{\LamDeltadruleStructRed}[1]{\LamDeltadrule[#1]{%
\LamDeltapremise{ \LamDeltamv{y}  \text{ fresh in }  \LamDeltant{t}  \text{ and }  \LamDeltant{t'} }%
\LamDeltapremise{ \LamDeltamv{z}  \text{ fresh in }  \LamDeltant{t}  \text{ and }  \LamDeltant{t'} }%
}{
 \LamDeltasym{(}   \Delta  \LamDeltamv{x}  :   \neg  \LamDeltasym{(}   \LamDeltant{T_{{\mathrm{1}}}}  \to  \LamDeltant{T_{{\mathrm{2}}}}   \LamDeltasym{)}   .  \LamDeltant{t}   \LamDeltasym{)} \, \LamDeltant{t'}  \redto   \Delta  \LamDeltamv{y}  :   \neg  \LamDeltant{T_{{\mathrm{2}}}}   .  \LamDeltasym{[}   \lambda  \LamDeltamv{z} :  \LamDeltant{T_{{\mathrm{1}}}}  \to  \LamDeltant{T_{{\mathrm{2}}}}   .  \LamDeltasym{(}  \LamDeltamv{y} \, \LamDeltasym{(}  \LamDeltamv{z} \, \LamDeltant{t'}  \LamDeltasym{)}  \LamDeltasym{)}   \LamDeltasym{/}  \LamDeltamv{x}  \LamDeltasym{]}  \LamDeltant{t}  }{%
{\LamDeltadrulename{StructRed}}{}%
}}

\newtheorem{thm}{Theorem}
\newtheorem{lemma}[thm]{Lemma}
\newtheorem{example}[thm]{Example}
\newtheorem{corollary}[thm]{Corollary}
\newtheorem{definition}[thm]{Definition}

\newcommand{\interp}[1]{[\negthinspace[#1]\negthinspace]}
\newcommand{\normto}[0]{\rightsquigarrow^{!}}

\newcommand{\redto}[0]{\rightsquigarrow}

\begin{document}

\title{Hereditary Substitution for the $\lambda \Delta$-Calculus}
\author{Harley Eades and Aaron Stump\\
  \institute{Department of Computer Science\\University of Iowa}}
\def\titlerunning{Hered. Subst. for the $\lambda\Delta$-Calculus}
\def\authorrunning{H. Eades and A. Stump}
\maketitle
\begin{abstract} 
  Hereditary substitution is a form of type-bounded iterated
  substitution, first made explicit by Watkins et al. and Adams in
  order to show normalization of proof terms for various constructive
  logics. This paper is the first to apply hereditary substitution to
  show normalization of a type theory corresponding to a
  non-constructive logic, namely the $\lambda\Delta$-calculus as
  formulated by Rehof. We show that there is a non-trivial
  extension of the hereditary substitution function of the
  simply-typed $\lambda$-calculus to one for the
  $\lambda\Delta$-calculus. Then hereditary substitution is used to
  prove normalization.
\end{abstract}

\section{Introduction}
\label{sec:introduction}
In 1992 M. Parigot defined an algorithmic interpretation of classical
natural deduction called the $\lambda\mu$-calculus
\cite{Parigot:1992}.  His original theory consisted of complete
sequents for types which often made the theory difficult to reason
about, especially when one wished to adapt any well-known results of
intuitionistic type theory to his classical type theory.  Later, in
1994 N. Rehof and M. S\o rensen defined the $\lambda\Delta$-calculus
which is provably equivalent to the $\lambda\mu$-calculus \footnote{By
  this we mean that everything provable in the $\lambda\mu$-calculus
  is provable in the $\lambda\Delta$-calculus, but $\beta$-reduction
  is not step-by-step equivalent.}.  Due to their equivalence, any
results obtained for the $\lambda\mu$-calculus apply to the
$\lambda\Delta$-calculus by translation. Now the
$\lambda\Delta$-calculus is essentially an extension of the simply
typed $\lambda$-calculus (STLC), hence adapting results from intuitionistic
type theory to the $\lambda\Delta$-calculus is less complicated.  The
main result of this paper is the adaptation of a well-known proof
technique for showing normalization of intuitionistic typed
$\lambda$-calculi, called hereditary substitution, to the
$\lambda\Delta$-calculus.  We stress that proving normalization of the
$\lambda\Delta$-calculus is not our contribution. This is already well
known \cite{David:2003,Parigot:1997}.  In fact it is strongly normalizing.
The adaptation of the proof method to the $\lambda\Delta$-calculus is however our main
contribution.

The central idea behind the hereditary substitution proof method is to
prove normalization using a lexicographic combination of an ordering
on types and the strict subexpression ordering on proofs. This central
idea has been used in normalization proofs dating all the way back to
Prawitz in 1965. Since then it has been used to show normalization for
many simply-typed $\lambda$-calculi
\cite{Amadio:1998,Girard:1989,Joachimski:1999,Levy:1976}.
% In \cite{Joachimski:1999} F. Joachimski and R. Matthes show weak and
% strong normalization for STLC using a lexicographic combination of
% the strict subexpression ordering on types and proofs. They then
% extend this to show strong normalization for STLC extended with sum
% types and commuting conversions. Furthermore, they extend their
% method again to show strong normalization of G\"odel's System T. For
% each system they also define a normalization function on normal
% terms which amounts to a function similar to hereditary
% substitution.
Extracting the constructive content from these proofs one will obtain
a function much like capture avoiding substitution, except when a
redex which was not present in the input is created as a result of
substitution, that redex is recursively reduced. This substitution
function is called hereditary substitution.  It was first made
explicit by K.  Watkins et al. in \cite{Watkins:2004} for
non-dependent types and R.  Adams in \cite{Adams:2004} for dependent
types.  In previous work, the authors showed how to apply
the hereditary-substitution method to prove normalization of
Stratified System F (SSF), a type theory of predicative polymorphism
studied by Leivant~\cite{Eades:2010,Leivant:1991}.

The motivation for using the hereditary substitution method over other
well-known methods for showing normalization is that it is simpler. It
provides a directly defined substitution which preserves normal forms. Its
definition is essentially a combination of the reduction relation with
capture avoiding substitution.
% In comparison to using the Tait-Girard reducibility method it is
% qualitatively less complex, because in the soundness theorem when
% using hereditary substitution no universal quantification over
% substitutions is needed.  Another motivating factor is that its
% simplicity makes it easier to apply to more theories and formalize in
% a proof assistant or automated theorem prover.  However, which
% theories it can be applied to is the question we hope to contribute to
% in this paper.  Up to this work the hereditary substitution function
% has only been defined for and used by intuitionistic typed
% $\lambda$-calculi.
It has found important application in logical frameworks based on
canonical forms \cite{Watkins:2004}.  In these frameworks hereditary
substitution replaces ordinary capture avoiding substitution in order
to maintain canonicity.  

We begin with defining the $\lambda\Delta$-calculus and presenting
some basic meta-results in Sect.~\ref{sec:the_lambda_delta-calculus}
and Sect.~\ref{sec:basic_syntactic_lemmas}.  Then we give the
definition of the hereditary substitution function for the simply
typed $\lambda$-calculus in Sect.~\ref{sec:the_hereditary_substitution_function_stlc}.  
In Sect.~\ref{subsec:the_final_extension} we extend
the definition of the hereditary substitution function for the simply
typed $\lambda$-calculus with a new function called the structural
hereditary substitution function.  Then its correctness properties are
presented in
Sect.~\ref{subsec:properties_of_the_hereditary_substitution_function}.
The hereditary substitution function is then used to conclude
normalization of the $\lambda\Delta$-calculus in
Sect.~\ref{sec:concluding_normalization}.  We conclude with a related
work section in Sect.~\ref{sec:related_work}.
% section introduction (end)

%
\section{The $\lambda\Delta$-Calculus}
\label{sec:the_lambda_delta-calculus}
The $\lambda\Delta$-calculus is a straightforward extension of the
simply typed $\lambda$-calculus.  The syntax is defined in
Figure~\ref{fig:syntax_mu}.
\begin{figure}[h]
    \begin{center}
      \begin{tabular}{lll}
        \begin{math}
          \begin{array}{rrll}
            \text{(Types)} & \LamDeltant{T},\LamDeltant{A},\LamDeltant{B},\LamDeltant{C} & ::= & 
               \perp \,|\,\LamDeltant{b}\,|\, \LamDeltant{A}  \to  \LamDeltant{B} \\
            \text{(Terms)} & \LamDeltant{t} & ::= & 
              \LamDeltamv{x}\,|\, \lambda  \LamDeltamv{x} : \LamDeltant{T}  .  \LamDeltant{t} \,|\, \Delta  \LamDeltamv{x}  :  \LamDeltant{T}  .  \LamDeltant{t} \,|\,\LamDeltant{t_{{\mathrm{1}}}} \, \LamDeltant{t_{{\mathrm{2}}}}\\
            \text{(Normal Forms)} & \LamDeltant{n},\LamDeltant{m} & ::= & 
              \LamDeltamv{x}\,|\, \lambda  \LamDeltamv{x} : \LamDeltant{T}  .  \LamDeltant{n} \,|\, \Delta  \LamDeltamv{x}  :  \LamDeltant{T}  .  \LamDeltant{n} \,|\,\LamDeltant{h} \, \LamDeltant{n}\\
            \text{(Heads)} & \LamDeltant{h} & ::= & \LamDeltamv{x}\,|\,\LamDeltant{h} \, \LamDeltant{n}\\
            \text{(Contexts)} & \Gamma & ::= &  \cdot \,|\,\LamDeltamv{x}  \LamDeltasym{:}  \LamDeltant{A}\,|\,\Gamma_{{\mathrm{1}}}  \LamDeltasym{,}  \Gamma_{{\mathrm{2}}}\\
          \end{array}
        \end{math}
      \end{tabular}
    \end{center}
    \caption{Syntax for Types and Terms}
    \label{fig:syntax_mu}
\end{figure}
The type $\LamDeltant{b}$ is an arbitrary base type.  Negation is defined as it
is in intuitionistic type theory, that is, $ \neg  \LamDeltant{A}  =^{def}  \LamDeltant{A}  \to   \perp  $, where $ \perp $ is absurdity.  Arbitrary syntactically defined
normal forms will be denoted by the meta-variables $\LamDeltant{n}$ and
$\LamDeltant{m}$, and arbitrary typing contexts will be denoted by the
meta-variable $\Gamma$.  We assume at all times that all variables in
the domain of $\Gamma$ are unique.  In addition we rearrange the
objects in $\Gamma$ freely without indication.

The typing rules are defined in Figure~\ref{fig:typing}.  The
operational semantics are the compatible closure of the rules in
Figure~\ref{fig:opsem}.
\begin{figure}[h]
  \begin{center}
    \begin{math}
      \begin{array}{ccc}
        \LamDeltadruleAx{} & \LamDeltadruleLam{} \\
        & \\
        \LamDeltadruleApp{} & \LamDeltadruleDelta{} \\
      \end{array}
    \end{math}
  \end{center}

  \caption[]{Typing Rules}
  \label{fig:typing}
\end{figure}
\begin{figure}[h]
    \begin{center}
    \begin{math}
      \begin{array}{ccc}
        \LamDeltadruleBeta{} \\
        & \\
        \LamDeltadruleStructRed{}
      \end{array}
    \end{math}
  \end{center}

  \caption{Operational Semantics}
  \label{fig:opsem}
\end{figure}
It is easy to see based on the typing rules that the
$\Delta$-abstraction is the introduction form for double negation.  We
annotate the $\Delta$-abstraction with the type of the bound variable
to make the definition of the hereditary substitution function a
little less complicated.  Removing this annotation should not cause
any significant problems.  On a more programmatic front the
$\Delta$-abstraction is a control operator.  It can simulate
Felleisen's control operators; see \cite{Rehof:1994} for more
information on this.  N. Rehof and M. S\o rensen also extend the
operational semantics with a structural reduction rule for the
$\Delta$-abstraction (\LamDeltadrulename{StructRed} in
Figure~\ref{fig:opsem}). This rule is called structural because it
does not amount to a computational step, rather pushes the application
into the body of the $\Delta$-abstraction potentially creating
additional redexes.  We denote the reflexive and transitive closure of
$ \redto $ as $ \redto^* $. We also define $\LamDeltant{t}  \redto^!  \LamDeltant{t'}$ to mean that $\LamDeltant{t}  \redto^*  \LamDeltant{t'}$ and $\LamDeltant{t'}$ is normal.  Now that we have defined the
$\lambda\Delta$-calculus we state several well-known meta-results that
will be needed throughout the sequel.
% section the_lambda_delta-calculus (end) 

%
\section{Basic Syntactic Lemmas}
\label{sec:basic_syntactic_lemmas}
The following meta-results are well-known so we omit their proofs.  We
do not always explicitly state the use of these results.  The first
two properties are weakening and substitution for the typing relation.

\begin{lemma}[Weakening for Typing]
  \label{lemma:weakening_for_typing}
  If $ \Gamma  \vdash  \LamDeltant{t}  :  \LamDeltant{T} $ then $ \Gamma  \LamDeltasym{,}  \LamDeltamv{x}  \LamDeltasym{:}  \LamDeltant{T'}  \vdash  \LamDeltant{t}  :  \LamDeltant{T} $ for any fresh variable $\LamDeltamv{x}$ and
  type $\LamDeltant{T'}$.
\end{lemma}
\begin{proof}
  Straightforward induction on the assumed typing derivation.
\end{proof}

\begin{lemma}[Substitution for Typing]
  \label{lemma:substitution_for_typing}
  If $ \Gamma  \vdash  \LamDeltant{t}  :  \LamDeltant{T} $ and $ \Gamma  \LamDeltasym{,}  \LamDeltamv{x}  \LamDeltasym{:}  \LamDeltant{T}  \LamDeltasym{,}  \Gamma'  \vdash  \LamDeltant{t'}  :  \LamDeltant{T'} $ then $ \Gamma  \vdash  \LamDeltasym{[}  \LamDeltant{t}  \LamDeltasym{/}  \LamDeltamv{x}  \LamDeltasym{]}  \LamDeltant{t'}  :  \LamDeltant{T'} $.
\end{lemma}
\begin{proof}
  Straightforward induction on the second assumed typing derivation.
\end{proof}
\noindent
The final three properties are, confluence, type preservation and inversion of the typing
relation. The proof of the confluence and type preservation can be found in \cite{Rehof:1994}
and the proof of the latter is trivial.

\begin{thm}[Confluence]
  \label{thm:confluence}
  If $\LamDeltant{t_{{\mathrm{1}}}}  \redto^*  \LamDeltant{t_{{\mathrm{2}}}}$ and $\LamDeltant{t_{{\mathrm{1}}}}  \redto^*  \LamDeltant{t_{{\mathrm{3}}}}$, then there exists a term $\LamDeltant{t_{{\mathrm{4}}}}$, such that,
  $\LamDeltant{t_{{\mathrm{2}}}}  \redto^*  \LamDeltant{t_{{\mathrm{4}}}}$ and $\LamDeltant{t_{{\mathrm{3}}}}  \redto^*  \LamDeltant{t_{{\mathrm{4}}}}$.
\end{thm}

\begin{thm}[Preservation]
  \label{thm:preservation}
  If $ \Gamma  \vdash  \LamDeltant{t}  :  \LamDeltant{T} $ and $\LamDeltant{t} \redto \LamDeltant{t'}$ then $ \Gamma  \vdash  \LamDeltant{t'}  :  \LamDeltant{T} $.
\end{thm}

\begin{thm}[Inversion]
  \label{theorem:inversion}
  \begin{itemize}
  \item[]
  \item[i.] If $ \Gamma  \vdash  \LamDeltamv{x}  :  \LamDeltant{T} $ then $\LamDeltamv{x} \in \Gamma$.
  \item[ii.] If $ \Gamma  \vdash   \lambda  \LamDeltamv{x} : \LamDeltant{T_{{\mathrm{1}}}}  .  \LamDeltant{t}   :   \LamDeltant{T_{{\mathrm{1}}}}  \to  \LamDeltant{T_{{\mathrm{2}}}}  $ then $ \Gamma  \LamDeltasym{,}  \LamDeltamv{x}  \LamDeltasym{:}  \LamDeltant{T_{{\mathrm{1}}}}  \vdash  \LamDeltant{t}  :  \LamDeltant{T_{{\mathrm{2}}}} $.
  \item[iii.] If $ \Gamma  \vdash   \Delta  \LamDeltamv{x}  :   \neg  \LamDeltant{T}   .  \LamDeltant{t}   :  \LamDeltant{T} $ then $ \Gamma  \LamDeltasym{,}  \LamDeltamv{x}  \LamDeltasym{:}   \neg  \LamDeltant{T}   \vdash  \LamDeltant{t}  :   \perp  $.
  \end{itemize}
\end{thm}
\begin{proof}
  This can be shown by straightforward induction on the assumed typing derivations.
\end{proof}
\noindent
At this point we have everything we need to state and prove correct
the hereditary substitution function.
% section basic_syntactic_lemmas (end)

%
\section{The Hereditary Substitution Function for STLC}
\label{sec:the_hereditary_substitution_function_stlc}
In the introduction we gave an informal definition of the hereditary
substitution function.  It is exactly like capture-avoiding
substitution, except that if any redexes are introduced as a result of
substitution those redexes are recursively reduced.  In fact
hereditary substitution in general does not modify any redexes already
in the input.  However, if there are no redexes present in the input
then the output of the hereditary substitution function will not have
any either.  This is one of the main correctness properties of the
hereditary substitution function.

The definition of the hereditary substitution function strongly
depends on the existence of an ordering on types.  In fact if no such
ordering exists then it is unclear if the hereditary substitution can
be defined and proved correct.  We say ``unclear'' here because it is
not known if there exists a means of proving the hereditary
substitution function correct without an ordering on types.  We
conjecture that one may be able to give some semantic interpretation
of hereditary substitution and show correctness with respect to the
semantics. However, this is just a conjecture.  Fortunately, a very
simple ordering exists on the types of STLC and the
$\lambda\Delta$-calculus.
\begin{definition}
  \label{def:ordering}
  We define an ordering on types $\LamDeltant{T}$ as the compatible closure of
  the following formulas.
  \begin{center}
    \begin{math}
      \begin{array}{lll}
         \LamDeltant{A}  \to  \LamDeltant{B}  & > & \LamDeltant{A}\\
         \LamDeltant{A}  \to  \LamDeltant{B}  & > & \LamDeltant{B}\\
      \end{array}
    \end{math}
  \end{center}
\end{definition}
The ordering defined above is simply the strict subexpression ordering
on types where the absurdity and base types are minimal elements.
This ordering is clearly well founded.

The definition of the hereditary substitution function depends on
being able to detect when a new redex has been created as a result of
substitution.  When a new redex is created it must be able to also
detect that the ordering on types has decreased.  To detect both of
these situations the hereditary substitution function uses the
following partial function.
\begin{definition}
  \label{def:ctype}
  We define the partial function $ \textsf{ctype} $ 
  which computes the type of an application in head normal form.  It is defined as follows:
  \begin{center}
    \begin{itemize}
    \item[] $ \textsf{ctype}_{ \LamDeltant{T} }( \LamDeltamv{x} , \LamDeltamv{x} )  = \LamDeltant{T}$
    \item[] $ \textsf{ctype}_{ \LamDeltant{T} }( \LamDeltamv{x} , \LamDeltant{t_{{\mathrm{1}}}} \, \LamDeltant{t_{{\mathrm{2}}}} )  = \LamDeltant{T''}$\\
      \begin{tabular}{lll}
        & Where $ \textsf{ctype}_{ \LamDeltant{T} }( \LamDeltamv{x} , \LamDeltant{t_{{\mathrm{1}}}} )  =  \LamDeltant{T'}  \to  \LamDeltant{T''} $.
      \end{tabular}    
    \end{itemize}
  \end{center}
\end{definition}
The following lemma list the most important results about the $ \textsf{ctype} $ function.
\begin{lemma}[Properties of $ \textsf{ctype} $]
  \label{lemma:ctype_props}
  \begin{itemize}
    \item[]
  \item[i.] If $ \textsf{ctype}_{ \LamDeltant{T} }( \LamDeltamv{x} , \LamDeltant{t} )  = \LamDeltant{T'}$ then $ \textsf{head}( \LamDeltant{t} )  = \LamDeltamv{x}$ and $\LamDeltant{T'} \leq \LamDeltant{T}$.
    
  \item[ii.] If $ \Gamma  \LamDeltasym{,}  \LamDeltamv{x}  \LamDeltasym{:}  \LamDeltant{T}  \LamDeltasym{,}  \Gamma'  \vdash  \LamDeltant{t}  :  \LamDeltant{T'} $ and $ \textsf{ctype}_{ \LamDeltant{T} }( \LamDeltamv{x} , \LamDeltant{t} )  = \LamDeltant{T''}$ then
    $\LamDeltant{T'} \equiv \LamDeltant{T''}$.    
  \end{itemize}
\end{lemma}
\begin{proof}
  Both cases can be shown by straightforward induction on the structure of $t$.
  The proof can be found in Appendix~\ref{subsec:proof_of_ctype_props}.
\end{proof}
We now have everything we need to state the hereditary substitution
function for STLC. We denote the hereditary substitution function by
$ [  \LamDeltant{t}  /  \LamDeltamv{x}  ]^{ \LamDeltant{A} }  \LamDeltant{t'} $ where $\LamDeltant{A}$ is the type of $\LamDeltamv{x}$ and is called
the cut type, due to the correspondence between hereditary
substitution and cut elimination. In the definition of the hereditary
substitution function it is assumed that all variables are renamed as
to prevent variable capture.  It is also defined with respect to the
termination metric $(\LamDeltant{A},\LamDeltant{t})$ in lexicographic combination of our
ordering on types and the strict subexpression ordering on terms.
\begin{definition}
  \label{def:hereditary_substitution_function}
  The hereditary substitution function is defined as follows:
  \small
  \begin{itemize}
  \item[] $ [  \LamDeltant{t}  /  \LamDeltamv{x}  ]^{ \LamDeltant{A} }  \LamDeltamv{x}  = \LamDeltant{t}$
  \item[] $ [  \LamDeltant{t}  /  \LamDeltamv{x}  ]^{ \LamDeltant{A} }  \LamDeltamv{y}  = \LamDeltamv{y}$\\

  \item[] $ [  \LamDeltant{t}  /  \LamDeltamv{x}  ]^{ \LamDeltant{A} }  \LamDeltasym{(}   \lambda  \LamDeltamv{y} : \LamDeltant{A'}  .  \LamDeltant{t'}   \LamDeltasym{)}  =  \lambda  \LamDeltamv{y} : \LamDeltant{A'}  .  \LamDeltasym{(}   [  \LamDeltant{t}  /  \LamDeltamv{x}  ]^{ \LamDeltant{A} }  \LamDeltant{t'}   \LamDeltasym{)} $\\

  \item[] $ [  \LamDeltant{t}  /  \LamDeltamv{x}  ]^{ \LamDeltant{A} }  \LamDeltasym{(}  \LamDeltant{t_{{\mathrm{1}}}} \, \LamDeltant{t_{{\mathrm{2}}}}  \LamDeltasym{)}  = \LamDeltasym{(}   [  \LamDeltant{t}  /  \LamDeltamv{x}  ]^{ \LamDeltant{A} }  \LamDeltant{t_{{\mathrm{1}}}}   \LamDeltasym{)} \, \LamDeltasym{(}   [  \LamDeltant{t}  /  \LamDeltamv{x}  ]^{ \LamDeltant{A} }  \LamDeltant{t_{{\mathrm{2}}}}   \LamDeltasym{)}$\\
    \begin{tabular}{lll}
      & Where $\LamDeltasym{(}   [  \LamDeltant{t}  /  \LamDeltamv{x}  ]^{ \LamDeltant{A} }  \LamDeltant{t_{{\mathrm{1}}}}   \LamDeltasym{)}$ is not a $\lambda$-abstraction, or both 
      $\LamDeltasym{(}   [  \LamDeltant{t}  /  \LamDeltamv{x}  ]^{ \LamDeltant{A} }  \LamDeltant{t_{{\mathrm{1}}}}   \LamDeltasym{)}$ \\
      & and  $\LamDeltant{t_{{\mathrm{1}}}}$ are $\lambda$-abstractions.\\
      & \\
    \end{tabular}

  \item[] $ [  \LamDeltant{t}  /  \LamDeltamv{x}  ]^{ \LamDeltant{A} }  \LamDeltasym{(}  \LamDeltant{t_{{\mathrm{1}}}} \, \LamDeltant{t_{{\mathrm{2}}}}  \LamDeltasym{)}  =  [  \LamDeltant{s'_{{\mathrm{2}}}}  /  \LamDeltamv{y}  ]^{ \LamDeltant{A''} }  \LamDeltant{s'_{{\mathrm{1}}}} $\\
    \begin{tabular}{lll}
      & Where $\LamDeltasym{(}   [  \LamDeltant{t}  /  \LamDeltamv{x}  ]^{ \LamDeltant{A} }  \LamDeltant{t_{{\mathrm{1}}}}   \LamDeltasym{)} =  \lambda  \LamDeltamv{y} : \LamDeltant{A''}  .  \LamDeltant{s'_{{\mathrm{1}}}} $ for some $\LamDeltamv{y}$, $\LamDeltant{s'_{{\mathrm{1}}}}$ and $\LamDeltant{A''}$, \\
      & $ [  \LamDeltant{t}  /  \LamDeltamv{x}  ]^{ \LamDeltant{A} }  \LamDeltant{t_{{\mathrm{2}}}}  = \LamDeltant{s'_{{\mathrm{2}}}}$, and $ \textsf{ctype}_{ \LamDeltant{A} }( \LamDeltamv{x} , \LamDeltant{t_{{\mathrm{1}}}} )  =  \LamDeltant{A''}  \to  \LamDeltant{A'} $. \\
      & \\
    \end{tabular}
  \end{itemize}
\end{definition}
The definition of the hereditary substitution function is similar to
the definition of capture-avoiding substitution.  The differences show up
in the cases for application.  The last case of the
hereditary substitution function handles the case when a new
$\beta$-redex is created as a result of substitution. This case
depends heavily on the following lemma. 
\begin{lemma}[Properties of $ \textsf{ctype} $ Continued]
  \label{lemma:ctype_props_cont1}
  If $ \Gamma  \LamDeltasym{,}  \LamDeltamv{x}  \LamDeltasym{:}  \LamDeltant{T}  \LamDeltasym{,}  \Gamma'  \vdash  \LamDeltant{t_{{\mathrm{1}}}} \, \LamDeltant{t_{{\mathrm{2}}}}  :  \LamDeltant{T'} $, $ \Gamma  \vdash  \LamDeltant{t}  :  \LamDeltant{T} $, $ [  \LamDeltant{t}  /  \LamDeltamv{x}  ]^{ \LamDeltant{T} }  \LamDeltant{t_{{\mathrm{1}}}}  =
   \lambda  \LamDeltamv{y} : \LamDeltant{T_{{\mathrm{1}}}}  .  \LamDeltant{t'} $, and $\LamDeltant{t_{{\mathrm{1}}}}$ is not a $\lambda$-abstraction, then
  $\LamDeltant{t_{{\mathrm{1}}}}$ is in head normal form and there exists a type $\LamDeltant{A}$ 
  such that $ \textsf{ctype}_{ \LamDeltant{T} }( \LamDeltamv{x} , \LamDeltant{t_{{\mathrm{1}}}} )  = \LamDeltant{A}$.
\end{lemma}
\begin{proof}
  This can be shown by induction on the structure of $\LamDeltant{t_{{\mathrm{1}}}} \, \LamDeltant{t_{{\mathrm{2}}}}$.  See
  part one of the proof in
  Appendix~\ref{subsec:proof_of_ctype_props_cont}.
\end{proof}
The previous properties state that if we have created a redex using
hereditary substitution, then $ \textsf{ctype} $ must be defined.  This in
turn tells us that in the case where hereditary substitution is applied to
a term of the form $\LamDeltant{t_{{\mathrm{1}}}} \, \LamDeltant{t_{{\mathrm{2}}}}$ and a new $\beta$-redex is created then
the head of $\LamDeltant{t_{{\mathrm{1}}}}$ must be the variable being replaced. Furthermore,
recursively applying the hereditary substitution function to $\LamDeltant{t_{{\mathrm{1}}}}$
must yield a $\lambda$-abstraction.  Hence, the cut type must be an
arrow type.  Now $ \textsf{ctype} $ then tells us that this arrow type must
be either equal or strictly larger than the type of the created
$\lambda$-abstraction.  Thus, we can see that recursively reducing the
application of the results of recursively applying hereditary
substitution to $\LamDeltant{t_{{\mathrm{1}}}}$ and $\LamDeltant{t_{{\mathrm{2}}}}$ terminates based on our
ordering.  This explanation reveals that $ \textsf{ctype} $ is instrumental
in the detection of newly created redexes and in proving properties of
the hereditary substitution function.  

The full normalization proof for STLC using hereditary substitution can be
found in \cite{Eades:2011}. The following example gives
some intuition of how the hereditary substitution function operates.
\begin{example}
  \label{ex:beta_reduction}
  \small
  Consider the terms $\LamDeltant{t} \equiv  \lambda  \LamDeltamv{f} :   \text{b}   \to   \text{b}    .  \LamDeltamv{f} $ and $\LamDeltant{t'} \equiv \LamDeltasym{(}  \LamDeltamv{x} \, \LamDeltasym{(}   \lambda  \LamDeltamv{y} :  \text{b}   .  \LamDeltamv{y}   \LamDeltasym{)}  \LamDeltasym{)} \, \LamDeltamv{z}$,
  where $\LamDeltamv{z}$ is a free variable of type $ \text{b} $.  Our goal is to compute 
  $ [  \LamDeltant{t}  /  \LamDeltamv{x}  ]^{ \LamDeltasym{(}   \LamDeltasym{(}    \text{b}   \to   \text{b}    \LamDeltasym{)}  \to  \LamDeltasym{(}    \text{b}   \to   \text{b}    \LamDeltasym{)}   \LamDeltasym{)} }  \LamDeltant{t'} $ using
  the definition of the hereditary substitution function in Definition~\ref{def:hereditary_substitution_function}.
  First, 
  \begin{center}
    $ [  \LamDeltant{t}  /  \LamDeltamv{x}  ]^{ \LamDeltasym{(}   \LamDeltasym{(}    \text{b}   \to   \text{b}    \LamDeltasym{)}  \to  \LamDeltasym{(}    \text{b}   \to   \text{b}    \LamDeltasym{)}   \LamDeltasym{)} }  \LamDeltasym{(}  \LamDeltamv{x} \, \LamDeltasym{(}   \lambda  \LamDeltamv{y} :  \text{b}   .  \LamDeltamv{y}   \LamDeltasym{)}  \LamDeltasym{)}  =  \lambda  \LamDeltamv{y} :  \text{b}   .  \LamDeltamv{y} $,
  \end{center} 
  because 
  \begin{center}
    \begin{tabular}{lll}
      $ \textsf{ctype}_{ \LamDeltasym{(}   \LamDeltasym{(}    \text{b}   \to   \text{b}    \LamDeltasym{)}  \to  \LamDeltasym{(}    \text{b}   \to   \text{b}    \LamDeltasym{)}   \LamDeltasym{)} }( \LamDeltamv{x} , \LamDeltamv{x} )  =  \LamDeltasym{(}    \text{b}   \to   \text{b}    \LamDeltasym{)}  \to  \LamDeltasym{(}    \text{b}   \to   \text{b}    \LamDeltasym{)} $, \\
      \\
      $ [  \LamDeltant{t}  /  \LamDeltamv{x}  ]^{ \LamDeltasym{(}   \LamDeltasym{(}    \text{b}   \to   \text{b}    \LamDeltasym{)}  \to  \LamDeltasym{(}    \text{b}   \to   \text{b}    \LamDeltasym{)}   \LamDeltasym{)} }  \LamDeltamv{x}  = \LamDeltant{t}$,\\
      \\
      $ [  \LamDeltant{t}  /  \LamDeltamv{x}  ]^{ \LamDeltasym{(}   \LamDeltasym{(}    \text{b}   \to   \text{b}    \LamDeltasym{)}  \to  \LamDeltasym{(}    \text{b}   \to   \text{b}    \LamDeltasym{)}   \LamDeltasym{)} }  \LamDeltasym{(}   \lambda  \LamDeltamv{y} :  \text{b}   .  \LamDeltamv{y}   \LamDeltasym{)}  =  \lambda  \LamDeltamv{y} :  \text{b}   .  \LamDeltamv{y} $,
    \end{tabular}
  \end{center}
  and
  \begin{center}
    $ [  \LamDeltasym{(}   \lambda  \LamDeltamv{y} :  \text{b}   .  \LamDeltamv{y}   \LamDeltasym{)}  /  \LamDeltamv{f}  ]^{ \LamDeltasym{(}    \text{b}   \to   \text{b}    \LamDeltasym{)} }  \LamDeltamv{f}  =  \lambda  \LamDeltamv{y} :  \text{b}   .  \LamDeltamv{y} $.
  \end{center} 
  Now the previous facts give us that 
  \begin{center}
    \begin{math}
       [  \LamDeltant{t}  /  \LamDeltamv{x}  ]^{ \LamDeltasym{(}   \LamDeltasym{(}    \text{b}   \to   \text{b}    \LamDeltasym{)}  \to  \LamDeltasym{(}    \text{b}   \to   \text{b}    \LamDeltasym{)}   \LamDeltasym{)} }  \LamDeltant{t'}  = \LamDeltamv{z},
    \end{math}
  \end{center}
  because $ [  \LamDeltant{t}  /  \LamDeltamv{x}  ]^{ \LamDeltasym{(}   \LamDeltasym{(}    \text{b}   \to   \text{b}    \LamDeltasym{)}  \to  \LamDeltasym{(}    \text{b}   \to   \text{b}    \LamDeltasym{)}   \LamDeltasym{)} }  \LamDeltamv{z}  = \LamDeltamv{z}$, and $ [  \LamDeltamv{z}  /  \LamDeltamv{y}  ]^{  \text{b}  }  \LamDeltamv{y}  = \LamDeltamv{z}$.
\end{example}

%   We would like to make one final remark regarding the definition of the
% hereditary substitution function.  Note that it is defined for even
% untyped terms.  In the case where hereditary substitution is applied
% to an untyped looping term the function will eventually just give up
% and return the original term.  For example, $[\lambda x.xx/y]^{\LamDeltant{A}}
% (y\,(\lambda x.xx)) = (\lambda x.xx)(\lambda x.xx)$ for any type
% $\LamDeltant{A}$.  In these types of cases $\LamDeltant{A}$ is only a termination
% metric. It no longer corresponds to the type of $\lambda x.xx$ only to
% $\LamDeltamv{y}$.  However, when the hereditary substitution function is
% restricted to typeable terms it has a number of nice properties.  We
% will see all of these in
% Sect.~\ref{subsec:properties_of_the_hereditary_substitution_function}.

% section the_hereditary_substitution_function_stlc (end)

%%% Local Variables: 
%%% mode: latex
%%% TeX-master: "paper"
%%% End:

%
\section{Extending The Hereditary Substitution Function to the $\lambda\Delta$-Calculus}
\label{sec:the_hereditary_substitution_function_for_the_ld-calculus}
Since the $\lambda\Delta$-calculus is an extension of STLC, we might expect that 
the hereditary substitution function for the $\lambda\Delta$-calculus is also an extension of the
hereditary substitution function for STLC.  In this section we show that this extension is
non-trivial by first considering the naive extension, and then discussing why it does not work.
Following this, we give the final extension and prove it correct.

\subsection{Problems with a Naive Extension}
\label{subsec:the_naive_extension}
Lets consider the definition of the hereditary substitution function
for STLC extended with two new cases. The first case for the $\Delta$-abstraction
whose definition parallels the definition for the
$\lambda$-abstraction.  The second is a new application case which
handles newly created structural redexes and is defined following the
same pattern as the case which handles $\beta$-redexes.  We use the same termination
metric we previously used. 
\begin{definition}
  \label{def:hereditary_substitution_function}
  The naive hereditary substitution function is defined as follows:
  \small
  \begin{itemize}
  \item[] $ [  \LamDeltant{t}  /  \LamDeltamv{x}  ]^{ \LamDeltant{A} }  \LamDeltamv{x}  = \LamDeltant{t}$
  \item[] $ [  \LamDeltant{t}  /  \LamDeltamv{x}  ]^{ \LamDeltant{A} }  \LamDeltamv{y}  = \LamDeltamv{y}$\\

  \item[] $ [  \LamDeltant{t}  /  \LamDeltamv{x}  ]^{ \LamDeltant{A} }  \LamDeltasym{(}   \lambda  \LamDeltamv{y} : \LamDeltant{A'}  .  \LamDeltant{t'}   \LamDeltasym{)}  =  \lambda  \LamDeltamv{y} : \LamDeltant{A'}  .  \LamDeltasym{(}   [  \LamDeltant{t}  /  \LamDeltamv{x}  ]^{ \LamDeltant{A} }  \LamDeltant{t'}   \LamDeltasym{)} $
  \item[] $ [  \LamDeltant{t}  /  \LamDeltamv{x}  ]^{ \LamDeltant{A} }  \LamDeltasym{(}   \Delta  \LamDeltamv{y}  :  \LamDeltant{A'}  .  \LamDeltant{t'}   \LamDeltasym{)}  =  \Delta  \LamDeltamv{y}  :  \LamDeltant{A'}  .  \LamDeltasym{(}   [  \LamDeltant{t}  /  \LamDeltamv{x}  ]^{ \LamDeltant{A} }  \LamDeltant{t'}   \LamDeltasym{)} $\\

  \item[] $ [  \LamDeltant{t}  /  \LamDeltamv{x}  ]^{ \LamDeltant{A} }  \LamDeltasym{(}  \LamDeltant{t_{{\mathrm{1}}}} \, \LamDeltant{t_{{\mathrm{2}}}}  \LamDeltasym{)}  = \LamDeltasym{(}   [  \LamDeltant{t}  /  \LamDeltamv{x}  ]^{ \LamDeltant{A} }  \LamDeltant{t_{{\mathrm{1}}}}   \LamDeltasym{)} \, \LamDeltasym{(}   [  \LamDeltant{t}  /  \LamDeltamv{x}  ]^{ \LamDeltant{A} }  \LamDeltant{t_{{\mathrm{2}}}}   \LamDeltasym{)}$\\
    \begin{tabular}{lll}
      & Where $\LamDeltasym{(}   [  \LamDeltant{t}  /  \LamDeltamv{x}  ]^{ \LamDeltant{A} }  \LamDeltant{t_{{\mathrm{1}}}}   \LamDeltasym{)}$ is not a $\lambda$-abstraction or $\Delta$-abstraction,  or both $\LamDeltasym{(}   [  \LamDeltant{t}  /  \LamDeltamv{x}  ]^{ \LamDeltant{A} }  \LamDeltant{t_{{\mathrm{1}}}}   \LamDeltasym{)}$ \\
      & and  $\LamDeltant{t_{{\mathrm{1}}}}$ are $\lambda$-abstractions or $\Delta$-abstractions.\\
      & \\
    \end{tabular}

  \item[] $ [  \LamDeltant{t}  /  \LamDeltamv{x}  ]^{ \LamDeltant{A} }  \LamDeltasym{(}  \LamDeltant{t_{{\mathrm{1}}}} \, \LamDeltant{t_{{\mathrm{2}}}}  \LamDeltasym{)}  =  [  \LamDeltant{s'_{{\mathrm{2}}}}  /  \LamDeltamv{y}  ]^{ \LamDeltant{A''} }  \LamDeltant{s'_{{\mathrm{1}}}} $\\
    \begin{tabular}{lll}
      & Where $\LamDeltasym{(}   [  \LamDeltant{t}  /  \LamDeltamv{x}  ]^{ \LamDeltant{A} }  \LamDeltant{t_{{\mathrm{1}}}}   \LamDeltasym{)} =  \lambda  \LamDeltamv{y} : \LamDeltant{A''}  .  \LamDeltant{s'_{{\mathrm{1}}}} $ for some $\LamDeltamv{y}$, $\LamDeltant{s'_{{\mathrm{1}}}}$ and $\LamDeltant{A''}$, \\
      & $ [  \LamDeltant{t}  /  \LamDeltamv{x}  ]^{ \LamDeltant{A} }  \LamDeltant{t_{{\mathrm{2}}}}  = \LamDeltant{s'_{{\mathrm{2}}}}$, and $ \textsf{ctype}_{ \LamDeltant{A} }( \LamDeltamv{x} , \LamDeltant{t_{{\mathrm{1}}}} )  =  \LamDeltant{A''}  \to  \LamDeltant{A'} $. \\
      & \\
    \end{tabular}

  \item[] $ [  \LamDeltant{t}  /  \LamDeltamv{x}  ]^{ \LamDeltant{A} }  \LamDeltasym{(}  \LamDeltant{t_{{\mathrm{1}}}} \, \LamDeltant{t_{{\mathrm{2}}}}  \LamDeltasym{)}  = \Delta z: \neg  \LamDeltant{A'} .[  \lambda  \LamDeltamv{y} :  \LamDeltant{A''}  \to  \LamDeltant{A'}   .  \LamDeltasym{(}  \LamDeltamv{z} \, \LamDeltasym{(}  \LamDeltamv{y} \, \LamDeltant{s_{{\mathrm{2}}}}  \LamDeltasym{)}  \LamDeltasym{)} /\LamDeltamv{y} ]^{ \neg  \LamDeltasym{(}   \LamDeltant{A''}  \to  \LamDeltant{A'}   \LamDeltasym{)} } \LamDeltamv{s}$\\
    \begin{tabular}{lll}
      & Where $\LamDeltasym{(}   [  \LamDeltant{t}  /  \LamDeltamv{x}  ]^{ \LamDeltant{A} }  \LamDeltant{t_{{\mathrm{1}}}}   \LamDeltasym{)} =  \Delta  \LamDeltamv{y}  :   \neg  \LamDeltasym{(}   \LamDeltant{A''}  \to  \LamDeltant{A'}   \LamDeltasym{)}   .  \LamDeltant{s} $ for some, $\LamDeltamv{y}$ $\LamDeltamv{s}$, 
      and $ \LamDeltant{A''}  \to  \LamDeltant{A'} $, \\      
      & $\LamDeltasym{(}   [  \LamDeltant{t}  /  \LamDeltamv{x}  ]^{ \LamDeltant{A} }  \LamDeltant{t_{{\mathrm{2}}}}   \LamDeltasym{)} = \LamDeltant{s_{{\mathrm{2}}}}$ for some $\LamDeltant{s_{{\mathrm{2}}}}$, $ \textsf{ctype}_{ \LamDeltant{A} }( \LamDeltamv{x} , \LamDeltant{t_{{\mathrm{1}}}} )  =  \LamDeltant{A''}  \to  \LamDeltant{A'} $, and $\LamDeltamv{z}$ is completely fresh.\\
      & \\
    \end{tabular}  
  \end{itemize}
\end{definition}
There is one glaring issue with this definition and it lies in the final case.  
We know from Lemma~\ref{lemma:ctype_props} and Lemma~\ref{lemma:ctype_props_cont} 
that $ \textsf{ctype}_{ \LamDeltant{A} }( \LamDeltamv{x} , \LamDeltant{t_{{\mathrm{1}}}} )  =  \LamDeltant{A''}  \to  \LamDeltant{A'} $ 
implies that $\LamDeltant{A} \geq  \LamDeltant{A''}  \to  \LamDeltant{A'}  <  \neg  \LamDeltasym{(}   \LamDeltant{A''}  \to  \LamDeltant{A'}   \LamDeltasym{)} $. Thus, this 
definition is not well founded!  To fix this issue instead of naively following 
the structural reduction rule we immediately simultaneously hereditarily reduce 
all redexes created by replacing $\LamDeltamv{y}$ with the linear $\lambda$-abstraction
$ \lambda  \LamDeltamv{y} :  \LamDeltant{A''}  \to  \LamDeltant{A'}   .  \LamDeltasym{(}  \LamDeltamv{z} \, \LamDeltasym{(}  \LamDeltamv{y} \, \LamDeltant{s_{{\mathrm{2}}}}  \LamDeltasym{)}  \LamDeltasym{)} $.  To accomplish this we will define mutually with
the hereditary substitution function a new function called the hereditary
structural substitution function.
% subsection the_naive_extension (end)

\subsection{A Correct Extension of Hereditary Substitution}
\label{subsec:the_final_extension}
In order to reduce structural redexes in the definition of the
hereditary substitution we will define by induction mutually with
the hereditary substitution function  a function called the
hereditary structural substitution function.  This function will
use the notion of a multi-substitution.  These are
given by the following grammar:
\begin{center}
  \begin{math}
    \Theta ::=  \cdot \,|\,\Theta  \LamDeltasym{,}   ( \LamDeltamv{y} , \LamDeltamv{z} , \LamDeltant{t} ) 
  \end{math}
\end{center}
We denote the hereditary structural substitution function by $ \langle  \Theta  \rangle^{ \LamDeltant{A} }_{ \LamDeltant{A'} }  \LamDeltant{t'} $ and hereditary substitution by $ [  \LamDeltant{t}  /  \LamDeltamv{x}  ]^{ \LamDeltant{A} }  \LamDeltant{t'} $.  The
type of all the first projections of the elements of $\Theta$ is
$ \neg  \LamDeltasym{(}   \LamDeltant{A}  \to  \LamDeltant{A'}   \LamDeltasym{)} $ and the type of the second projections is $ \neg  \LamDeltant{A'} $.
Both functions are defined by mutual induction using the metric
$(\LamDeltant{A},f,\LamDeltant{t'})$, where $f \in \{0,1\}$, in lexicographic
combination with the ordering on types, the natural number ordering,
and the strict subexpression on terms.  The meta-variable $f$ labels
each function and is equal to $0$ in the definition of the hereditary
substitution function and is equal to $1$ in the definition of the
hereditary structural substitution function.  Again, in the definitions of
the hereditary substitution and hereditary structural substitution
function it is assumed that all variables have been renamed as to
prevent variable capture.  The following is the final definition of the
hereditary substitution function for the $\lambda\Delta$-calculus.

\begin{definition}
  \label{def:hereditary_substitution_function}
  The hereditary substitution function is defined as follows:
  \small
  \begin{itemize}
  \item[] $ \langle  \Theta  \rangle^{ \LamDeltant{A_{{\mathrm{1}}}} }_{ \LamDeltant{A_{{\mathrm{2}}}} }  \LamDeltamv{x}  =  \lambda  \LamDeltamv{y} :  \LamDeltant{A_{{\mathrm{1}}}}  \to  \LamDeltant{A_{{\mathrm{2}}}}   .  \LamDeltasym{(}  \LamDeltamv{z} \, \LamDeltasym{(}  \LamDeltamv{y} \, \LamDeltant{t}  \LamDeltasym{)}  \LamDeltasym{)} $\\
    \begin{tabular}{lll}
      & Where $ ( \LamDeltamv{x} , \LamDeltamv{z} , \LamDeltant{t} )  \in \Theta$, for some $\LamDeltamv{z}$ and $\LamDeltant{t}$, and $\LamDeltamv{y}$ is fresh in $\LamDeltamv{x}$, $\LamDeltamv{z}$, and $\LamDeltant{t}$.\\
      & \\
    \end{tabular}
  \item[] $ \langle  \Theta  \rangle^{ \LamDeltant{A_{{\mathrm{1}}}} }_{ \LamDeltant{A_{{\mathrm{2}}}} }  \LamDeltamv{x}  = x$\\
    \begin{tabular}{lll}
      & Where $ ( \LamDeltamv{x} , \LamDeltamv{z} , \LamDeltant{t} )  \not\in \Theta$ for any $\LamDeltamv{z}$ or $\LamDeltant{t}$.\\
      & \\
    \end{tabular}
  \item[] $ \langle  \Theta  \rangle^{ \LamDeltant{A_{{\mathrm{1}}}} }_{ \LamDeltant{A_{{\mathrm{2}}}} }  \LamDeltasym{(}   \lambda  \LamDeltamv{y} : \LamDeltant{A}  .  \LamDeltant{t}   \LamDeltasym{)}  =  \lambda  \LamDeltamv{y} : \LamDeltant{A}  .   \langle  \Theta  \rangle^{ \LamDeltant{A_{{\mathrm{1}}}} }_{ \LamDeltant{A_{{\mathrm{2}}}} }  \LamDeltant{t}  $\\
  \item[] $ \langle  \Theta  \rangle^{ \LamDeltant{A_{{\mathrm{1}}}} }_{ \LamDeltant{A_{{\mathrm{2}}}} }  \LamDeltasym{(}   \Delta  \LamDeltamv{y}  :  \LamDeltant{A}  .  \LamDeltant{t}   \LamDeltasym{)}  =  \Delta  \LamDeltamv{y}  :  \LamDeltant{A}  .   \langle  \Theta  \rangle^{ \LamDeltant{A_{{\mathrm{1}}}} }_{ \LamDeltant{A_{{\mathrm{2}}}} }  \LamDeltant{t}  $\\

  \item[] $ \langle  \Theta  \rangle^{ \LamDeltant{A_{{\mathrm{1}}}} }_{ \LamDeltant{A_{{\mathrm{2}}}} }  \LamDeltasym{(}  \LamDeltamv{x} \, \LamDeltant{t'}  \LamDeltasym{)}  = \LamDeltamv{z} \,  [  \LamDeltant{t}  /  \LamDeltamv{y}  ]^{ \LamDeltant{A_{{\mathrm{1}}}} }  \LamDeltant{s} $\\
    \begin{tabular}{lll}
      & Where $ ( \LamDeltamv{x} , \LamDeltamv{z} , \LamDeltant{t} )  \in \Theta$, $t' \equiv  \lambda  \LamDeltamv{y} : \LamDeltant{A_{{\mathrm{1}}}}  .  \LamDeltant{t''} $, for some $\LamDeltamv{y}$ and
      $\LamDeltant{t''}$, and $ \langle  \Theta  \rangle^{ \LamDeltant{A_{{\mathrm{1}}}} }_{ \LamDeltant{A_{{\mathrm{2}}}} }  \LamDeltant{t''}  = \LamDeltamv{s}$.\\
      & \\
    \end{tabular}
  \item[] $ \langle  \Theta  \rangle^{ \LamDeltant{A_{{\mathrm{1}}}} }_{ \LamDeltant{A_{{\mathrm{2}}}} }  \LamDeltasym{(}  \LamDeltamv{x} \, \LamDeltant{t'}  \LamDeltasym{)}  = \LamDeltamv{z} \, \LamDeltasym{(}   \Delta  \LamDeltamv{z_{{\mathrm{2}}}}  :   \neg  \LamDeltant{A_{{\mathrm{2}}}}   .  \LamDeltant{s}   \LamDeltasym{)}$\\
    \begin{tabular}{lll}
      & Where $ ( \LamDeltamv{x} , \LamDeltamv{z} , \LamDeltant{t} )  \in \Theta$, $t' \equiv  \Delta  \LamDeltamv{y}  :   \neg  \LamDeltasym{(}   \LamDeltant{A_{{\mathrm{1}}}}  \to  \LamDeltant{A_{{\mathrm{2}}}}   \LamDeltasym{)}   .  \LamDeltant{t''} $, for some $\LamDeltamv{y}$ and $\LamDeltant{t''}$, and \\
      & $ \langle  \Theta  \LamDeltasym{,}   ( \LamDeltamv{y} , \LamDeltamv{z_{{\mathrm{2}}}} , \LamDeltant{t} )   \rangle^{ \LamDeltant{A_{{\mathrm{1}}}} }_{ \LamDeltant{A_{{\mathrm{2}}}} }  \LamDeltant{t''}  = \LamDeltamv{s}$, for some fresh $\LamDeltamv{z_{{\mathrm{2}}}}$.\\
      & \\
    \end{tabular}
  \item[] $ \langle  \Theta  \rangle^{ \LamDeltant{A_{{\mathrm{1}}}} }_{ \LamDeltant{A_{{\mathrm{2}}}} }  \LamDeltasym{(}  \LamDeltamv{x} \, \LamDeltant{t'}  \LamDeltasym{)}  = \LamDeltamv{z} \, \LamDeltant{s'}$\\
    \begin{tabular}{lll}
      & Where $ ( \LamDeltamv{x} , \LamDeltamv{z} , \LamDeltant{t} )  \in \Theta$, $\LamDeltant{t'}$ is not an abstraction, and $ \langle  \Theta  \rangle^{ \LamDeltant{A_{{\mathrm{1}}}} }_{ \LamDeltant{A_{{\mathrm{2}}}} }  \LamDeltant{t'}  = \LamDeltant{s'}$.\\
      & \\
    \end{tabular}
  \item[] $ \langle  \Theta  \rangle^{ \LamDeltant{A_{{\mathrm{1}}}} }_{ \LamDeltant{A_{{\mathrm{2}}}} }  \LamDeltasym{(}  \LamDeltant{t_{{\mathrm{1}}}} \, \LamDeltant{t_{{\mathrm{2}}}}  \LamDeltasym{)}  = \LamDeltant{s_{{\mathrm{1}}}} \, \LamDeltant{s_{{\mathrm{2}}}}$\\
    \begin{tabular}{lll}
      & Where $\LamDeltant{t_{{\mathrm{1}}}}$ is either not a variable, or it is both a variable
      and $(\LamDeltant{t_{{\mathrm{1}}}},\LamDeltamv{z'},\LamDeltant{t'}) \not \in \Theta$ for any \\
      & $\LamDeltant{t'}$ and $\LamDeltamv{z'}$, $ \langle  \Theta  \rangle^{ \LamDeltant{A_{{\mathrm{1}}}} }_{ \LamDeltant{A_{{\mathrm{2}}}} }  \LamDeltant{t_{{\mathrm{1}}}}  = \LamDeltant{s_{{\mathrm{1}}}}$, and $ \langle  \Theta  \rangle^{ \LamDeltant{A_{{\mathrm{1}}}} }_{ \LamDeltant{A_{{\mathrm{2}}}} }  \LamDeltant{t_{{\mathrm{2}}}}  = \LamDeltant{s_{{\mathrm{2}}}}$.\\
      & \\
    \end{tabular}

  \item[] $ [  \LamDeltant{t}  /  \LamDeltamv{x}  ]^{ \LamDeltant{A} }  \LamDeltamv{x}  = \LamDeltant{t}$
  \item[] $ [  \LamDeltant{t}  /  \LamDeltamv{x}  ]^{ \LamDeltant{A} }  \LamDeltamv{y}  = \LamDeltamv{y}$\\

  \item[] $ [  \LamDeltant{t}  /  \LamDeltamv{x}  ]^{ \LamDeltant{A} }  \LamDeltasym{(}   \lambda  \LamDeltamv{y} : \LamDeltant{A'}  .  \LamDeltant{t'}   \LamDeltasym{)}  =  \lambda  \LamDeltamv{y} : \LamDeltant{A'}  .  \LamDeltasym{(}   [  \LamDeltant{t}  /  \LamDeltamv{x}  ]^{ \LamDeltant{A} }  \LamDeltant{t'}   \LamDeltasym{)} $
  \item[] $ [  \LamDeltant{t}  /  \LamDeltamv{x}  ]^{ \LamDeltant{A} }  \LamDeltasym{(}   \Delta  \LamDeltamv{y}  :  \LamDeltant{A'}  .  \LamDeltant{t'}   \LamDeltasym{)}  =  \Delta  \LamDeltamv{y}  :  \LamDeltant{A'}  .  \LamDeltasym{(}   [  \LamDeltant{t}  /  \LamDeltamv{x}  ]^{ \LamDeltant{A} }  \LamDeltant{t'}   \LamDeltasym{)} $\\

  \item[] $ [  \LamDeltant{t}  /  \LamDeltamv{x}  ]^{ \LamDeltant{A} }  \LamDeltasym{(}  \LamDeltant{t_{{\mathrm{1}}}} \, \LamDeltant{t_{{\mathrm{2}}}}  \LamDeltasym{)}  = \LamDeltasym{(}   [  \LamDeltant{t}  /  \LamDeltamv{x}  ]^{ \LamDeltant{A} }  \LamDeltant{t_{{\mathrm{1}}}}   \LamDeltasym{)} \, \LamDeltasym{(}   [  \LamDeltant{t}  /  \LamDeltamv{x}  ]^{ \LamDeltant{A} }  \LamDeltant{t_{{\mathrm{2}}}}   \LamDeltasym{)}$\\
    \begin{tabular}{lll}
      & Where $\LamDeltasym{(}   [  \LamDeltant{t}  /  \LamDeltamv{x}  ]^{ \LamDeltant{A} }  \LamDeltant{t_{{\mathrm{1}}}}   \LamDeltasym{)}$ is not a $\lambda$-abstraction or $\Delta$-abstraction,  or both $\LamDeltasym{(}   [  \LamDeltant{t}  /  \LamDeltamv{x}  ]^{ \LamDeltant{A} }  \LamDeltant{t_{{\mathrm{1}}}}   \LamDeltasym{)}$ \\
      & and  $\LamDeltant{t_{{\mathrm{1}}}}$ are $\lambda$-abstractions or $\Delta$-abstractions.\\
      & \\
    \end{tabular}

  \item[] $ [  \LamDeltant{t}  /  \LamDeltamv{x}  ]^{ \LamDeltant{A} }  \LamDeltasym{(}  \LamDeltant{t_{{\mathrm{1}}}} \, \LamDeltant{t_{{\mathrm{2}}}}  \LamDeltasym{)}  =  [  \LamDeltant{s'_{{\mathrm{2}}}}  /  \LamDeltamv{y}  ]^{ \LamDeltant{A''} }  \LamDeltant{s'_{{\mathrm{1}}}} $\\
    \begin{tabular}{lll}
      & Where $\LamDeltasym{(}   [  \LamDeltant{t}  /  \LamDeltamv{x}  ]^{ \LamDeltant{A} }  \LamDeltant{t_{{\mathrm{1}}}}   \LamDeltasym{)} =  \lambda  \LamDeltamv{y} : \LamDeltant{A''}  .  \LamDeltant{s'_{{\mathrm{1}}}} $ for some $\LamDeltamv{y}$, $\LamDeltant{s'_{{\mathrm{1}}}}$ and $\LamDeltant{A''}$, \\
      & $ [  \LamDeltant{t}  /  \LamDeltamv{x}  ]^{ \LamDeltant{A} }  \LamDeltant{t_{{\mathrm{2}}}}  = \LamDeltant{s'_{{\mathrm{2}}}}$, and $ \textsf{ctype}_{ \LamDeltant{A} }( \LamDeltamv{x} , \LamDeltant{t_{{\mathrm{1}}}} )  =  \LamDeltant{A''}  \to  \LamDeltant{A'} $. \\
      & \\
    \end{tabular}

  \item[] $ [  \LamDeltant{t}  /  \LamDeltamv{x}  ]^{ \LamDeltant{A} }  \LamDeltasym{(}  \LamDeltant{t_{{\mathrm{1}}}} \, \LamDeltant{t_{{\mathrm{2}}}}  \LamDeltasym{)}  = \Delta z: \neg  \LamDeltant{A'} . \langle   ( \LamDeltamv{y} , \LamDeltamv{z} , \LamDeltant{s_{{\mathrm{2}}}} )   \rangle^{ \LamDeltant{A''} }_{ \LamDeltant{A'} }  \LamDeltant{s} $\\
    \begin{tabular}{lll}
      & Where $\LamDeltasym{(}   [  \LamDeltant{t}  /  \LamDeltamv{x}  ]^{ \LamDeltant{A} }  \LamDeltant{t_{{\mathrm{1}}}}   \LamDeltasym{)} =  \Delta  \LamDeltamv{y}  :   \neg  \LamDeltasym{(}   \LamDeltant{A''}  \to  \LamDeltant{A'}   \LamDeltasym{)}   .  \LamDeltant{s} $ for some $\LamDeltamv{y}$ $\LamDeltamv{s}$, 
      and $ \LamDeltant{A''}  \to  \LamDeltant{A'} $, \\      
      & $\LamDeltasym{(}   [  \LamDeltant{t}  /  \LamDeltamv{x}  ]^{ \LamDeltant{A} }  \LamDeltant{t_{{\mathrm{2}}}}   \LamDeltasym{)} = \LamDeltant{s_{{\mathrm{2}}}}$ for some $\LamDeltant{s_{{\mathrm{2}}}}$, $ \textsf{ctype}_{ \LamDeltant{A} }( \LamDeltamv{x} , \LamDeltant{t_{{\mathrm{1}}}} )  =  \LamDeltant{A''}  \to  \LamDeltant{A'} $, and $\LamDeltamv{z}$ is fresh.\\
      & \\
    \end{tabular}  
  \end{itemize}
\end{definition}
We can see in the final case of the hereditary substitution function that the cut type has decreased.  Hence, this
case is now well founded. Lets consider an example which illustrates how our new definition operates. 
\begin{example}
  \label{ex:struct_reduction}
  \small
  Consider the terms $\LamDeltant{t} \equiv  \Delta  \LamDeltamv{f}  :   \neg  \LamDeltasym{(}    \text{b}   \to   \text{b}    \LamDeltasym{)}   .  \LamDeltasym{(}  \LamDeltamv{f} \, \LamDeltasym{(}   \Delta  \LamDeltamv{f'}  :   \neg  \LamDeltasym{(}    \text{b}   \to   \text{b}    \LamDeltasym{)}   .  \LamDeltasym{(}  \LamDeltamv{f'} \, \LamDeltasym{(}   \lambda  \LamDeltamv{z} :  \text{b}   .  \LamDeltamv{z}   \LamDeltasym{)}  \LamDeltasym{)}   \LamDeltasym{)}  \LamDeltasym{)} $ and $\LamDeltant{t'} \equiv \LamDeltamv{x} \, \LamDeltamv{u}$,
  where $\LamDeltamv{u}$ is a free variable of type $ \text{b} $.  Again, our goal is to compute $ [  \LamDeltant{t}  /  \LamDeltamv{x}  ]^{ \LamDeltasym{(}    \text{b}   \to   \text{b}    \LamDeltasym{)} }  \LamDeltant{t'} $ using
  the definition of the hereditary substitution function in Definition~\ref{def:hereditary_substitution_function}.
  Now
  \begin{center}
    $ [  \LamDeltant{t}  /  \LamDeltamv{x}  ]^{ \LamDeltasym{(}    \text{b}   \to   \text{b}    \LamDeltasym{)} }  \LamDeltasym{(}  \LamDeltamv{x} \, \LamDeltamv{u}  \LamDeltasym{)}  =  \Delta  \LamDeltamv{z_{{\mathrm{1}}}}  :   \neg   \text{b}    .  \LamDeltasym{(}  \LamDeltamv{z_{{\mathrm{1}}}} \, \LamDeltasym{(}   \Delta  \LamDeltamv{z_{{\mathrm{2}}}}  :   \neg   \text{b}    .  \LamDeltasym{(}  \LamDeltamv{z_{{\mathrm{2}}}} \, \LamDeltamv{u}  \LamDeltasym{)}   \LamDeltasym{)}  \LamDeltasym{)} $,
  \end{center} 
  because 
  \begin{center}
    \begin{tabular}{llllll}
      $ \textsf{ctype}_{ \LamDeltasym{(}    \text{b}   \to   \text{b}    \LamDeltasym{)} }( \LamDeltamv{x} , \LamDeltamv{x} )  = \LamDeltasym{(}    \text{b}   \to   \text{b}    \LamDeltasym{)}$, &
      &
      $ [  \LamDeltant{t}  /  \LamDeltamv{x}  ]^{ \LamDeltasym{(}    \text{b}   \to   \text{b}    \LamDeltasym{)} }  \LamDeltamv{x}  = \LamDeltant{t}$,
      &
      $ [  \LamDeltant{t}  /  \LamDeltamv{x}  ]^{ \LamDeltasym{(}    \text{b}   \to   \text{b}    \LamDeltasym{)} }  \LamDeltamv{u}  = \LamDeltamv{u}$,
    \end{tabular}
  \end{center}
  and for some fresh variable $\LamDeltamv{z_{{\mathrm{1}}}}$ of type $ \neg   \text{b}  $
  \begin{center}
    \begin{math}
      \begin{array}{lll}
         \Delta  \LamDeltamv{z_{{\mathrm{1}}}}  :   \neg   \text{b}    .   \langle   ( \LamDeltamv{f} , \LamDeltamv{z_{{\mathrm{1}}}} , \LamDeltamv{u} )   \rangle^{  \text{b}  }_{  \text{b}  }  \LamDeltasym{(}  \LamDeltamv{f} \, \LamDeltasym{(}   \Delta  \LamDeltamv{f'}  :   \neg  \LamDeltasym{(}    \text{b}   \to   \text{b}    \LamDeltasym{)}   .  \LamDeltasym{(}  \LamDeltamv{f'} \, \LamDeltasym{(}   \lambda  \LamDeltamv{z} :  \text{b}   .  \LamDeltamv{z}   \LamDeltasym{)}  \LamDeltasym{)}   \LamDeltasym{)}  \LamDeltasym{)}    & = & \\
         \Delta  \LamDeltamv{z_{{\mathrm{1}}}}  :   \neg   \text{b}    .  \LamDeltasym{(}  \LamDeltamv{z_{{\mathrm{1}}}} \, \LamDeltasym{(}   \Delta  \LamDeltamv{z_{{\mathrm{2}}}}  :   \neg   \text{b}    .  \LamDeltasym{(}  \LamDeltamv{z_{{\mathrm{2}}}} \, \LamDeltamv{u}  \LamDeltasym{)}   \LamDeltasym{)}  \LamDeltasym{)} \\
      \end{array}
    \end{math}
  \end{center}
  where
  \begin{center}
    \begin{math}
      \begin{array}{lll}
         \langle   ( \LamDeltamv{f} , \LamDeltamv{z_{{\mathrm{1}}}} , \LamDeltamv{u} )   \rangle^{  \text{b}  }_{  \text{b}  }  \LamDeltasym{(}  \LamDeltamv{f} \, \LamDeltasym{(}   \Delta  \LamDeltamv{f'}  :   \neg  \LamDeltasym{(}    \text{b}   \to   \text{b}    \LamDeltasym{)}   .  \LamDeltasym{(}  \LamDeltamv{f'} \, \LamDeltasym{(}   \lambda  \LamDeltamv{z} :  \text{b}   .  \LamDeltamv{z}   \LamDeltasym{)}  \LamDeltasym{)}   \LamDeltasym{)}  \LamDeltasym{)}   & = & \\
        \LamDeltamv{z_{{\mathrm{1}}}} \, \LamDeltasym{(}   \Delta  \LamDeltamv{z_{{\mathrm{2}}}}  :   \neg   \text{b}    .   \langle   ( \LamDeltamv{f} , \LamDeltamv{z_{{\mathrm{1}}}} , \LamDeltamv{u} )   \LamDeltasym{,}   ( \LamDeltamv{f'} , \LamDeltamv{z_{{\mathrm{2}}}} , \LamDeltamv{u} )   \rangle^{  \text{b}  }_{  \text{b}  }  \LamDeltasym{(}  \LamDeltamv{f'} \, \LamDeltasym{(}   \lambda  \LamDeltamv{z} :  \text{b}   .  \LamDeltamv{z}   \LamDeltasym{)}  \LamDeltasym{)}    \LamDeltasym{)}
      \end{array}
    \end{math}
    \end{center}
  because
  \begin{center}
    \begin{math}
      \begin{array}{lll}
         ( \LamDeltamv{f} , \LamDeltamv{z_{{\mathrm{1}}}} , \LamDeltamv{u} )  \in \langle ( \LamDeltamv{f} , \LamDeltamv{z_{{\mathrm{1}}}} , \LamDeltamv{u} ) \rangle,  \Delta  \LamDeltamv{f'}  :   \neg  \LamDeltasym{(}    \text{b}   \to   \text{b}    \LamDeltasym{)}   .  \LamDeltasym{(}  \LamDeltamv{f'} \, \LamDeltasym{(}   \lambda  \LamDeltamv{z} :  \text{b}   .  \LamDeltamv{z}   \LamDeltasym{)}  \LamDeltasym{)}  \equiv  \Delta  \LamDeltamv{f'}  :   \neg  \LamDeltasym{(}    \text{b}   \to   \text{b}    \LamDeltasym{)}   .  \LamDeltasym{(}  \LamDeltamv{f'} \, \LamDeltasym{(}   \lambda  \LamDeltamv{z} :  \text{b}   .  \LamDeltamv{z}   \LamDeltasym{)}  \LamDeltasym{)} ,
      \end{array}
    \end{math}
  \end{center}
    and for some fresh variable $\LamDeltamv{z_{{\mathrm{2}}}}$ of type $ \neg   \text{b}  $
    \begin{center}
    \begin{math}
      \begin{array}{lll}
         \langle   ( \LamDeltamv{f} , \LamDeltamv{z_{{\mathrm{1}}}} , \LamDeltamv{u} )   \LamDeltasym{,}   ( \LamDeltamv{f'} , \LamDeltamv{z_{{\mathrm{2}}}} , \LamDeltamv{u} )   \rangle^{  \text{b}  }_{  \text{b}  }  \LamDeltasym{(}  \LamDeltamv{f'} \, \LamDeltasym{(}   \lambda  \LamDeltamv{z} :  \text{b}   .  \LamDeltamv{z}   \LamDeltasym{)}  \LamDeltasym{)}   & = & \LamDeltamv{z_{{\mathrm{2}}}} \, \LamDeltamv{u}\\
      \end{array}
    \end{math}
  \end{center}
  because
  \begin{center}
    \begin{math}
      \begin{array}{llllll}
         ( \LamDeltamv{f'} , \LamDeltamv{z_{{\mathrm{2}}}} , \LamDeltamv{u} )  \in \langle ( \LamDeltamv{f} , \LamDeltamv{z_{{\mathrm{1}}}} , \LamDeltamv{u} )   \LamDeltasym{,}   ( \LamDeltamv{f'} , \LamDeltamv{z_{{\mathrm{2}}}} , \LamDeltamv{u} ) \rangle, &
        &
         \lambda  \LamDeltamv{z} :  \text{b}   .  \LamDeltamv{z}  \equiv  \lambda  \LamDeltamv{z} :  \text{b}   .  \LamDeltamv{z} ,
        &
         \langle   ( \LamDeltamv{f} , \LamDeltamv{z_{{\mathrm{1}}}} , \LamDeltamv{u} )   \rangle^{  \text{b}  }_{  \text{b}  }  \LamDeltamv{z}  = z
      \end{array}
    \end{math}
  \end{center}
\end{example}
In the next section we prove the definition of the hereditary substitution function correct.
% subsection the_final_extension (end)

\subsection{Properties of the Hereditary Substitution Function}
\label{subsec:properties_of_the_hereditary_substitution_function}
There are two main ways the hereditary substitution function is used.
It either replaces capture-avoiding substitution in ones' type theory
or it is used in some other way.  For example, in Canonical LF
hereditary substitution replaces capture-avoiding substitution
\cite{Watkins:2004,Adams:2004}.  However, in \cite{Abel:2008} it is
used only as a normalization function.  No matter how it is used there
are three correctness results which must be proven.  These are
totality, type preservation, and normality preservation.  There is an
additional correctness property we feel one must prove when hereditary
substitution is used as a normalization function.  This property is
called soundness with respect to reduction.  It shows that hereditary
substitution does nothing more than what capture-avoiding substitution
and $\beta$-reduction can do.

We introduce some notation to make working with multi-substitutions a
bit easier.  The sets of all first, second, and third projections of
the triples in $\Theta$ are denoted $ \Theta ^1 $, $ \Theta ^2 $, and $ \Theta ^3 $ respectively.  We denote the assumption of all elements of
$\Theta^i$ having the type $\LamDeltant{T}$ as $\Theta^i : \LamDeltant{T}$. This latter
notation is used in typing contexts to indicate the addition of all the
variables in $\Theta^j$ for $j \in \{1,2\}$ to the context with the
specified type.  We denote this as $\Gamma,\Theta^j : \LamDeltant{T},\Gamma'$ for
some contexts $\Gamma$ and $\Gamma'$.  The notation $ \Gamma  \vdash   \Theta ^3   :  \LamDeltant{T} $
is defined as for all $\LamDeltant{t} \in  \Theta ^3 $ the typing judgment $ \Gamma  \vdash  \LamDeltant{t}  :  \LamDeltant{T} $ holds.  Finally, we denote terms in $ \Theta ^3 $ being normal
as $ \mathsf{norm}(  \Theta ^3  ) $.

All of the following properties will depend on a few more 
properties of the $ \textsf{ctype} $ function.  They are listed in 
the following lemma.
\begin{lemma}[Properties of $ \textsf{ctype} $ Continued]
  \label{lemma:ctype_props_cont}
  \begin{itemize}
    \item[]
  \item[i.] If $ \Gamma  \LamDeltasym{,}  \LamDeltamv{x}  \LamDeltasym{:}  \LamDeltant{T}  \LamDeltasym{,}  \Gamma'  \vdash  \LamDeltant{t_{{\mathrm{1}}}} \, \LamDeltant{t_{{\mathrm{2}}}}  :  \LamDeltant{T'} $, $ \Gamma  \vdash  \LamDeltant{t}  :  \LamDeltant{T} $, 
    $ [  \LamDeltant{t}  /  \LamDeltamv{x}  ]^{ \LamDeltant{T} }  \LamDeltant{t_{{\mathrm{1}}}}  =  \lambda  \LamDeltamv{y} : \LamDeltant{T_{{\mathrm{1}}}}  .  \LamDeltant{t'} $, and $\LamDeltant{t_{{\mathrm{1}}}}$ is not a
    $\lambda$-abstraction, then $t_1$ is in head normal form and
    there exists a type $\LamDeltant{A}$ such that $ \textsf{ctype}_{ \LamDeltant{T} }( \LamDeltamv{x} , \LamDeltant{t_{{\mathrm{1}}}} )  = \LamDeltant{A}$.
    
  \item[ii.] If $ \Gamma  \LamDeltasym{,}  \LamDeltamv{x}  \LamDeltasym{:}  \LamDeltant{T}  \LamDeltasym{,}  \Gamma'  \vdash  \LamDeltant{t_{{\mathrm{1}}}} \, \LamDeltant{t_{{\mathrm{2}}}}  :  \LamDeltant{T'} $, $ \Gamma  \vdash  \LamDeltant{t}  :  \LamDeltant{T} $, 
    $ [  \LamDeltant{t}  /  \LamDeltamv{x}  ]^{ \LamDeltant{T} }  \LamDeltant{t_{{\mathrm{1}}}}  =  \Delta  \LamDeltamv{y}  :   \neg  \LamDeltasym{(}   \LamDeltant{T''}  \to  \LamDeltant{T'}   \LamDeltasym{)}   .  \LamDeltant{t'} $, and $\LamDeltant{t_{{\mathrm{1}}}}$ is not a
    $\Delta$-abstraction, then there exists a type $\LamDeltant{A}$ such that
    $ \textsf{ctype}_{ \LamDeltant{T} }( \LamDeltamv{x} , \LamDeltant{t_{{\mathrm{1}}}} )  = \LamDeltant{A}$.
  \end{itemize}
\end{lemma}
\begin{proof}
  Both parts can be shown by induction on the structure of $\LamDeltant{t_{{\mathrm{1}}}} \, \LamDeltant{t_{{\mathrm{2}}}}$.  
  See Appendix~\ref{subsec:proof_of_ctype_props_cont}.
\end{proof}
These are all similar to the properties in Lemma~\ref{lemma:ctype_props_cont1}.
The first two properties of the hereditary substitution function are
totality and type preservation.  The latter is similar to substitution
for typing using the hereditary substitution function.
\begin{lemma}[Totality and Type Preservation]
  \label{lemma:totality_and_type_preservation}
  \begin{itemize}
  \item[]
  \item[i.] If $ \Gamma  \vdash   \Theta ^3   :  \LamDeltant{A} $ and $ \Gamma  \LamDeltasym{,}   \Theta ^1 :   \neg  \LamDeltasym{(}   \LamDeltant{A}  \to  \LamDeltant{A'}   \LamDeltasym{)}    \vdash  \LamDeltant{t'}  :  \LamDeltant{B} $, then there
    exists a term $\LamDeltamv{s}$ such that $ \langle  \Theta  \rangle^{ \LamDeltant{A} }_{ \LamDeltant{A'} }  \LamDeltant{t'}  = \LamDeltamv{s}$ and $ \Gamma  \LamDeltasym{,}   \Theta ^2 :   \neg  \LamDeltant{A'}    \vdash  \LamDeltant{s}  :  \LamDeltant{B} $.
  
  \item[ii.] If $ \Gamma  \vdash  \LamDeltant{t}  :  \LamDeltant{A} $ and $ \Gamma  \LamDeltasym{,}  \LamDeltamv{x}  \LamDeltasym{:}  \LamDeltant{A}  \LamDeltasym{,}  \Gamma'  \vdash  \LamDeltant{t'}  :  \LamDeltant{B} $, then there exists a term $\LamDeltamv{s}$ 
    such that $ [  \LamDeltant{t}  /  \LamDeltamv{x}  ]^{ \LamDeltant{A} }  \LamDeltant{t'}  = \LamDeltamv{s}$ and $ \Gamma  \LamDeltasym{,}  \Gamma'  \vdash  \LamDeltant{s}  :  \LamDeltant{B} $.
  \end{itemize}
\end{lemma}
\begin{proof}
  This can be shown by mutual induction using the lexicographic
  combination $(\LamDeltant{A}, f,\LamDeltant{t'})$ of our ordering on types, the
  natural number ordering where $f \in \{0,1\}$, and the strict
  subexpression ordering on terms. See
  Appendix~\ref{subsec:proof_of_totality_and_type_preservation}.
\end{proof}
The next property shows that the hereditary substitution function is
normality preserving.  That is, if the input to the hereditary
substitution function is normal then so is the output.  This is
crucial for the normalization argument. The proof of normality
preservation depends on the following auxiliary result.
\begin{lemma}
  \label{lemma:ssub_var_head}
  For any $\Theta$, $\LamDeltant{A}$ and $\LamDeltant{A'}$, if $\LamDeltant{n_{{\mathrm{1}}}} \, \LamDeltant{n_{{\mathrm{2}}}}$ is normal then 
  $ \textsf{head}(  \langle  \Theta  \rangle^{ \LamDeltant{A} }_{ \LamDeltant{A'} }  \LamDeltasym{(}  \LamDeltant{n_{{\mathrm{1}}}} \, \LamDeltant{n_{{\mathrm{2}}}}  \LamDeltasym{)}  ) $ is a variable.
\end{lemma}
\begin{proof}
  This proof is by induction on the form of $\LamDeltant{n_{{\mathrm{1}}}} \, \LamDeltant{n_{{\mathrm{2}}}}$. See
  Appendix~\ref{subsec:proof_of_ssub_var_head}.
\end{proof}
\begin{lemma}[Normality Preservation]
  \label{lemma:normality_preservation}
  \begin{itemize}
    \item[]
  \item[i.] If $ \mathsf{norm}(  \Theta ^3  ) $, $ \Gamma  \vdash   \Theta ^3   :  \LamDeltant{A} $ and $ \Gamma  \LamDeltasym{,}   \Theta ^1 :   \neg  \LamDeltasym{(}   \LamDeltant{A}  \to  \LamDeltant{A'}   \LamDeltasym{)}    \vdash  \LamDeltant{n'}  :  \LamDeltant{B} $, 
    then there exists a normal form $\LamDeltant{m}$ such that $ \langle  \Theta  \rangle^{ \LamDeltant{A} }_{ \LamDeltant{A'} }  \LamDeltant{n'}  = \LamDeltant{m}$.
    
  \item[ii.] If $ \Gamma  \vdash  \LamDeltant{n}  :  \LamDeltant{A} $ and $ \Gamma  \LamDeltasym{,}  \LamDeltamv{x}  \LamDeltasym{:}  \LamDeltant{A}  \LamDeltasym{,}  \Gamma'  \vdash  \LamDeltant{n'}  :  \LamDeltant{B} $ then there exists a term $\LamDeltant{m}$ 
    such that $ [  \LamDeltant{n}  /  \LamDeltamv{x}  ]^{ \LamDeltant{A} }  \LamDeltant{n'}  = \LamDeltant{m}$. 
  \end{itemize} 
\end{lemma}
\begin{proof}
  This can be shown by mutual induction using the lexicographic
  combination $(\LamDeltant{A}, f,\LamDeltant{t'})$ of our ordering on types, the
  natural number ordering where $f \in \{0,1\}$, and the strict
  subexpression ordering on terms.
  See Appendix~\ref{subsec:proof_of_normality_preservation}.
\end{proof}
\noindent
The final correctness property of the hereditary substitution function is
soundness with respect to reduction.  We need one last piece of notation.
Suppose $\Theta = (\LamDeltamv{x_{{\mathrm{1}}}},\LamDeltamv{z_{{\mathrm{1}}}},\LamDeltant{t_{{\mathrm{1}}}}),\ldots,(\LamDeltamv{x_{\LamDeltamv{i}}},\LamDeltamv{z_{\LamDeltamv{i}}},\LamDeltant{t_{\LamDeltamv{i}}})$ for
some natural number $\LamDeltamv{i}$.
Then $ \langle  \Theta  \rangle^{\uparrow^{ \LamDeltant{A} }_{ \LamDeltant{A'} } }\, \LamDeltant{t'}  =^{def} [  \lambda  \LamDeltamv{y} :  \LamDeltant{A}  \to  \LamDeltant{A'}   .  \LamDeltasym{(}  \LamDeltamv{z_{\LamDeltamv{i}}} \, \LamDeltasym{(}  \LamDeltamv{y} \, \LamDeltant{t_{\LamDeltamv{i}}}  \LamDeltasym{)}  \LamDeltasym{)} /\LamDeltamv{x_{\LamDeltamv{i}}} ](\cdots([  \lambda  \LamDeltamv{y} :  \LamDeltant{A}  \to  \LamDeltant{A'}   .  \LamDeltasym{(}  \LamDeltamv{z_{{\mathrm{1}}}} \, \LamDeltasym{(}  \LamDeltamv{y} \, \LamDeltant{t_{{\mathrm{1}}}}  \LamDeltasym{)}  \LamDeltasym{)} /\LamDeltamv{x_{{\mathrm{1}}}} ]\LamDeltant{t_{{\mathrm{1}}}})\cdots)$.

\begin{lemma}[Soundness with Respect to Reduction]
  \label{lemma:soundness_reduction}  
  \begin{itemize}
  \item[]
  \item[i.] If $ \Gamma  \vdash   \Theta ^3   :  \LamDeltant{A} $ and $ \Gamma  \LamDeltasym{,}   \Theta ^1 :   \neg  \LamDeltasym{(}   \LamDeltant{A}  \to  \LamDeltant{A'}   \LamDeltasym{)}    \vdash  \LamDeltant{t'}  :  \LamDeltant{B} $, then
    $ \langle  \Theta  \rangle^{\uparrow^{ \LamDeltant{A} }_{ \LamDeltant{A'} } }\, \LamDeltant{t'}   \redto^*   \langle  \Theta  \rangle^{ \LamDeltant{A} }_{ \LamDeltant{A'} }  \LamDeltant{t'} $.
  
  \item[ii.] If $ \Gamma  \vdash  \LamDeltant{t}  :  \LamDeltant{A} $ and $ \Gamma  \LamDeltasym{,}  \LamDeltamv{x}  \LamDeltasym{:}  \LamDeltant{A}  \LamDeltasym{,}  \Gamma'  \vdash  \LamDeltant{t'}  :  \LamDeltant{B} $ then 
    $\LamDeltasym{[}  \LamDeltant{t}  \LamDeltasym{/}  \LamDeltamv{x}  \LamDeltasym{]}  \LamDeltant{t'}  \redto^*   [  \LamDeltant{t}  /  \LamDeltamv{x}  ]^{ \LamDeltant{A} }  \LamDeltant{t'} $.   
  \end{itemize}  
\end{lemma}
\begin{proof}
  This can be shown by mutual induction using the lexicographic
  combination $(\LamDeltant{A}, f,\LamDeltant{t'})$ of our ordering on types, the
  natural number ordering where $f \in \{0,1\}$, and the strict
  subexpression ordering on terms.
  See Appendix~\ref{subsec:soundness_with_respect_to_reduction}.
\end{proof}
Using these properties it is now possible to conclude normalization for the $\lambda\Delta$-calculus.
% subsection properties_of_the_hereditary_substitution_function (end)
% section the_hereditary_substitution_function_for_the_ld-calculus (end)

%
\section{Concluding Normalization}
\label{sec:concluding_normalization}
We now define the interpretation $\interp{\LamDeltant{T}}_\Gamma$ of types
$\LamDeltant{T}$ in typing context $\Gamma$.  This is in fact the same
interpretation of types that was used to show normalization using
hereditary substitution of Stratified System F in \cite{Eades:2010}.
\begin{definition}
  \label{def:semantics}
  The interpretation of types $\interp{\LamDeltant{T}}_\Gamma$ is defined by:
  \begin{center}
    \begin{math}
      \begin{array}{lll}
        \LamDeltant{n} \in \interp{\LamDeltant{T}}_\Gamma & \iff &  \Gamma  \vdash  \LamDeltant{n}  :  \LamDeltant{T} 
      \end{array}
    \end{math}
  \end{center}
  We extend this definition to non-normal terms $t$ in the following way:
  \begin{center}
    \begin{math}
      \begin{array}{lll}
        \LamDeltant{t} \in \interp{\LamDeltant{T}}_\Gamma & \iff & \exists \LamDeltant{n}.\LamDeltant{t} \normto \LamDeltant{n} \in \interp{\LamDeltant{T}}_\Gamma
      \end{array}
    \end{math}
  \end{center}
\end{definition}
Type soundness depends on the following lemma. It shows that the
interpretation of types is closed under hereditary substitution.
\begin{lemma}[Hereditary Substitution for the Interpretation of Types]
  \label{lemma:substitution_for_the_interpretation_of_types}
  If $\LamDeltant{n} \in \interp{\LamDeltant{T}}_\Gamma$ and $\LamDeltant{n'} \in \interp{\LamDeltant{T'}}_{\Gamma  \LamDeltasym{,}  \LamDeltamv{x}  \LamDeltasym{:}  \LamDeltant{T}  \LamDeltasym{,}  \Gamma'}$, then
  $ [  \LamDeltant{n}  /  \LamDeltamv{x}  ]^{ \LamDeltant{T} }  \LamDeltant{n'}  \in \interp{\LamDeltant{T'}}_{\Gamma  \LamDeltasym{,}  \Gamma'}$.
\end{lemma}
\begin{proof}
  We know by Lemma~\ref{lemma:totality_and_type_preservation} that there exists a term $\LamDeltamv{s}$ such that
  $ [  \LamDeltant{n}  /  \LamDeltamv{x}  ]^{ \LamDeltant{T} }  \LamDeltant{n'}   \LamDeltasym{=}  \LamDeltant{s}$ and $ \Gamma  \LamDeltasym{,}  \Gamma'  \vdash  \LamDeltant{s}  :  \LamDeltant{T'} $, and by Lemma~\ref{lemma:normality_preservation} $\LamDeltamv{s}$ is
  normal.  Therefore, $\LamDeltamv{s} \in \interp{\LamDeltant{T'}}_{\Gamma  \LamDeltasym{,}  \Gamma'}$.
\end{proof}
Using the previous lemma and the properties of the hereditary
substitution function we can now prove type soundness.
\begin{thm}[Type Soundness]
  \label{theorem:type_soundness}
  If $ \Gamma  \vdash  \LamDeltant{t}  :  \LamDeltant{T} $ then $\LamDeltant{t} \in \interp{\LamDeltant{T}}_\Gamma$.
\end{thm}
The only hard case in the proof of the type soundness theorem is the
case for applications. Using the previous lemma and the properties of
the hereditary substitution function, however, it goes through with
ease.  Consider the application case of the proof of type soundness.
Note that the proof is by induction on the assumed typing derivation (See Appendix~\ref{subsec:proof_of_type_soundness}).
\begin{center}
  \begin{math}
    \LamDeltadruleApp{}{}
  \end{math}
\end{center}
By the induction hypothesis we know $\LamDeltant{t_{{\mathrm{1}}}} \in \interp{ \LamDeltant{A}  \to  \LamDeltant{B} }_\Gamma$ and $\LamDeltant{t_{{\mathrm{2}}}} \in \interp{\LamDeltant{A}}_\Gamma$.
  So by the definition of the interpretation of types we know there exists normal forms $\LamDeltant{n_{{\mathrm{1}}}}$ and $\LamDeltant{n_{{\mathrm{2}}}}$
  such that $\LamDeltant{t_{{\mathrm{1}}}}  \redto^*  \LamDeltant{n_{{\mathrm{1}}}} \in \interp{ \LamDeltant{A}  \to  \LamDeltant{B} }_\Gamma$ and $\LamDeltant{t_{{\mathrm{2}}}}  \redto^*  \LamDeltant{n_{{\mathrm{2}}}} \in \interp{\LamDeltant{A}}_\Gamma$. Assume $\LamDeltamv{y}$ is a fresh
  variable in $\LamDeltant{n_{{\mathrm{1}}}}$ and $\LamDeltant{n_{{\mathrm{2}}}}$ of type $\LamDeltant{A}$.    
  Then by hereditary 
  substitution for the interpretation of types (Lemma~\ref{lemma:substitution_for_the_interpretation_of_types}) 
  $ [  \LamDeltant{n_{{\mathrm{1}}}}  /  \LamDeltamv{y}  ]^{ \LamDeltant{A} }  \LamDeltasym{(}  \LamDeltamv{y} \, \LamDeltant{n_{{\mathrm{2}}}}  \LamDeltasym{)}  \in \interp{\LamDeltant{B}}_\Gamma$.  
  It suffices to show that $\LamDeltant{t_{{\mathrm{1}}}} \, \LamDeltant{t_{{\mathrm{2}}}}  \redto^*   [  \LamDeltant{n_{{\mathrm{1}}}}  /  \LamDeltamv{y}  ]^{ \LamDeltant{A} }  \LamDeltasym{(}  \LamDeltamv{y} \, \LamDeltant{n_{{\mathrm{2}}}}  \LamDeltasym{)} $.  This is an easy consequence of soundness with respect
  to reduction (Lemma~\ref{lemma:soundness_reduction}), that is, $\LamDeltant{t_{{\mathrm{1}}}} \, \LamDeltant{t_{{\mathrm{2}}}}  \redto^*  \LamDeltant{n_{{\mathrm{1}}}} \, \LamDeltant{n_{{\mathrm{2}}}} = \LamDeltasym{[}  \LamDeltant{n_{{\mathrm{1}}}}  \LamDeltasym{/}  \LamDeltamv{y}  \LamDeltasym{]}  \LamDeltasym{(}  \LamDeltamv{y} \, \LamDeltant{n_{{\mathrm{2}}}}  \LamDeltasym{)}$ 
  and by soundness with respect to reduction $\LamDeltasym{[}  \LamDeltant{n_{{\mathrm{1}}}}  \LamDeltasym{/}  \LamDeltamv{y}  \LamDeltasym{]}  \LamDeltasym{(}  \LamDeltamv{y} \, \LamDeltant{n_{{\mathrm{2}}}}  \LamDeltasym{)}  \redto^*   [  \LamDeltant{n_{{\mathrm{1}}}}  /  \LamDeltamv{y}  ]^{ \LamDeltant{A} }  \LamDeltasym{(}  \LamDeltamv{y} \, \LamDeltant{n_{{\mathrm{2}}}}  \LamDeltasym{)} $.  Therefore, 
  $\LamDeltant{t_{{\mathrm{1}}}} \, \LamDeltant{t_{{\mathrm{2}}}} \in \interp{\LamDeltant{B}}_\Gamma$.  

\ \\
\noindent
Finally, we conclude normalization for the $\lambda\Delta$-calculus using hereditary substitution.
\begin{corollary}[Normalization]
  \label{corollary:normalization}
  If $ \Gamma  \vdash  \LamDeltant{t}  :  \LamDeltant{T} $ then there exists a term $\LamDeltant{n}$ such that $\LamDeltant{t} \normto \LamDeltant{n}$.
\end{corollary}
% section concluding_normalization (end)

%
\section{Related Work}
\label{sec:related_work}
We first compare the proof method normalization using hereditary
substitution with other known proof methods.  The
$\lambda\Delta$-calculus could have been proven weakly and strongly
normalizing by translation to $\lambda\mu$-calculus.  It is true that
this is not as complicated as the proof method here, but a proof by
translation does not yield a direct proof.  

A direct proof of weak and strong normalization could have been given
using the Tait-Girard reducibility method.  However, we claim that the
proof method used here is less complicated.  The statement of the type
soundness theorem is qualitatively less complex due to the fact that
there is no need to universally quantify over the set of well-formed
substitutions.  We are able to prove type soundness on open terms
directly.  Additionally, the formalization of normalization using
hereditary substitution does not require recursive types to define the
semantics of types which are required when formalizing a proof using
reducibility.  

R. David and K. Nour give a short proof of normalization of the
$\lambda\Delta$-calculus in \cite{David:2003}.  There they use a
rather complicated lexicographic combination to give a completely
arithmetical proof of strong normalization.  While they show strong
normalization their proof method is comparable to using hereditary
substitution.  As we mentioned in the introduction hereditary
substitution is the constructive content of normalization proofs using
the lexicographic combination of an ordering on types and the strict
subexpression ordering on terms.  It is currently unknown if
hereditary substitution can be extended to show strong normalization,
but we conjecture that the constructive content of the proof of Lemma
3..6 in David and Nour's work would yield a hereditary substitution
like function.  Furthermore, for simply typed theories we believe it
is enough to show weak normalization and never need to show strong
normalization.  It is well-known due to the work of G. Barthe et
al. in \cite{Barthe:2001} that for the entire left hand side of the
$\lambda$-cube weak normalization implies strong normalization.  We
conjecture that this result would extend to the left hand side of the
classical $\lambda$-cube given in \cite{Barthe:1997}.  Thus, showing
normalization using hereditary substitution is less complicated than
the work of David and Nour's.

Similar to the work of David and Nour is the work of F. Joachimski and
R. Matthes.  In \cite{Joachimski:1999} they prove weak and strong
normalization of various simply typed theories.  The proof method used
is induction on various lexicographic combinations similar to
hereditary substitution.  After proving weak normalization of each
type theory they extract the constructive content of the proof
yielding a normalization function which depends on a substitution
function similar to the hereditary substitution function.  In contrast
once hereditary substitution is defined for a type theory we can easily
define a normalization function. Note that the following function is 
the computational content of the type-soundness theorem 
(Theorem~\ref{theorem:type_soundness}).
\begin{definition}
  \label{def:norm_fun_hs}
  We define a normalization function for the $\lambda\Delta$-calculus using hereditary
  substitution as follows:
  \begin{center}
    \begin{itemize}
    \item[] $ \mathsf{norm}\, \LamDeltamv{x}  = \LamDeltamv{x}$\\
    \item[] $ \mathsf{norm}\, \LamDeltasym{(}   \lambda  \LamDeltamv{x} : \LamDeltant{A}  .  \LamDeltant{t}   \LamDeltasym{)}  =  \lambda  \LamDeltamv{x} : \LamDeltant{A}  .  \LamDeltasym{(}   \mathsf{norm}\, \LamDeltant{t}   \LamDeltasym{)} $\\
    \item[] $ \mathsf{norm}\, \LamDeltasym{(}   \Delta  \LamDeltamv{x}  :  \LamDeltant{A}  .  \LamDeltant{t}   \LamDeltasym{)}  =  \Delta  \LamDeltamv{x}  :  \LamDeltant{A}  .  \LamDeltasym{(}   \mathsf{norm}\, \LamDeltant{t}   \LamDeltasym{)} $\\
    \item[] $ \mathsf{norm}\, \LamDeltasym{(}  \LamDeltant{t_{{\mathrm{1}}}} \, \LamDeltant{t_{{\mathrm{2}}}}  \LamDeltasym{)}  =  [  \LamDeltant{n_{{\mathrm{1}}}}  /  \LamDeltamv{r}  ]^{ \LamDeltant{A} }  \LamDeltasym{(}  \LamDeltamv{r} \, \LamDeltant{n_{{\mathrm{2}}}}  \LamDeltasym{)} $\\
      \begin{tabular}{lll}
      & Where $ \mathsf{norm}\, \LamDeltant{t_{{\mathrm{1}}}}  = \LamDeltant{n_{{\mathrm{1}}}}$, $ \mathsf{norm}\, \LamDeltant{t_{{\mathrm{2}}}}  = \LamDeltant{n_{{\mathrm{2}}}}$, $\LamDeltant{A}$ is
      the type of $\LamDeltant{t_{{\mathrm{1}}}}$, and $\LamDeltamv{r}$ is fresh in $\LamDeltant{t_{{\mathrm{1}}}}$\\
      &  and $\LamDeltant{t_{{\mathrm{2}}}}$.\\
    \end{tabular}
    \end{itemize}
  \end{center}
\end{definition}

This function is similar to the normalization functions in Joachimski 
and Matthes' work.  We could use the above normalization function to
decide $\beta\eta$-equality for the $\lambda\Delta$-calculus.  Indeed
this one of the main application of hereditary substitution.

A. Abel in 2006 shows how to implement a normalizer using sized
heterogeneous types which is a function similar to the hereditary
substitution function in \cite{Abel:2006}.  He then uses hereditary
substitution to prove normalization of the type level of a type theory
with higher-order subtyping in \cite{Abel:2008}.  This results in a
purely syntactic metatheory.  C. Keller and T. Altenkirch recently
implemented hereditary substitution as a normalization function for
the simply typed $\lambda$-calculus in Agda \cite{Keller:2010}.  Their
results show that hereditary substitution can be used to decide $\beta
\eta$-equality.  They found hereditary substitution to be convenient
to use in a total type theory, because it can be implemented without a
termination proof.  This is because the hereditary-substitution
function can be recognized as structurally recursive, and hence
accepted directly by Agda's termination checker.

One point which sets the current work apart from all of the related
work just considered is that they were all concerned with
intuitionistic type theories.  Here we apply hereditary substitution
on a classical type theory.  To our knowledge this is the first time
this has been done.
% section defining_a_normalization_function (end)

%
\section{Conclusion}
\label{sec:conclusion}
We briefly gave an overview of the hereditary substitution proof
method for showing normalization of typed $\lambda$-calculi and showed
how to extend and apply it to the $\lambda\Delta$-calculus.  In
Section~\ref{subsec:the_final_extension} we
defined the hereditary substitution function for the
$\lambda\Delta$-calculus which involved a new function called the
structural hereditary substitution function.  Then we proved the main
properties of the hereditary substitution function in
Section~\ref{subsec:properties_of_the_hereditary_substitution_function}.
Lastly, we concluded normalization in Section~\ref{sec:concluding_normalization}.

\textbf{Future work.} The authors conjecture that the current work may
extend to yield a direct proof of normalization of the
$\lambda\mu$-calculus using hereditary substitution.
% section conclusion (end)

\bibliographystyle{eptcs}

\begin{thebibliography}{10}
\providecommand{\bibitemdeclare}[2]{}
\providecommand{\surnamestart}{}
\providecommand{\surnameend}{}
\providecommand{\urlprefix}{Available at }
\providecommand{\url}[1]{\texttt{#1}}
\providecommand{\href}[2]{\texttt{#2}}
\providecommand{\urlalt}[2]{\href{#1}{#2}}
\providecommand{\doi}[1]{doi:\urlalt{http://dx.doi.org/#1}{#1}}
\providecommand{\bibinfo}[2]{#2}

\bibitemdeclare{inproceedings}{Abel:2006}
\bibitem{Abel:2006}
\bibinfo{author}{Andreas \surnamestart Abel\surnameend} (\bibinfo{year}{2006}):
  \emph{\bibinfo{title}{Implementing a normalizer using sized heterogeneous
  types}}.
\newblock In: {\sl \bibinfo{booktitle}{In Workshop on Mathematically Structured
  Functional Programming, MSFP}}, \doi{10.1017/S0956796809007266}.

\bibitemdeclare{inproceedings}{Abel:2008}
\bibitem{Abel:2008}
\bibinfo{author}{Andreas \surnamestart Abel\surnameend} \&
  \bibinfo{author}{Dulma \surnamestart Rodriguez\surnameend}
  (\bibinfo{year}{2008}): \emph{\bibinfo{title}{Syntactic Metatheory of
  Higher-Order Subtyping}}.
\newblock In: {\sl \bibinfo{booktitle}{Proceedings of the 22nd international
  workshop on Computer Science Logic}}, \bibinfo{series}{CSL '08},
  \bibinfo{publisher}{Springer-Verlag}, \bibinfo{address}{Berlin, Heidelberg},
  pp. \bibinfo{pages}{446--460}, \doi{10.1007/978-3-540-87531-4\_32}.

\bibitemdeclare{phdthesis}{Adams:2004}
\bibitem{Adams:2004}
\bibinfo{author}{Robin \surnamestart Adams\surnameend} (\bibinfo{year}{2004}):
  \emph{\bibinfo{title}{A Modular Hierarchy of Logical Frameworks}}.
\newblock Ph.D. thesis.

\bibitemdeclare{book}{Amadio:1998}
\bibitem{Amadio:1998}
\bibinfo{author}{R.M. \surnamestart Amadio\surnameend} \& \bibinfo{author}{P.L.
  \surnamestart Curien\surnameend} (\bibinfo{year}{1998}):
  \emph{\bibinfo{title}{Domains and lambda-calculi}}.
\newblock \bibinfo{series}{Cambridge tracts in theoretical computer science},
  \bibinfo{publisher}{Cambridge University Press},
  \doi{10.1017/CBO9780511983504}.

\bibitemdeclare{article}{Barthe:1997}
\bibitem{Barthe:1997}
\bibinfo{author}{G.~\surnamestart Barthe\surnameend},
  \bibinfo{author}{J.~\surnamestart Hatcliff\surnameend} \&
  \bibinfo{author}{M.~\surnamestart Heine~S{\o}rensen\surnameend}
  (\bibinfo{year}{1997}): \emph{\bibinfo{title}{A notion of classical pure type
  system (preliminary version)}}.
\newblock {\sl \bibinfo{journal}{Electronic Notes in Theoretical Computer
  Science}} \bibinfo{volume}{6}, pp. \bibinfo{pages}{4--59},
  \doi{10.1016/S1571-0661(05)80170-7}.

\bibitemdeclare{article}{Barthe:2001}
\bibitem{Barthe:2001}
\bibinfo{author}{Gilles \surnamestart Barthe\surnameend}, \bibinfo{author}{John
  \surnamestart Hatcliff\surnameend} \& \bibinfo{author}{Morten~Heine
  \surnamestart SÃ¸rensen\surnameend} (\bibinfo{year}{2001}):
  \emph{\bibinfo{title}{Weak normalization implies strong normalization in a
  class of non-dependent pure type systems}}.
\newblock {\sl \bibinfo{journal}{Theoretical Computer Science}}
  \bibinfo{volume}{269}(\bibinfo{number}{1-2}), pp. \bibinfo{pages}{317 --
  361}, \doi{10.1016/S0304-3975(01)00012-3}.

\bibitemdeclare{article}{David:2003}
\bibitem{David:2003}
\bibinfo{author}{Rene \surnamestart David\surnameend} \& \bibinfo{author}{Karim
  \surnamestart Nour\surnameend} (\bibinfo{year}{2003}):
  \emph{\bibinfo{title}{A short proof of the strong normalization of the simply
  typed lambdamu-calculus}}.
\newblock {\sl \bibinfo{journal}{SCHEDAE INFORMATICAE}} \bibinfo{volume}{12},
  pp. \bibinfo{pages}{27--33}.

\bibitemdeclare{conference}{Eades:2010}
\bibitem{Eades:2010}
\bibinfo{author}{Harley \surnamestart Eades\surnameend} \&
  \bibinfo{author}{Aaron \surnamestart Stump\surnameend}
  (\bibinfo{year}{2010}): \emph{\bibinfo{title}{Hereditary Substitution for
  Stratified System F}}.
\newblock In: {\sl \bibinfo{booktitle}{Proof-Search in Type Theories (PSTT)}}.

\bibitemdeclare{book}{Girard:1989}
\bibitem{Girard:1989}
\bibinfo{author}{Jean-Yves \surnamestart Girard\surnameend},
  \bibinfo{author}{Yves \surnamestart Lafont\surnameend} \&
  \bibinfo{author}{Paul \surnamestart Taylor\surnameend}
  (\bibinfo{year}{1989}): \emph{\bibinfo{title}{Proofs and Types (Cambridge
  Tracts in Theoretical Computer Science)}}.
\newblock \bibinfo{publisher}{{Cambridge University Press}}.

\bibitemdeclare{unpublished}{Eades:2011}
\bibitem{Eades:2011}
\bibinfo{author}{Harley~Eades \surnamestart II\surnameend} \&
  \bibinfo{author}{Aaron \surnamestart Stump\surnameend}
  (\bibinfo{year}{2011}): \emph{\bibinfo{title}{Using the Hereditary
  Substitution Function in Normalization Proofs}}.
\newblock
  \urlprefix\url{http://metatheorem.org/wp-content/papers/qual_companion_report.pdf}.

\bibitemdeclare{misc}{Joachimski:1999}
\bibitem{Joachimski:1999}
\bibinfo{author}{Felix \surnamestart Joachimski\surnameend} \&
  \bibinfo{author}{Ralph \surnamestart Matthes\surnameend}
  (\bibinfo{year}{1999}): \emph{\bibinfo{title}{Short Proofs of Normalization
  for the simply-typed lambda-calculus, permutative conversions and G{\"o}del's
  T}}.

\bibitemdeclare{inproceedings}{Keller:2010}
\bibitem{Keller:2010}
\bibinfo{author}{Chantal \surnamestart Keller\surnameend} \&
  \bibinfo{author}{Thorsten \surnamestart Altenkirch\surnameend}
  (\bibinfo{year}{2010}): \emph{\bibinfo{title}{Hereditary substitutions for
  simple types, formalized}}.
\newblock In: {\sl \bibinfo{booktitle}{Proceedings of the third ACM SIGPLAN
  workshop on Mathematically structured functional programming}},
  \bibinfo{series}{MSFP '10}, \bibinfo{publisher}{ACM}, \bibinfo{address}{New
  York, NY, USA}, pp. \bibinfo{pages}{3--10}, \doi{10.1145/1863597.1863601}.

\bibitemdeclare{article}{Leivant:1991}
\bibitem{Leivant:1991}
\bibinfo{author}{D.~\surnamestart Leivant\surnameend} (\bibinfo{year}{1991}):
  \emph{\bibinfo{title}{Finitely stratified polymorphism}}.
\newblock {\sl \bibinfo{journal}{Inf. Comput.}}
  \bibinfo{volume}{93}(\bibinfo{number}{1}), pp. \bibinfo{pages}{93--113},
  \doi{10.1016/0890-5401(91)90053-5}.

\bibitemdeclare{article}{Levy:1976}
\bibitem{Levy:1976}
\bibinfo{author}{Jean-Jacques \surnamestart L{\'e}vy\surnameend}
  (\bibinfo{year}{1976}): \emph{\bibinfo{title}{An algebraic interpretation of
  the λβK-calculus; and an application of a labelled λ-calculus}}.
\newblock {\sl \bibinfo{journal}{Theoretical Computer Science}}
  \bibinfo{volume}{2}(\bibinfo{number}{1}), pp. \bibinfo{pages}{97 -- 114},
  \doi{10.1016/0304-3975(76)90009-8}.

\bibitemdeclare{incollection}{Parigot:1992}
\bibitem{Parigot:1992}
\bibinfo{author}{Michel \surnamestart Parigot\surnameend}
  (\bibinfo{year}{1992}): \emph{\bibinfo{title}{Lambda-Mu-Calculus: An
  algorithmic interpretation of classical natural deduction}}.
\newblock In \bibinfo{editor}{Andrei \surnamestart Voronkov\surnameend},
  editor: {\sl \bibinfo{booktitle}{Logic Programming and Automated Reasoning}},
  {\sl \bibinfo{series}{Lecture Notes in Computer Science}}
  \bibinfo{volume}{624}, \bibinfo{publisher}{Springer Berlin / Heidelberg}, pp.
  \bibinfo{pages}{190--201}, \doi{10.1007/BFb0013061}.

\bibitemdeclare{article}{Parigot:1997}
\bibitem{Parigot:1997}
\bibinfo{author}{Michel \surnamestart Parigot\surnameend}
  (\bibinfo{year}{1997}): \emph{\bibinfo{title}{Proofs of Strong Normalization
  for Second Order Classical Natural Deduction}}.
\newblock {\sl \bibinfo{journal}{Journal of Symbolic Logic}}
  \bibinfo{volume}{62}(\bibinfo{number}{4}), pp. \bibinfo{pages}{1461--1479},
  \doi{10.2307/2275652}.

\bibitemdeclare{inproceedings}{Rehof:1994}
\bibitem{Rehof:1994}
\bibinfo{author}{Jakob \surnamestart Rehof\surnameend} \&
  \bibinfo{author}{Morten~Heine \surnamestart S{\o}rensen\surnameend}
  (\bibinfo{year}{1994}): \emph{\bibinfo{title}{The LambdaDelta-calculus}}.
\newblock In: {\sl \bibinfo{booktitle}{Proceedings of the International
  Conference on Theoretical Aspects of Computer Software}},
  \bibinfo{series}{TACS '94}, \bibinfo{publisher}{Springer-Verlag},
  \bibinfo{address}{London, UK}, pp. \bibinfo{pages}{516--542},
  \doi{10.1007/3-540-57887-0\_113}.

\bibitemdeclare{incollection}{Watkins:2004}
\bibitem{Watkins:2004}
\bibinfo{author}{Kevin \surnamestart Watkins\surnameend},
  \bibinfo{author}{Iliano \surnamestart Cervesato\surnameend},
  \bibinfo{author}{Frank \surnamestart Pfenning\surnameend} \&
  \bibinfo{author}{David \surnamestart Walker\surnameend}
  (\bibinfo{year}{2004}): \emph{\bibinfo{title}{A Concurrent Logical Framework:
  The Propositional Fragment}}.
\newblock In \bibinfo{editor}{Stefano \surnamestart Berardi\surnameend},
  \bibinfo{editor}{Mario \surnamestart Coppo\surnameend} \&
  \bibinfo{editor}{Ferruccio \surnamestart Damiani\surnameend}, editors: {\sl
  \bibinfo{booktitle}{Types for Proofs and Programs}}, {\sl
  \bibinfo{series}{Lecture Notes in Computer Science}} \bibinfo{volume}{3085},
  \bibinfo{publisher}{Springer Berlin / Heidelberg}, pp.
  \bibinfo{pages}{355--377},
  \doi{10.1007/978-3-540-24849-1\_23}.

\end{thebibliography}

% BIB

% END BIB

\appendix 

\section{Proofs}
\label{sec:proofs}
\subsection{Proof of Properties of $ \textsf{ctype}_{ \LamDeltant{T} } $}
\label{subsec:proof_of_ctype_props}
We prove part one first. This is a proof by induction on the structure of $t$.

\begin{itemize}
\item[Case.] Suppose $\LamDeltant{t} \equiv \LamDeltamv{x}$.  Then $ \textsf{ctype}_{ \LamDeltant{T} }( \LamDeltamv{x} , \LamDeltamv{x} )  = \LamDeltant{T}$.  Clearly,
  $ \textsf{head}( \LamDeltamv{x} )  = \LamDeltamv{x}$ and $\LamDeltant{T}$ is a subexpression of itself.
  
\item[Case.] Suppose $\LamDeltant{t} \equiv \LamDeltant{t_{{\mathrm{1}}}} \, \LamDeltant{t_{{\mathrm{2}}}}$.  Then $ \textsf{ctype}_{ \LamDeltant{T} }( \LamDeltamv{x} , \LamDeltant{t_{{\mathrm{1}}}} \, \LamDeltant{t_{{\mathrm{2}}}} )  = \LamDeltant{T''}$
  when $ \textsf{ctype}_{ \LamDeltant{T} }( \LamDeltamv{x} , \LamDeltant{t_{{\mathrm{1}}}} )  =  \LamDeltant{T'}  \to  \LamDeltant{T''} $.  Now $\LamDeltant{t} > \LamDeltant{t_{{\mathrm{1}}}}$ so by the induction
  hypothesis $ \textsf{head}( \LamDeltant{t_{{\mathrm{1}}}} )  = \LamDeltamv{x}$ and $ \LamDeltant{T'}  \to  \LamDeltant{T''} $ is a subexpression of $\LamDeltant{T}$.
  Therefore, $ \textsf{head}( \LamDeltant{t_{{\mathrm{1}}}} \, \LamDeltant{t_{{\mathrm{2}}}} )  = \LamDeltamv{x}$ and certainly $\LamDeltant{T''}$ is a subexpression of $\LamDeltant{T}$.
\end{itemize}

\ \\
We now prove part two.  This is also a proof by induction on the structure of $t$.

\begin{itemize}
\item[Case.] Suppose $\LamDeltant{t} \equiv \LamDeltamv{x}$.  Then $ \textsf{ctype}_{ \LamDeltant{T} }( \LamDeltamv{x} , \LamDeltamv{x} )  = \LamDeltant{T}$.  Clearly,
  $\LamDeltant{T} \equiv \LamDeltant{T}$.
  
\item[Case.] Suppose $\LamDeltant{t} \equiv \LamDeltant{t_{{\mathrm{1}}}} \, \LamDeltant{t_{{\mathrm{2}}}}$.  Then $ \textsf{ctype}_{ \LamDeltant{T} }( \LamDeltamv{x} , \LamDeltant{t_{{\mathrm{1}}}} \, \LamDeltant{t_{{\mathrm{2}}}} )  = \LamDeltant{T_{{\mathrm{2}}}}$
  when $ \textsf{ctype}_{ \LamDeltant{T} }( \LamDeltamv{x} , \LamDeltant{t_{{\mathrm{1}}}} )  =  \LamDeltant{T_{{\mathrm{1}}}}  \to  \LamDeltant{T_{{\mathrm{2}}}} $.  By inversion on the assumed typing
  derivation we know there exists type $\LamDeltant{T''}$ such that $ \Gamma  \LamDeltasym{,}  \LamDeltamv{x}  \LamDeltasym{:}  \LamDeltant{T}  \LamDeltasym{,}  \Gamma'  \vdash  \LamDeltant{t_{{\mathrm{1}}}}  :   \LamDeltant{T''}  \to  \LamDeltant{T'}  $.
  Now $\LamDeltant{t} > \LamDeltant{t_{{\mathrm{1}}}}$ so by the induction hypothesis $ \LamDeltant{T_{{\mathrm{1}}}}  \to  \LamDeltant{T_{{\mathrm{2}}}}  \equiv  \LamDeltant{T''}  \to  \LamDeltant{T'} $.
  Therefore, $\LamDeltant{T_{{\mathrm{1}}}} \equiv \LamDeltant{T''}$ and $\LamDeltant{T_{{\mathrm{2}}}} \equiv \LamDeltant{T'}$.
\end{itemize}

\subsection{Proof of Properties of $ \textsf{ctype}_{ \LamDeltant{T} } $ Continued}
\label{subsec:proof_of_ctype_props_cont}
We prove part one first. This is a proof by induction on the structure of $\LamDeltant{t_{{\mathrm{1}}}} \, \LamDeltant{t_{{\mathrm{2}}}}$.

\ \\
The only possibilities for the form of $\LamDeltant{t_{{\mathrm{1}}}}$ is $\LamDeltamv{x}$ or $\LamDeltant{s_{{\mathrm{1}}}} \, \LamDeltant{s_{{\mathrm{2}}}}$.  All other 
forms would not result in $ [  \LamDeltant{t}  /  \LamDeltamv{x}  ]^{ \LamDeltant{T} }  \LamDeltant{t_{{\mathrm{1}}}} $ being a $\lambda$-abstraction and $\LamDeltant{t_{{\mathrm{1}}}}$ not.
If $\LamDeltant{t_{{\mathrm{1}}}} \equiv \LamDeltamv{x}$ then there exist a type $\LamDeltant{T''}$ such that $\LamDeltant{T} \equiv  \LamDeltant{T''}  \to  \LamDeltant{T'} $ and
$ \textsf{ctype}_{ \LamDeltant{T} }( \LamDeltamv{x} , \LamDeltamv{x} \, \LamDeltant{t_{{\mathrm{2}}}} )  = \LamDeltant{T'}$ when $ \textsf{ctype}_{ \LamDeltant{T} }( \LamDeltamv{x} , \LamDeltamv{x} )  = \LamDeltant{T} \equiv  \LamDeltant{T''}  \to  \LamDeltant{T'} $ in this case.  We know
$\LamDeltant{T''}$ to exist by inversion on $ \Gamma  \LamDeltasym{,}  \LamDeltamv{x}  \LamDeltasym{:}  \LamDeltant{T}  \LamDeltasym{,}  \Gamma'  \vdash  \LamDeltant{t_{{\mathrm{1}}}} \, \LamDeltant{t_{{\mathrm{2}}}}  :  \LamDeltant{T'} $.

\ \\
Now suppose $\LamDeltant{t_{{\mathrm{1}}}} \equiv \LamDeltant{s_{{\mathrm{1}}}} \, \LamDeltant{s_{{\mathrm{2}}}}$.  Now knowing $\LamDeltant{t_{{\mathrm{1}}}}$ to not a $\lambda$-abstraction
implies that $\LamDeltant{s_{{\mathrm{1}}}}$ is also not a $\lambda$-abstraction or $ [  \LamDeltant{t}  /  \LamDeltamv{x}  ]^{ \LamDeltant{T} }  \LamDeltant{t_{{\mathrm{1}}}} $ would be an application
instead of a $\lambda$-abstraction.  So it must be the case that $ [  \LamDeltant{t}  /  \LamDeltamv{x}  ]^{ \LamDeltant{T} }  \LamDeltant{s_{{\mathrm{1}}}} $ is a $\lambda$-abstraction
and $\LamDeltant{s_{{\mathrm{1}}}}$ is not.  Since $\LamDeltant{s_{{\mathrm{1}}}} < \LamDeltant{t_{{\mathrm{1}}}}$ we can apply the induction hypothesis to obtain there exists
a type $\LamDeltant{A}$ such that $ \textsf{ctype}_{ \LamDeltant{T} }( \LamDeltamv{x} , \LamDeltant{s_{{\mathrm{1}}}} )  = \LamDeltant{A}$.  
Now by inversion on $ \Gamma  \LamDeltasym{,}  \LamDeltamv{x}  \LamDeltasym{:}  \LamDeltant{T}  \LamDeltasym{,}  \Gamma'  \vdash  \LamDeltant{t_{{\mathrm{1}}}} \, \LamDeltant{t_{{\mathrm{2}}}}  :  \LamDeltant{T'} $ we know there exists a type $\LamDeltant{T''}$ such that
$ \Gamma  \LamDeltasym{,}  \LamDeltamv{x}  \LamDeltasym{:}  \LamDeltant{T}  \LamDeltasym{,}  \Gamma'  \vdash  \LamDeltant{t_{{\mathrm{1}}}}  :   \LamDeltant{T''}  \to  \LamDeltant{T'}  $.  We know $\LamDeltant{t_{{\mathrm{1}}}} \equiv \LamDeltant{s_{{\mathrm{1}}}} \, \LamDeltant{s_{{\mathrm{2}}}}$ so by inversion on
$ \Gamma  \LamDeltasym{,}  \LamDeltamv{x}  \LamDeltasym{:}  \LamDeltant{T}  \LamDeltasym{,}  \Gamma'  \vdash  \LamDeltant{t_{{\mathrm{1}}}}  :   \LamDeltant{T''}  \to  \LamDeltant{T'}  $ we know there exists a type $\LamDeltant{A''}$ such that
$ \Gamma  \LamDeltasym{,}  \LamDeltamv{x}  \LamDeltasym{:}  \LamDeltant{T}  \LamDeltasym{,}  \Gamma'  \vdash  \LamDeltant{s_{{\mathrm{1}}}}  :   \LamDeltant{A''}  \to  \LamDeltasym{(}   \LamDeltant{T''}  \to  \LamDeltant{T'}   \LamDeltasym{)}  $.
By part two of Lemma~\ref{lemma:ctype_props} we know $\LamDeltant{A} \equiv  \LamDeltant{A''}  \to  \LamDeltasym{(}   \LamDeltant{T''}  \to  \LamDeltant{T'}   \LamDeltasym{)} $ and
$ \textsf{ctype}_{ \LamDeltant{T} }( \LamDeltamv{x} , \LamDeltant{t_{{\mathrm{1}}}} )  =  \textsf{ctype}_{ \LamDeltant{T} }( \LamDeltamv{x} , \LamDeltant{s_{{\mathrm{1}}}} \, \LamDeltant{s_{{\mathrm{2}}}} )  =  \LamDeltant{T''}  \to  \LamDeltant{T'} $ 
when $ \textsf{ctype}_{ \LamDeltant{T} }( \LamDeltamv{x} , \LamDeltant{s_{{\mathrm{1}}}} )  =  \LamDeltant{A''}  \to  \LamDeltasym{(}   \LamDeltant{T''}  \to  \LamDeltant{A'}   \LamDeltasym{)} $, because we know $ \textsf{ctype}_{ \LamDeltant{T} }( \LamDeltamv{x} , \LamDeltant{s_{{\mathrm{1}}}} )  = \LamDeltant{A}$.

\ \\
The proof of part two is similar to the proof of part one.
% subsection proof_of_ctype_props (end)

\subsection{Proof of Totality and Type Preservation}
\label{subsec:proof_of_totality_and_type_preservation}
This is a mutually inductive proof using the lexicographic combination
$(\LamDeltant{A}, f,\LamDeltant{t'})$ of our ordering on types,
the natural number ordering where $f \in \{0,1\}$, and
the strict subexpression ordering on terms. We first prove part one
and then part two.  In both parts we case split on $\LamDeltant{t'}$.

\ \\
\noindent Part One.
\begin{itemize}
\item[Case.] Suppose $\LamDeltant{t'}$ is a variable $\LamDeltamv{x}$.  Then either there exists
  a term $\LamDeltant{a}$ such that $ ( \LamDeltamv{x} , \LamDeltamv{z} , \LamDeltant{a} )  \in \Theta$ or not.  Suppose so. Then 
  $ \langle  \Theta  \rangle^{ \LamDeltant{A} }_{ \LamDeltant{A'} }  \LamDeltamv{x}  =  \lambda  \LamDeltamv{y} :  \LamDeltant{A}  \to  \LamDeltant{A'}   .  \LamDeltasym{(}  \LamDeltamv{z} \, \LamDeltasym{(}  \LamDeltamv{y} \, \LamDeltant{a}  \LamDeltasym{)}  \LamDeltasym{)} $ where $\LamDeltamv{y}$ is fresh in $\LamDeltamv{x}$, $\LamDeltamv{z}$ and $\LamDeltant{a}$.
  Now suppose there does not exist any term $\LamDeltant{a}$ or $\LamDeltamv{z}$ such that $ ( \LamDeltamv{x} , \LamDeltamv{z} , \LamDeltant{a} )  \in \Theta$.  Then
  $ \langle  \Theta  \rangle^{ \LamDeltant{A} }_{ \LamDeltant{A'} }  \LamDeltamv{x}  = \LamDeltamv{x}$. Typing clearly holds, because if $ ( \LamDeltamv{x} , \LamDeltamv{z} , \LamDeltant{a} )  \in \Theta$ then 
  $\LamDeltant{B} \equiv  \neg  \LamDeltasym{(}   \LamDeltant{A}  \to  \LamDeltant{A'}   \LamDeltasym{)} $ and we know $ \Gamma  \LamDeltasym{,}   \Theta ^2 :   \neg  \LamDeltant{A'}    \vdash   \lambda  \LamDeltamv{y} :  \LamDeltant{A}  \to  \LamDeltant{A'}   .  \LamDeltasym{(}  \LamDeltamv{z} \, \LamDeltasym{(}  \LamDeltamv{y} \, \LamDeltant{a}  \LamDeltasym{)}  \LamDeltasym{)}   :  \LamDeltant{B} $ or
  $\LamDeltamv{x} \not\in  \Theta ^1 $ then it must be the case that $\LamDeltamv{x}  \LamDeltasym{:}  \LamDeltant{B} \in \Gamma$, hence,
  by assumption and weakening for typing $ \Gamma  \LamDeltasym{,}   \Theta ^2 :   \neg  \LamDeltant{A'}    \vdash  \LamDeltamv{x}  :  \LamDeltant{B} $.

\item[Case.] It must be the case that $\LamDeltant{B} \equiv  \LamDeltant{B_{{\mathrm{1}}}}  \to  \LamDeltant{B_{{\mathrm{2}}}} $ for some types $\LamDeltant{B_{{\mathrm{1}}}}$ and
  $\LamDeltant{B_{{\mathrm{2}}}}$.  Suppose $\LamDeltant{t'} \equiv  \lambda  \LamDeltamv{y} : \LamDeltant{B_{{\mathrm{1}}}}  .  \LamDeltant{t'_{{\mathrm{1}}}} $. Then 
  $ \langle  \Theta  \rangle^{ \LamDeltant{A} }_{ \LamDeltant{A'} }  \LamDeltant{t'}  =  \langle  \Theta  \rangle^{ \LamDeltant{A} }_{ \LamDeltant{A'} }  \LamDeltasym{(}   \lambda  \LamDeltamv{y} : \LamDeltant{B_{{\mathrm{1}}}}  .  \LamDeltant{t'_{{\mathrm{1}}}}   \LamDeltasym{)}  =  \lambda  \LamDeltamv{y} : \LamDeltant{B_{{\mathrm{1}}}}  .   \langle  \Theta  \rangle^{ \LamDeltant{A} }_{ \LamDeltant{A'} }  \LamDeltant{t'_{{\mathrm{1}}}}  $.  Now sense
  $(\LamDeltant{A},1,\LamDeltant{t'}) > (\LamDeltant{A},1,\LamDeltant{t'_{{\mathrm{1}}}})$ we may apply the induction hypothesis to obtain
  that there exists a term $\LamDeltamv{s}$ such that $ \langle  \Theta  \rangle^{ \LamDeltant{A} }_{ \LamDeltant{A'} }  \LamDeltant{t'_{{\mathrm{1}}}}  = \LamDeltamv{s}$, and $ \Gamma  \LamDeltasym{,}   \Theta ^2 :   \neg  \LamDeltant{A'}    \LamDeltasym{,}  \LamDeltamv{y}  \LamDeltasym{:}  \LamDeltant{B_{{\mathrm{1}}}}  \vdash  \LamDeltant{s}  :  \LamDeltant{B_{{\mathrm{2}}}} $.
  Thus, by definition and the typing rule for $\lambda$-abstractions we obtain $ \langle  \Theta  \rangle^{ \LamDeltant{A} }_{ \LamDeltant{A'} }  \LamDeltant{t'}  =  \lambda  \LamDeltamv{y} : \LamDeltant{B_{{\mathrm{1}}}}  .  \LamDeltant{s} $
  and $ \Gamma  \LamDeltasym{,}   \Theta ^2 :   \neg  \LamDeltant{A'}    \vdash   \lambda  \LamDeltamv{y} : \LamDeltant{B_{{\mathrm{1}}}}  .  \LamDeltant{s}   :   \LamDeltant{B_{{\mathrm{1}}}}  \to  \LamDeltant{B_{{\mathrm{2}}}}  $.

\item[Case.] Suppose $\LamDeltant{t'} \equiv  \Delta  \LamDeltamv{y}  :   \neg  \LamDeltant{B}   .  \LamDeltant{t'_{{\mathrm{1}}}} $. Similar to the previous case.

\item[Case.] Suppose $\LamDeltant{t'} \equiv \LamDeltant{t'_{{\mathrm{1}}}} \, \LamDeltant{t'_{{\mathrm{2}}}}$. We have two cases to consider.
  \begin{itemize}
  \item[Case.] Suppose $\LamDeltant{t'_{{\mathrm{1}}}} \equiv \LamDeltamv{x}$ for some variable $\LamDeltamv{x}$.  In each
    case $\LamDeltant{B} \equiv  \perp $.
    \begin{itemize}
    \item[Case.] Suppose $\LamDeltant{t'_{{\mathrm{2}}}} \equiv  \lambda  \LamDeltamv{y} : \LamDeltant{A}  .  \LamDeltant{t''_{{\mathrm{2}}}} $, for some $\LamDeltamv{y}$ and $\LamDeltant{t''_{{\mathrm{2}}}}$,
      $ ( \LamDeltamv{x} , \LamDeltamv{z} , \LamDeltant{t} )  \in \Theta$. Since 
      $(A,1,\LamDeltant{t'}) > (A,1,\LamDeltant{t''_{{\mathrm{2}}}})$ and the typing assumptions
      hold by inversion we can apply the induction
      hypothesis to obtain $ \langle  \Theta  \rangle^{ \LamDeltant{A} }_{ \LamDeltant{A'} }  \LamDeltant{t''_{{\mathrm{2}}}}  = \LamDeltamv{s}$ for some term $\LamDeltamv{s}$ and
      $ \Gamma  \LamDeltasym{,}   \Theta ^2 :   \neg  \LamDeltant{A'}    \LamDeltasym{,}  \LamDeltamv{y}  \LamDeltasym{:}  \LamDeltant{A}  \vdash  \LamDeltant{s}  :  \LamDeltant{A'} $.
      Furthermore, sense $(A,1,t') > (A,0,s)$, the previous typing condition and
      the typing assumptions we also know from the induction hypothesis that 
      $ [  \LamDeltant{t}  /  \LamDeltamv{y}  ]^{ \LamDeltant{A} }  \LamDeltant{s}   \LamDeltasym{=}  \LamDeltant{s'}$ for some term $\LamDeltant{s'}$ and
      $ \Gamma  \LamDeltasym{,}   \Theta ^2 :   \neg  \LamDeltant{A'}    \vdash  \LamDeltant{s'}  :  \LamDeltant{A'} $. Finally, by definition we know 
      $ \langle  \Theta  \rangle^{ \LamDeltant{A} }_{ \LamDeltant{A'} }  \LamDeltant{t'}  = \LamDeltamv{z} \, \LamDeltasym{(}   [  \LamDeltant{t}  /  \LamDeltamv{y}  ]^{ \LamDeltant{A} }  \LamDeltant{s}   \LamDeltasym{)} = \LamDeltamv{z} \, \LamDeltant{s'}$ and 
      by using the application typing rule that $ \Gamma  \LamDeltasym{,}   \Theta ^2 :   \neg  \LamDeltant{A'}    \vdash  \LamDeltamv{z} \, \LamDeltant{s'}  :  \LamDeltant{B} $.

    \item[Case.] Suppose $\LamDeltant{t'_{{\mathrm{2}}}} \equiv  \Delta  \LamDeltamv{y}  :   \neg  \LamDeltasym{(}   \LamDeltant{A}  \to  \LamDeltant{A'}   \LamDeltasym{)}   .  \LamDeltant{t''_{{\mathrm{2}}}} $, for some $\LamDeltamv{y}$ and 
      $\LamDeltant{t''_{{\mathrm{2}}}}$, $ ( \LamDeltamv{x} , \LamDeltamv{z} , \LamDeltant{t} )  \in \Theta$. Since $(A,1,t') > (A,1,\LamDeltant{t''_{{\mathrm{2}}}})$ we know from the 
      induction hypothesis that $ \langle  \Theta  \LamDeltasym{,}   ( \LamDeltamv{y} , \LamDeltamv{z_{{\mathrm{2}}}} , \LamDeltant{t} )   \rangle^{ \LamDeltant{A} }_{ \LamDeltant{A'} }  \LamDeltant{t''_{{\mathrm{2}}}}  = \LamDeltamv{s}$ for some fresh variable $\LamDeltamv{z}$
      and term $\LamDeltamv{s}$, and $ \Gamma  \LamDeltasym{,}   \Theta ^2 :   \neg  \LamDeltant{A'}    \LamDeltasym{,}  \LamDeltamv{z_{{\mathrm{2}}}}  \LamDeltasym{:}   \neg  \LamDeltant{A'}   \vdash  \LamDeltant{s}  :   \perp  $.
      Finally, $ \langle  \Theta  \rangle^{ \LamDeltant{A} }_{ \LamDeltant{A'} }  \LamDeltant{t'}  = \LamDeltamv{z} \, \LamDeltasym{(}   \Delta  \LamDeltamv{z_{{\mathrm{2}}}}  :   \neg  \LamDeltant{A'}   .  \LamDeltant{s}   \LamDeltasym{)}$ by definition, and by
      using the application typing rule 
      $ \Gamma  \LamDeltasym{,}   \Theta ^2 :   \neg  \LamDeltant{A'}    \vdash  \LamDeltamv{z} \, \LamDeltasym{(}   \Delta  \LamDeltamv{z_{{\mathrm{2}}}}  :   \neg  \LamDeltant{A'}   .  \LamDeltant{s}   \LamDeltasym{)}  :  \LamDeltant{B} $.

    \item[Case.] Suppose $\LamDeltant{t'_{{\mathrm{2}}}}$ is not an abstraction, and $ ( \LamDeltamv{x} , \LamDeltamv{z} , \LamDeltant{t} )  \in \Theta$.  
      Since $(A,1,t') > (A,1,\LamDeltant{t'_{{\mathrm{2}}}})$ we know from the 
      induction hypothesis that $ \langle  \Theta  \rangle^{ \LamDeltant{A} }_{ \LamDeltant{A'} }  \LamDeltant{t'_{{\mathrm{2}}}}  = \LamDeltamv{s}$ for some term $\LamDeltamv{s}$
      and $ \Gamma  \LamDeltasym{,}   \Theta ^2 :   \neg  \LamDeltant{A'}    \vdash  \LamDeltant{s}  :   \LamDeltant{A}  \to  \LamDeltant{A'}  $.  Finally,
      $ \langle  \Theta  \rangle^{ \LamDeltant{A} }_{ \LamDeltant{A'} }  \LamDeltant{t'}  = \LamDeltamv{z} \, \LamDeltant{s}$ by definition, and by
      using the application typing rule 
      $ \Gamma  \LamDeltasym{,}   \Theta ^2 :   \neg  \LamDeltant{A'}    \vdash  \LamDeltamv{z} \, \LamDeltant{s}  :  \LamDeltant{B} $.

    \item[Case.] Suppose $ ( \LamDeltamv{x} , \LamDeltamv{z} , \LamDeltant{t''} )  \not \in \Theta$ for any term $\LamDeltant{t''}$ and $\LamDeltamv{z}$.  
      Since $(A,1,t') > (A,1,\LamDeltant{t'_{{\mathrm{2}}}})$ we know from the 
      induction hypothesis that $ \langle  \Theta  \rangle^{ \LamDeltant{A} }_{ \LamDeltant{A'} }  \LamDeltant{t'_{{\mathrm{2}}}}  = \LamDeltamv{s}$ for some term $\LamDeltamv{s}$
      and $ \Gamma  \LamDeltasym{,}   \Theta ^2 :   \neg  \LamDeltant{A'}    \vdash  \LamDeltant{s}  :   \LamDeltant{A}  \to  \LamDeltant{A'}  $.  Finally,
      $ \langle  \Theta  \rangle^{ \LamDeltant{A} }_{ \LamDeltant{A'} }  \LamDeltant{t'}  = \LamDeltamv{x} \, \LamDeltant{s}$ by definition, and by
      using the application typing rule 
      $ \Gamma  \LamDeltasym{,}   \Theta ^2 :   \neg  \LamDeltant{A'}    \vdash  \LamDeltamv{x} \, \LamDeltant{s}  :  \LamDeltant{B} $.
    \end{itemize}
  \item[Case.] Suppose $\LamDeltant{t'_{{\mathrm{1}}}}$ is not a variable.  This case follows easily from
    the induction hypothesis.
  \end{itemize}

\end{itemize}

\ \\
\noindent Part two.
\begin{itemize}
\item[Case.] Suppose $\LamDeltant{t'}$ is either $\LamDeltamv{x}$ or a variable $\LamDeltamv{y}$ distinct from $\LamDeltamv{x}$.  
  Trivial in both cases.

\item[Case.] Suppose $\LamDeltant{t'} \equiv  \lambda  \LamDeltamv{y} : \LamDeltant{A_{{\mathrm{1}}}}  .  \LamDeltant{t'_{{\mathrm{1}}}} $.  By inversion 
  we know there exists a type $\LamDeltant{A_{{\mathrm{2}}}}$ such that
  $ \Gamma  \LamDeltasym{,}  \LamDeltamv{x}  \LamDeltasym{:}  \LamDeltant{A}  \LamDeltasym{,}  \Gamma'  \LamDeltasym{,}  \LamDeltamv{y}  \LamDeltasym{:}  \LamDeltant{A_{{\mathrm{1}}}}  \vdash  \LamDeltant{t'_{{\mathrm{1}}}}  :  \LamDeltant{A_{{\mathrm{2}}}} $.
  We also know that $\LamDeltant{t'_{{\mathrm{1}}}}$ is a strict subexpression of $\LamDeltant{t'}$, hence we can apply the second part of the
  induction hypothesis to obtain
  $ [  \LamDeltant{t}  /  \LamDeltamv{x}  ]^{ \LamDeltant{A} }  \LamDeltant{t'_{{\mathrm{1}}}}  = \LamDeltant{s_{{\mathrm{1}}}}$ and $ \Gamma  \LamDeltasym{,}  \Gamma'  \LamDeltasym{,}  \LamDeltamv{y}  \LamDeltasym{:}  \LamDeltant{A_{{\mathrm{1}}}}  \vdash  \LamDeltant{s_{{\mathrm{1}}}}  :  \LamDeltant{A_{{\mathrm{2}}}} $
  for some term $\LamDeltant{s_{{\mathrm{1}}}}$.  By the definition of the hereditary substitution function 
  \begin{center}
    \begin{math}
      \begin{array}{lll}
         [  \LamDeltant{t}  /  \LamDeltamv{x}  ]^{ \LamDeltant{A} }  \LamDeltant{t'}  & = &  \lambda  \LamDeltamv{y} : \LamDeltant{A_{{\mathrm{1}}}}  .   [  \LamDeltant{t}  /  \LamDeltamv{x}  ]^{ \LamDeltant{A} }  \LamDeltant{t'_{{\mathrm{1}}}}   \\
        & = &  \lambda  \LamDeltamv{y} : \LamDeltant{A_{{\mathrm{1}}}}  .  \LamDeltant{s_{{\mathrm{1}}}} .
      \end{array}
    \end{math}
  \end{center}
  It suffices to show that $ \Gamma  \LamDeltasym{,}  \Gamma'  \vdash   \lambda  \LamDeltamv{y} : \LamDeltant{A_{{\mathrm{1}}}}  .  \LamDeltant{s_{{\mathrm{1}}}}   :   \LamDeltant{A_{{\mathrm{1}}}}  \to  \LamDeltant{A_{{\mathrm{2}}}}  $.  
  By simply applying the typing rule $\LamDeltadrulename{Lam}$ using
  $ \Gamma  \LamDeltasym{,}  \Gamma'  \LamDeltasym{,}  \LamDeltamv{y}  \LamDeltasym{:}  \LamDeltant{A_{{\mathrm{1}}}}  \vdash  \LamDeltant{s_{{\mathrm{1}}}}  :  \LamDeltant{A_{{\mathrm{2}}}} $ we obtain $ \Gamma  \LamDeltasym{,}  \Gamma'  \vdash   \lambda  \LamDeltamv{y} : \LamDeltant{A_{{\mathrm{1}}}}  .  \LamDeltant{s_{{\mathrm{1}}}}   :   \LamDeltant{A_{{\mathrm{1}}}}  \to  \LamDeltant{A_{{\mathrm{2}}}}  $.
  
\item[Case.] Suppose $\LamDeltant{t'} \equiv  \Delta  \LamDeltamv{y}  :   \neg  \LamDeltant{B}   .  \LamDeltant{t'_{{\mathrm{1}}}} $.  Similar to the previous case.

\item[Case.] Suppose $\LamDeltant{t'} \equiv \LamDeltant{t'_{{\mathrm{1}}}} \, \LamDeltant{t'_{{\mathrm{2}}}}$.  By inversion we know
  $ \Gamma  \LamDeltasym{,}  \LamDeltamv{x}  \LamDeltasym{:}  \LamDeltant{A}  \LamDeltasym{,}  \Gamma'  \vdash  \LamDeltant{t'_{{\mathrm{1}}}}  :   \LamDeltant{B'}  \to  \LamDeltant{B}  $ and
  $ \Gamma  \LamDeltasym{,}  \LamDeltamv{x}  \LamDeltasym{:}  \LamDeltant{A}  \LamDeltasym{,}  \Gamma'  \vdash  \LamDeltant{t'_{{\mathrm{2}}}}  :  \LamDeltant{B'} $ for some type $\LamDeltant{B'}$.
  Clearly, $\LamDeltant{t'_{{\mathrm{1}}}}$ and $\LamDeltant{t'_{{\mathrm{2}}}}$ are strict subexpressions of $\LamDeltant{t'}$.  Thus, by the second part of the
  induction hypothesis there exists terms $\LamDeltant{s_{{\mathrm{1}}}}$ and $\LamDeltant{s_{{\mathrm{2}}}}$ such that $ [  \LamDeltant{t}  /  \LamDeltamv{x}  ]^{ \LamDeltant{A} }  \LamDeltant{t'_{{\mathrm{1}}}}  = \LamDeltant{s_{{\mathrm{1}}}}$ and
  $ [  \LamDeltant{t}  /  \LamDeltamv{x}  ]^{ \LamDeltant{A} }  \LamDeltant{t'_{{\mathrm{2}}}}  = \LamDeltant{s_{{\mathrm{2}}}}$, and $ \Gamma  \LamDeltasym{,}  \Gamma'  \vdash  \LamDeltant{s_{{\mathrm{1}}}}  :   \LamDeltant{B'}  \to  \LamDeltant{B'}  $ and
  $ \Gamma  \LamDeltasym{,}  \Gamma'  \vdash  \LamDeltant{s_{{\mathrm{2}}}}  :  \LamDeltant{B'} $.  We case split on whether or not $\LamDeltant{s_{{\mathrm{1}}}}$ is a $\lambda$-abstraction or
  a $\Delta$-abstraction and $\LamDeltant{t'_{{\mathrm{1}}}}$ is not, or $\LamDeltant{s_{{\mathrm{1}}}}$ and $\LamDeltant{t'_{{\mathrm{1}}}}$ are both a $\lambda$-abstraction or
  a $\Delta$-abstraction.
  We only consider the non-trivial cases when $\LamDeltant{s_{{\mathrm{1}}}} \equiv  \lambda  \LamDeltamv{y} : \LamDeltant{B'}  .  \LamDeltant{s'_{{\mathrm{1}}}} $ and $\LamDeltant{t'_{{\mathrm{1}}}}$ is not a 
  $\lambda$-abstraction, and $\LamDeltant{s_{{\mathrm{1}}}} \equiv  \Delta  \LamDeltamv{y}  :   \neg  \LamDeltasym{(}   \LamDeltant{B'}  \to  \LamDeltant{B}   \LamDeltasym{)}   .  \LamDeltant{s'_{{\mathrm{1}}}} $ and $\LamDeltant{t'_{{\mathrm{1}}}}$ is not a 
  $\Delta$-abstraction.  Consider the former.

  \ \\
  Now by Lemma~\ref{lemma:ctype_props} it is the case that 
  there exists a $\LamDeltant{B''}$ such that $ \textsf{ctype}_{ \LamDeltant{A} }( \LamDeltamv{x} , \LamDeltant{t'_{{\mathrm{1}}}} )  = \LamDeltant{B''}$, 
  $\LamDeltant{B''} \equiv  \LamDeltant{B'}  \to  \LamDeltant{B} $, and $\LamDeltant{B}$ is a subexpression of $\LamDeltant{A}$, hence
  $\LamDeltant{A} > \LamDeltant{B'}$.  By the definition of the hereditary substitution function
  $ [  \LamDeltant{t}  /  \LamDeltamv{x}  ]^{ \LamDeltant{A} }  \LamDeltasym{(}  \LamDeltant{t'_{{\mathrm{1}}}} \, \LamDeltant{t'_{{\mathrm{2}}}}  \LamDeltasym{)}  =  [  \LamDeltant{s_{{\mathrm{2}}}}  /  \LamDeltamv{y}  ]^{ \LamDeltant{B'} }  \LamDeltant{s'_{{\mathrm{1}}}} $. Therefore, by the induction hypothesis there exists a 
  term $\LamDeltamv{s}$ such that $ [  \LamDeltant{s_{{\mathrm{2}}}}  /  \LamDeltamv{y}  ]^{ \LamDeltant{A} }  \LamDeltant{s'_{{\mathrm{1}}}}  = \LamDeltamv{s}$ and $ \Gamma  \LamDeltasym{,}  \Gamma'  \vdash  \LamDeltant{s}  :  \LamDeltant{B} $.

  \ \\
  At this point consider when $\LamDeltant{s_{{\mathrm{1}}}} \equiv  \Delta  \LamDeltamv{y}  :   \neg  \LamDeltasym{(}   \LamDeltant{B'}  \to  \LamDeltant{B}   \LamDeltasym{)}   .  \LamDeltant{s'_{{\mathrm{1}}}} $ and $\LamDeltant{t'_{{\mathrm{1}}}}$ is not a 
  $\Delta$-abstraction.  Again, by Lemma~\ref{lemma:ctype_props} it is the case that 
  there exists a $\LamDeltant{B''}$ such that $ \textsf{ctype}_{ \LamDeltant{A} }( \LamDeltamv{x} , \LamDeltant{t'_{{\mathrm{1}}}} )  = \LamDeltant{B''}$, $\LamDeltant{B''} \equiv  \LamDeltant{B'}  \to  \LamDeltant{B} $ and
  $ \LamDeltant{B'}  \to  \LamDeltant{B} $ is a subexpression of $\LamDeltant{A}$.  Hence, $\LamDeltant{A} > \LamDeltant{B'}$. Let $\LamDeltamv{r}$ be a fresh variable of type $ \neg  \LamDeltant{B} $.    
  Then by the induction hypothesis, there exists a term $s''$, such that, $ \langle   ( \LamDeltamv{y} , \LamDeltamv{r} , \LamDeltant{s_{{\mathrm{2}}}} )   \rangle^{ \LamDeltant{B'} }_{ \LamDeltant{B} }  \LamDeltant{s'_{{\mathrm{1}}}}  = s''$ and 
  $ \Gamma  \LamDeltasym{,}  \LamDeltamv{r}  \LamDeltasym{:}   \neg  \LamDeltant{B}   \vdash  \LamDeltant{s''}  :   \perp  $.  Therefore, $ [  \LamDeltant{t}  /  \LamDeltamv{x}  ]^{ \LamDeltant{A} }  \LamDeltasym{(}  \LamDeltant{t'_{{\mathrm{1}}}} \, \LamDeltant{t'_{{\mathrm{2}}}}  \LamDeltasym{)}  =  \Delta  \LamDeltamv{r}  :   \neg  \LamDeltant{B}   .   \langle   ( \LamDeltamv{y} , \LamDeltamv{r} , \LamDeltant{s_{{\mathrm{2}}}} )   \rangle^{ \LamDeltant{B'} }_{ \LamDeltant{B} }  \LamDeltant{s'_{{\mathrm{1}}}}   =  \Delta  \LamDeltamv{r}  :   \neg  \LamDeltant{B}   .  \LamDeltant{s''} $,
  and by the $\Delta$-abstraction typing rule $ \Gamma  \vdash   \Delta  \LamDeltamv{r}  :   \neg  \LamDeltant{B}   .  \LamDeltant{s''}   :  \LamDeltant{B} $.
\end{itemize}
% subsection proof_of_totality_and_type_preservation (end)

\subsection{Proof of Lemma~\ref{lemma:ssub_var_head}}
\label{subsec:proof_of_ssub_var_head}
This is a proof by induction on the form of $\LamDeltant{n_{{\mathrm{1}}}} \, \LamDeltant{n_{{\mathrm{2}}}}$.
In every case where $\LamDeltant{n_{{\mathrm{1}}}}$ is a variable and $(\LamDeltant{n_{{\mathrm{1}}}},\LamDeltamv{z},\LamDeltant{t}) \in \Theta$ for some 
term $\LamDeltant{t}$ and variable $\LamDeltamv{z}$, we know by definition
that $ \langle  \Theta  \rangle^{ \LamDeltant{A} }_{ \LamDeltant{A'} }  \LamDeltasym{(}  \LamDeltant{n_{{\mathrm{1}}}} \, \LamDeltant{n_{{\mathrm{2}}}}  \LamDeltasym{)}  = \LamDeltamv{z} \, \LamDeltant{t_{{\mathrm{2}}}}$ for some variable $\LamDeltamv{z}$ 
and term $\LamDeltant{t_{{\mathrm{2}}}}$.  In the case where $\LamDeltant{n_{{\mathrm{1}}}}$ is a variable and 
$(\LamDeltant{n_{{\mathrm{1}}}},\LamDeltamv{z},\LamDeltant{t}) \not\in \Theta$ for some 
term $\LamDeltant{t}$ and variable $\LamDeltamv{z}$, we know by definition
that $ \langle  \Theta  \rangle^{ \LamDeltant{A} }_{ \LamDeltant{A'} }  \LamDeltasym{(}  \LamDeltant{n_{{\mathrm{1}}}} \, \LamDeltant{n_{{\mathrm{2}}}}  \LamDeltasym{)}  = \LamDeltasym{(}   \langle  \Theta  \rangle^{ \LamDeltant{A} }_{ \LamDeltant{A'} }  \LamDeltant{n_{{\mathrm{1}}}}   \LamDeltasym{)} \, \LamDeltasym{(}   \langle  \Theta  \rangle^{ \LamDeltant{A} }_{ \LamDeltant{A'} }  \LamDeltant{n_{{\mathrm{2}}}}   \LamDeltasym{)}$.  Now
by hypothesis and definition $ \langle  \Theta  \rangle^{ \LamDeltant{A} }_{ \LamDeltant{A'} }  \LamDeltant{n_{{\mathrm{1}}}}  = \LamDeltant{n_{{\mathrm{1}}}}$.  Thus,
$\LamDeltasym{(}   \langle  \Theta  \rangle^{ \LamDeltant{A} }_{ \LamDeltant{A'} }  \LamDeltant{n_{{\mathrm{1}}}}   \LamDeltasym{)} \, \LamDeltasym{(}   \langle  \Theta  \rangle^{ \LamDeltant{A} }_{ \LamDeltant{A'} }  \LamDeltant{n_{{\mathrm{2}}}}   \LamDeltasym{)} = \LamDeltant{n_{{\mathrm{1}}}} \, \LamDeltasym{(}   \langle  \Theta  \rangle^{ \LamDeltant{A} }_{ \LamDeltant{A'} }  \LamDeltant{n_{{\mathrm{2}}}}   \LamDeltasym{)}$ and we know $\LamDeltant{n_{{\mathrm{1}}}}$ is
a variable.  The final case is when $\LamDeltant{n_{{\mathrm{1}}}}$ is not a variable.  Then it must
be the case that $\LamDeltant{n_{{\mathrm{1}}}}$ is a normal application.  So by the induction hypothesis
$ \textsf{head}(  \langle  \Theta  \rangle^{ \LamDeltant{A} }_{ \LamDeltant{A'} }  \LamDeltant{n_{{\mathrm{1}}}}  ) $ is a variable.  Therefore,
$ \textsf{head}(  \langle  \Theta  \rangle^{ \LamDeltant{A} }_{ \LamDeltant{A'} }  \LamDeltasym{(}  \LamDeltant{n_{{\mathrm{1}}}} \, \LamDeltant{n_{{\mathrm{2}}}}  \LamDeltasym{)}  )  =  \textsf{head}( \LamDeltasym{(}   \langle  \Theta  \rangle^{ \LamDeltant{A} }_{ \LamDeltant{A'} }  \LamDeltant{n_{{\mathrm{1}}}}   \LamDeltasym{)} \, \LamDeltasym{(}   \langle  \Theta  \rangle^{ \LamDeltant{A} }_{ \LamDeltant{A'} }  \LamDeltant{n_{{\mathrm{2}}}}   \LamDeltasym{)} )  =  \textsf{head}(  \langle  \Theta  \rangle^{ \LamDeltant{A} }_{ \LamDeltant{A'} }  \LamDeltant{n_{{\mathrm{1}}}}  ) $ is a variable.
% subsection proof_of_ssub_var_head (end)

\subsection{Proof of Normality Preservation}
\label{subsec:proof_of_normality_preservation}
This is a mutually inductive proof using the lexicographic combination
$(\LamDeltant{A}, f,\LamDeltant{n'})$ of our ordering on types,
the natural number ordering where $f \in \{0,1\}$, and
the strict subexpression ordering on terms. We first prove part one
and then part two.  In both parts we case split on $\LamDeltant{n'}$.

\ \\
\noindent Part One.
\begin{itemize}
\item[Case.] Suppose $\LamDeltant{n'}$ is a variable $\LamDeltamv{x}$.  Then either there exists
  a normal form $\LamDeltant{m}$ and variable $\LamDeltamv{z}$, such that, $ ( \LamDeltamv{x} , \LamDeltamv{z} , \LamDeltant{m} )  \in \Theta$ or not.  
  Suppose so. Then 
  $ \langle  \Theta  \rangle^{ \LamDeltant{A} }_{ \LamDeltant{A'} }  \LamDeltamv{x}  =  \lambda  \LamDeltamv{y} :  \LamDeltant{A}  \to  \LamDeltant{A'}   .  \LamDeltasym{(}  \LamDeltamv{z} \, \LamDeltasym{(}  \LamDeltamv{y} \, \LamDeltant{m}  \LamDeltasym{)}  \LamDeltasym{)} $ where $\LamDeltamv{y}$ is fresh in $\LamDeltamv{x}$, $\LamDeltamv{z}$ and 
  $\LamDeltant{m}$.  Clearly, $ \lambda  \LamDeltamv{y} :  \LamDeltant{A}  \to  \LamDeltant{A'}   .  \LamDeltasym{(}  \LamDeltamv{z} \, \LamDeltasym{(}  \LamDeltamv{y} \, \LamDeltant{m}  \LamDeltasym{)}  \LamDeltasym{)} $ is normal.
  Now suppose there does not exist any term $\LamDeltant{m}$ or $\LamDeltamv{z}$ such that 
  $ ( \LamDeltamv{x} , \LamDeltamv{z} , \LamDeltant{m} )  \in \Theta$.  Then $ \langle  \Theta  \rangle^{ \LamDeltant{A} }_{ \LamDeltant{A'} }  \LamDeltamv{x}  = \LamDeltamv{x}$ which is clearly normal.

\item[Case.] Suppose $\LamDeltant{n'} \equiv  \lambda  \LamDeltamv{y} : \LamDeltant{B_{{\mathrm{1}}}}  .  \LamDeltant{n'_{{\mathrm{1}}}} $.  Then 
  $ \langle  \Theta  \rangle^{ \LamDeltant{A} }_{ \LamDeltant{A'} }  \LamDeltant{n'}  =  \langle  \Theta  \rangle^{ \LamDeltant{A} }_{ \LamDeltant{A'} }  \LamDeltasym{(}   \lambda  \LamDeltamv{y} : \LamDeltant{B_{{\mathrm{1}}}}  .  \LamDeltant{n'_{{\mathrm{1}}}}   \LamDeltasym{)}  =  \lambda  \LamDeltamv{y} : \LamDeltant{B_{{\mathrm{1}}}}  .   \langle  \Theta  \rangle^{ \LamDeltant{A} }_{ \LamDeltant{A'} }  \LamDeltant{n'_{{\mathrm{1}}}}  $.  Now sense
  $(\LamDeltant{A},1,\LamDeltant{n'}) > (\LamDeltant{A},1,\LamDeltant{n'_{{\mathrm{1}}}})$ we may apply the induction hypothesis to obtain
  that there exists a term $\LamDeltant{m}$ such that $ \langle  \Theta  \rangle^{ \LamDeltant{A} }_{ \LamDeltant{A'} }  \LamDeltant{n'_{{\mathrm{1}}}}  = \LamDeltant{m}$.  
  Thus, by definition we obtain $ \langle  \Theta  \rangle^{ \LamDeltant{A} }_{ \LamDeltant{A'} }  \LamDeltant{n'}  =  \lambda  \LamDeltamv{y} : \LamDeltant{B_{{\mathrm{1}}}}  .  \LamDeltant{m} $.

\item[Case.] Suppose $\LamDeltant{n'} \equiv  \Delta  \LamDeltamv{y}  :   \neg  \LamDeltant{B}   .  \LamDeltant{n'_{{\mathrm{1}}}} $. Similar to the previous case.

\item[Case.] Suppose $\LamDeltant{n'} \equiv \LamDeltant{n'_{{\mathrm{1}}}} \, \LamDeltant{n'_{{\mathrm{2}}}}$. We have two cases to consider.
  \begin{itemize}
  \item[Case.] Suppose $\LamDeltant{n'_{{\mathrm{1}}}} \equiv \LamDeltamv{x}$ for some variable $\LamDeltamv{x}$.     
    \begin{itemize}
    \item[Case.] Suppose $\LamDeltant{n'_{{\mathrm{2}}}} \equiv  \lambda  \LamDeltamv{y} : \LamDeltant{A}  .  \LamDeltant{n''_{{\mathrm{2}}}} $, for some $\LamDeltamv{y}$ and $\LamDeltant{n''_{{\mathrm{2}}}}$,
      $ ( \LamDeltamv{x} , \LamDeltamv{z} , \LamDeltant{n} )  \in \Theta$. Since 
      $(A,1,\LamDeltant{n'}) > (A,1,\LamDeltant{n''_{{\mathrm{2}}}})$ and the typing assumptions
      hold by inversion we can apply the induction
      hypothesis to obtain $ \langle  \Theta  \rangle^{ \LamDeltant{A} }_{ \LamDeltant{A'} }  \LamDeltant{n''_{{\mathrm{2}}}}  = \LamDeltant{m}$ for some term $\LamDeltant{m}$.
      We know from Lemma~\ref{lemma:normality_preservation} that 
      $ \Gamma  \LamDeltasym{,}   \Theta ^2 :   \neg  \LamDeltant{A'}    \LamDeltasym{,}  \LamDeltamv{y}  \LamDeltasym{:}  \LamDeltant{A}  \vdash  \LamDeltant{m}  :  \LamDeltant{A'} $.
      Furthermore, sense $(A,1,n') > (A,0,m)$, the previous typing condition and
      the typing assumptions we also know from the induction hypothesis that 
      $ [  \LamDeltant{t}  /  \LamDeltamv{y}  ]^{ \LamDeltant{A} }  \LamDeltant{m}   \LamDeltasym{=}  \LamDeltant{m'}$ for some term $\LamDeltant{m'}$. Finally, by definition we know 
      $ \langle  \Theta  \rangle^{ \LamDeltant{A} }_{ \LamDeltant{A'} }  \LamDeltant{n'}  = \LamDeltamv{z} \, \LamDeltasym{(}   [  \LamDeltant{n}  /  \LamDeltamv{y}  ]^{ \LamDeltant{A} }  \LamDeltant{m}   \LamDeltasym{)} = \LamDeltamv{z} \, \LamDeltant{m'}$.  It is easy to see that
      $\LamDeltamv{z} \, \LamDeltant{m'}$ is normal.      

    \item[Case.] Suppose $\LamDeltant{n'_{{\mathrm{2}}}} \equiv  \Delta  \LamDeltamv{y}  :   \neg  \LamDeltasym{(}   \LamDeltant{A}  \to  \LamDeltant{A'}   \LamDeltasym{)}   .  \LamDeltant{n''_{{\mathrm{2}}}} $, for some $\LamDeltamv{y}$ and 
      $\LamDeltant{n''_{{\mathrm{2}}}}$, $ ( \LamDeltamv{x} , \LamDeltamv{z} , \LamDeltant{n} )  \in \Theta$. Since $(A,1,n') > (A,1,\LamDeltant{n''_{{\mathrm{2}}}})$ we know from the 
      induction hypothesis that $ \langle  \Theta  \LamDeltasym{,}   ( \LamDeltamv{y} , \LamDeltamv{z_{{\mathrm{2}}}} , \LamDeltant{n} )   \rangle^{ \LamDeltant{A} }_{ \LamDeltant{A'} }  \LamDeltant{n''_{{\mathrm{2}}}}  = \LamDeltant{m}$ for some normal form 
      $\LamDeltant{m}$, and $ \Gamma  \LamDeltasym{,}   \Theta ^2 :   \neg  \LamDeltant{A'}    \LamDeltasym{,}  \LamDeltamv{z_{{\mathrm{2}}}}  \LamDeltasym{:}   \neg  \LamDeltant{A'}   \vdash  \LamDeltant{m}  :   \perp  $.
      Finally, $ \langle  \Theta  \rangle^{ \LamDeltant{A} }_{ \LamDeltant{A'} }  \LamDeltant{n'}  = \LamDeltamv{z} \, \LamDeltasym{(}   \Delta  \LamDeltamv{z_{{\mathrm{2}}}}  :   \neg  \LamDeltant{A'}   .  \LamDeltant{m}   \LamDeltasym{)}$ by definition.

    \item[Case.] Suppose $\LamDeltant{n'_{{\mathrm{2}}}}$ is not an abstraction, and $ ( \LamDeltamv{x} , \LamDeltamv{z} , \LamDeltant{n} )  \in \Theta$.  
      Since $(A,1,n') > (A,1,\LamDeltant{n'_{{\mathrm{2}}}})$ we know from the 
      induction hypothesis that $ \langle  \Theta  \rangle^{ \LamDeltant{A} }_{ \LamDeltant{A'} }  \LamDeltant{n'_{{\mathrm{2}}}}  = \LamDeltant{m}$ for some normal form $\LamDeltant{m}$.  Finally,
      $ \langle  \Theta  \rangle^{ \LamDeltant{A} }_{ \LamDeltant{A'} }  \LamDeltant{n'}  = \LamDeltamv{z} \, \LamDeltant{m}$ by definition.      

    \item[Case.] Suppose $ ( \LamDeltamv{x} , \LamDeltamv{z} , \LamDeltant{n''} )  \not \in \Theta$ for any term $\LamDeltant{n''}$ and $\LamDeltamv{z}$.  
      Since $(A,1,n') > (A,1,\LamDeltant{n'_{{\mathrm{2}}}})$ we know from the 
      induction hypothesis that $ \langle  \Theta  \rangle^{ \LamDeltant{A} }_{ \LamDeltant{A'} }  \LamDeltant{n'_{{\mathrm{2}}}}  = \LamDeltant{m}$ for some term $\LamDeltant{m}$.  Finally,
      $ \langle  \Theta  \rangle^{ \LamDeltant{A} }_{ \LamDeltant{A'} }  \LamDeltant{n'}  = \LamDeltamv{x} \, \LamDeltant{m}$ by definition.      
    \end{itemize}
  \item[Case.] Suppose $\LamDeltant{n'_{{\mathrm{1}}}}$ is not a variable.  This case follows easily from
    the induction hypothesis and Lemma~\ref{lemma:ssub_var_head}.
  \end{itemize}
\end{itemize}

\ \\
\noindent Part two.
\begin{itemize}
\item[Case.] Suppose $\LamDeltant{n'}$ is either $\LamDeltamv{x}$ or a variable $\LamDeltamv{y}$ distinct from $\LamDeltamv{x}$.  
  Trivial in both cases.

\item[Case.] Suppose $\LamDeltant{n'} \equiv  \lambda  \LamDeltamv{y} : \LamDeltant{B_{{\mathrm{1}}}}  .  \LamDeltant{n'_{{\mathrm{1}}}} $.  
  We also know that $\LamDeltant{n'_{{\mathrm{1}}}}$ is a strict subexpression of $\LamDeltant{n'}$, hence we can apply the second part of the
  induction hypothesis to obtain
  $ [  \LamDeltant{n}  /  \LamDeltamv{x}  ]^{ \LamDeltant{A} }  \LamDeltant{n'_{{\mathrm{1}}}}  = \LamDeltant{m_{{\mathrm{1}}}}$ for some normal form $\LamDeltant{m_{{\mathrm{1}}}}$.  
  By the definition of the hereditary substitution function 
  \begin{center}
    \begin{math}
      \begin{array}{lll}
         [  \LamDeltant{n}  /  \LamDeltamv{x}  ]^{ \LamDeltant{A} }  \LamDeltant{n'}  & = &  \lambda  \LamDeltamv{y} : \LamDeltant{A_{{\mathrm{1}}}}  .   [  \LamDeltant{n}  /  \LamDeltamv{x}  ]^{ \LamDeltant{A} }  \LamDeltant{n'_{{\mathrm{1}}}}   \\
        & = &  \lambda  \LamDeltamv{y} : \LamDeltant{B_{{\mathrm{1}}}}  .  \LamDeltant{m_{{\mathrm{1}}}} .
      \end{array}
    \end{math}
  \end{center}
  Clearly, $ \lambda  \LamDeltamv{y} : \LamDeltant{B_{{\mathrm{1}}}}  .  \LamDeltant{m_{{\mathrm{1}}}} $ is normal.
  
\item[Case.] Suppose $\LamDeltant{n'} \equiv  \Delta  \LamDeltamv{y}  :   \neg  \LamDeltant{B}   .  \LamDeltant{n'_{{\mathrm{1}}}} $.  Similar to the previous case.

\item[Case.] Suppose $\LamDeltant{t'} \equiv \LamDeltant{t'_{{\mathrm{1}}}} \, \LamDeltant{t'_{{\mathrm{2}}}}$.
  Clearly, $\LamDeltant{n'_{{\mathrm{1}}}}$ and $\LamDeltant{n'_{{\mathrm{2}}}}$ are strict subexpressions of $\LamDeltant{n'}$.  Thus, by the 
  induction hypothesis there exists normal forms $\LamDeltant{n_{{\mathrm{1}}}}$ and $\LamDeltant{n_{{\mathrm{2}}}}$ such that 
  $ [  \LamDeltant{n}  /  \LamDeltamv{x}  ]^{ \LamDeltant{A} }  \LamDeltant{n'_{{\mathrm{1}}}}  = \LamDeltant{m_{{\mathrm{1}}}}$ and
  $ [  \LamDeltant{n}  /  \LamDeltamv{x}  ]^{ \LamDeltant{A} }  \LamDeltant{n'_{{\mathrm{2}}}}  = \LamDeltant{m_{{\mathrm{2}}}}$.
  We case split on whether or not $\LamDeltant{m_{{\mathrm{1}}}}$ is a $\lambda$-abstraction or
  a $\Delta$-abstraction and $\LamDeltant{n'_{{\mathrm{1}}}}$ is not, or $\LamDeltant{m_{{\mathrm{1}}}}$ and $\LamDeltant{n'_{{\mathrm{1}}}}$ are both a 
  $\lambda$-abstraction or a $\Delta$-abstraction.  
  We only consider the non-trivial cases when $\LamDeltant{m_{{\mathrm{1}}}} \equiv  \lambda  \LamDeltamv{y} : \LamDeltant{B'}  .  \LamDeltant{m'_{{\mathrm{1}}}} $ and 
  $\LamDeltant{n'_{{\mathrm{1}}}}$ is not a 
  $\lambda$-abstraction, and $\LamDeltant{m_{{\mathrm{1}}}} \equiv  \Delta  \LamDeltamv{y}  :   \neg  \LamDeltasym{(}   \LamDeltant{B'}  \to  \LamDeltant{B}   \LamDeltasym{)}   .  \LamDeltant{m'_{{\mathrm{1}}}} $ and $\LamDeltant{n'_{{\mathrm{1}}}}$ is not a 
  $\Delta$-abstraction.  Consider the former.

  \ \\
  Now by Lemma~\ref{lemma:ctype_props} it is the case that 
  there exists a $\LamDeltant{B''}$ such that $ \textsf{ctype}_{ \LamDeltant{A} }( \LamDeltamv{x} , \LamDeltant{t'_{{\mathrm{1}}}} )  = \LamDeltant{B''}$, 
  $\LamDeltant{B''} \equiv  \LamDeltant{B'}  \to  \LamDeltant{B} $, and $\LamDeltant{B}$ is a subexpression of $\LamDeltant{A}$, hence
  $\LamDeltant{A} > \LamDeltant{B'}$.  By the definition of the hereditary substitution function
  $ [  \LamDeltant{n}  /  \LamDeltamv{x}  ]^{ \LamDeltant{A} }  \LamDeltasym{(}  \LamDeltant{n'_{{\mathrm{1}}}} \, \LamDeltant{n'_{{\mathrm{2}}}}  \LamDeltasym{)}  =  [  \LamDeltant{m_{{\mathrm{2}}}}  /  \LamDeltamv{y}  ]^{ \LamDeltant{B'} }  \LamDeltant{m'_{{\mathrm{1}}}} $. Therefore, by the induction hypothesis there 
  exists a 
  normal form $\LamDeltant{m}$ such that $ [  \LamDeltant{m_{{\mathrm{2}}}}  /  \LamDeltamv{y}  ]^{ \LamDeltant{A} }  \LamDeltant{m'_{{\mathrm{1}}}}  = \LamDeltant{m}$.

  \ \\
  At this point consider when $\LamDeltant{m_{{\mathrm{1}}}} \equiv  \Delta  \LamDeltamv{y}  :   \neg  \LamDeltasym{(}   \LamDeltant{B'}  \to  \LamDeltant{B}   \LamDeltasym{)}   .  \LamDeltant{m'_{{\mathrm{1}}}} $ and $\LamDeltant{n'_{{\mathrm{1}}}}$ is not a 
  $\Delta$-abstraction.  Again, by Lemma~\ref{lemma:ctype_props} it is the case that 
  there exists a $\LamDeltant{B''}$ such that $ \textsf{ctype}_{ \LamDeltant{A} }( \LamDeltamv{x} , \LamDeltant{t'_{{\mathrm{1}}}} )  = \LamDeltant{B''}$, $\LamDeltant{B''} \equiv  \LamDeltant{B'}  \to  \LamDeltant{B} $ and
  $ \LamDeltant{B'}  \to  \LamDeltant{B} $ is a subexpression of $\LamDeltant{A}$.  Hence, $\LamDeltant{A} > \LamDeltant{B'}$. Let $\LamDeltamv{r}$ be a fresh variable of type $ \neg  \LamDeltant{B} $.    
  Then by the induction hypothesis, there exists a term $m''$, such that, $ \langle   ( \LamDeltamv{y} , \LamDeltamv{r} , \LamDeltant{m_{{\mathrm{2}}}} )   \rangle^{ \LamDeltant{B'} }_{ \LamDeltant{B} }  \LamDeltant{m'_{{\mathrm{1}}}}  = m''$ and 
  Therefore, $ [  \LamDeltant{n}  /  \LamDeltamv{x}  ]^{ \LamDeltant{A} }  \LamDeltasym{(}  \LamDeltant{n'_{{\mathrm{1}}}} \, \LamDeltant{n'_{{\mathrm{2}}}}  \LamDeltasym{)}  =  \Delta  \LamDeltamv{r}  :   \neg  \LamDeltant{B}   .   \langle   ( \LamDeltamv{y} , \LamDeltamv{r} , \LamDeltant{m_{{\mathrm{2}}}} )   \rangle^{ \LamDeltant{B'} }_{ \LamDeltant{B} }  \LamDeltant{m'_{{\mathrm{1}}}}   =  \Delta  \LamDeltamv{r}  :   \neg  \LamDeltant{B}   .  \LamDeltant{m''} $.
\end{itemize}
% subsection proof_of_normality_preservation (end)

\subsection{Proof of Soundness with Respect to Reduction}
\label{subsec:soundness_with_respect_to_reduction}
This is a mutually inductive proof using the lexicographic combination
$(\LamDeltant{A}, f,\LamDeltant{t'})$ of our ordering on types,
the natural number ordering where $f \in \{0,1\}$, and
the strict subexpression ordering on terms. We first prove part one
and then part two.  In both parts we case split on $\LamDeltant{t'}$.

\ \\
\noindent Part One.
\begin{itemize}
\item[Case.] Suppose $\LamDeltant{t'}$ is a variable $\LamDeltamv{x}$.  Then either there exists
  a term $\LamDeltant{a}$ such that $ ( \LamDeltamv{x} , \LamDeltamv{z} , \LamDeltant{a} )  \in \Theta$ or not.  Suppose so. 
  Then by definition we know $ \langle  \Theta  \rangle^{\uparrow^{ \LamDeltant{A} }_{ \LamDeltant{A'} } }\, \LamDeltamv{x}  =  \lambda  \LamDeltamv{y} :  \LamDeltant{A}  \to  \LamDeltant{A'}   .  \LamDeltasym{(}  \LamDeltamv{z} \, \LamDeltasym{(}  \LamDeltamv{y} \, \LamDeltant{a}  \LamDeltasym{)}  \LamDeltasym{)} $,
  for some fresh variable $\LamDeltamv{y}$.  Now $ \langle  \Theta  \rangle^{ \LamDeltant{A} }_{ \LamDeltant{A'} }  \LamDeltamv{x}  =  \lambda  \LamDeltamv{y} :  \LamDeltant{A}  \to  \LamDeltant{A'}   .  \LamDeltasym{(}  \LamDeltamv{z} \, \LamDeltasym{(}  \LamDeltamv{y} \, \LamDeltant{a}  \LamDeltasym{)}  \LamDeltasym{)} $,
  where we choose the same $\LamDeltamv{y}$.  Thus, $ \langle  \Theta  \rangle^{\uparrow^{ \LamDeltant{A} }_{ \LamDeltant{A'} } }\, \LamDeltamv{x}   \redto^*   \langle  \Theta  \rangle^{ \LamDeltant{A} }_{ \LamDeltant{A'} }  \LamDeltamv{x} $.
  Now suppose there does not exist any term $\LamDeltant{a}$ or $\LamDeltamv{z}$ such that $ ( \LamDeltamv{x} , \LamDeltamv{z} , \LamDeltant{a} )  \in \Theta$.
  Then $ \langle  \Theta  \rangle^{ \LamDeltant{A} }_{ \LamDeltant{A'} }  \LamDeltamv{x}  =  \langle  \Theta  \rangle^{\uparrow^{ \LamDeltant{A} }_{ \LamDeltant{A'} } }\, \LamDeltamv{x}  = x$. Thus, $ \langle  \Theta  \rangle^{\uparrow^{ \LamDeltant{A} }_{ \LamDeltant{A'} } }\, \LamDeltamv{x}   \redto^*   \langle  \Theta  \rangle^{ \LamDeltant{A} }_{ \LamDeltant{A'} }  \LamDeltamv{x} $.

\item[Case.] Suppose $\LamDeltant{t'} \equiv  \lambda  \LamDeltamv{y} : \LamDeltant{B_{{\mathrm{1}}}}  .  \LamDeltant{t'_{{\mathrm{1}}}} $. This case follows from the
  induction hypothesis.

\item[Case.] Suppose $\LamDeltant{t'} \equiv  \Delta  \LamDeltamv{y}  :   \neg  \LamDeltant{B}   .  \LamDeltant{t'_{{\mathrm{1}}}} $. Similar to the previous case.

\item[Case.] Suppose $\LamDeltant{t'} \equiv \LamDeltant{t'_{{\mathrm{1}}}} \, \LamDeltant{t'_{{\mathrm{2}}}}$. We have two cases to consider.
  \begin{itemize}
  \item[Case.] Suppose $\LamDeltant{t'_{{\mathrm{1}}}} \equiv \LamDeltamv{x}$ for some variable $\LamDeltamv{x}$. 
    \begin{itemize}
    \item[Case.] Suppose $\LamDeltant{t'_{{\mathrm{2}}}} \equiv  \lambda  \LamDeltamv{y} : \LamDeltant{A}  .  \LamDeltant{t''_{{\mathrm{2}}}} $, for some $\LamDeltamv{y}$ and $\LamDeltant{t''_{{\mathrm{2}}}}$,
      $ ( \LamDeltamv{x} , \LamDeltamv{z} , \LamDeltant{t} )  \in \Theta$. 
      Now
      \begin{center}
        \begin{math}
          \begin{array}{lll}
             \langle  \Theta  \rangle^{\uparrow^{ \LamDeltant{A} }_{ \LamDeltant{A'} } }\, \LamDeltasym{(}  \LamDeltamv{x} \, \LamDeltasym{(}   \lambda  \LamDeltamv{y} : \LamDeltant{A}  .  \LamDeltant{t''_{{\mathrm{2}}}}   \LamDeltasym{)}  \LamDeltasym{)}  & = & \LamDeltasym{(}   \lambda  \LamDeltamv{y} :  \LamDeltant{A}  \to  \LamDeltant{A'}   .  \LamDeltasym{(}  \LamDeltamv{z} \, \LamDeltasym{(}  \LamDeltamv{y} \, \LamDeltant{t}  \LamDeltasym{)}  \LamDeltasym{)}   \LamDeltasym{)} \, \LamDeltasym{(}   \lambda  \LamDeltamv{y} : \LamDeltant{A}  .  \LamDeltasym{(}   \langle  \Theta  \rangle^{\uparrow^{ \LamDeltant{A} }_{ \LamDeltant{A'} } }\, \LamDeltant{t''_{{\mathrm{2}}}}   \LamDeltasym{)}   \LamDeltasym{)}\\
                                             & \redto & \LamDeltamv{z} \, \LamDeltasym{(}  \LamDeltasym{(}   \lambda  \LamDeltamv{y} : \LamDeltant{A}  .  \LamDeltasym{(}   \langle  \Theta  \rangle^{\uparrow^{ \LamDeltant{A} }_{ \LamDeltant{A'} } }\, \LamDeltant{t''_{{\mathrm{2}}}}   \LamDeltasym{)}   \LamDeltasym{)} \, \LamDeltant{t}  \LamDeltasym{)}\\
                                             & \redto & \LamDeltamv{z} \, \LamDeltasym{(}  \LamDeltasym{[}  \LamDeltant{t}  \LamDeltasym{/}  \LamDeltamv{y}  \LamDeltasym{]}  \LamDeltasym{(}   \langle  \Theta  \rangle^{\uparrow^{ \LamDeltant{A} }_{ \LamDeltant{A'} } }\, \LamDeltant{t''_{{\mathrm{2}}}}   \LamDeltasym{)}  \LamDeltasym{)}\\
          \end{array}
        \end{math}
      \end{center}
      Since 
      $(A,1,\LamDeltant{t'}) > (A,1,\LamDeltant{t''_{{\mathrm{2}}}})$ we can apply the induction
      hypothesis to obtain $ \langle  \Theta  \rangle^{\uparrow^{ \LamDeltant{A} }_{ \LamDeltant{A'} } }\, \LamDeltant{t''_{{\mathrm{2}}}}   \redto^*   \langle  \Theta  \rangle^{ \LamDeltant{A} }_{ \LamDeltant{A'} }  \LamDeltant{t''_{{\mathrm{2}}}} $. Hence,
      \begin{center}
        \begin{math}
          \begin{array}{lll}
            \LamDeltamv{z} \, \LamDeltasym{(}  \LamDeltasym{[}  \LamDeltant{t}  \LamDeltasym{/}  \LamDeltamv{y}  \LamDeltasym{]}  \LamDeltasym{(}   \langle  \Theta  \rangle^{\uparrow^{ \LamDeltant{A} }_{ \LamDeltant{A'} } }\, \LamDeltant{t''_{{\mathrm{2}}}}   \LamDeltasym{)}  \LamDeltasym{)} & \redto^* & \LamDeltamv{z} \, \LamDeltasym{(}  \LamDeltasym{[}  \LamDeltant{t}  \LamDeltasym{/}  \LamDeltamv{y}  \LamDeltasym{]}  \LamDeltasym{(}   \langle  \Theta  \rangle^{ \LamDeltant{A} }_{ \LamDeltant{A'} }  \LamDeltant{t''_{{\mathrm{2}}}}   \LamDeltasym{)}  \LamDeltasym{)}\\
          \end{array}
        \end{math}
      \end{center}
      Furthermore, sense $(A,1,t') > (A,0, \langle  \Theta  \rangle^{ \LamDeltant{A} }_{ \LamDeltant{A'} }  \LamDeltant{t''_{{\mathrm{2}}}} )$, we also know from the induction hypothesis that 
      \begin{center}
        \begin{math}
          \begin{array}{lll}
            \LamDeltamv{z} \, \LamDeltasym{(}  \LamDeltasym{[}  \LamDeltant{t}  \LamDeltasym{/}  \LamDeltamv{y}  \LamDeltasym{]}  \LamDeltasym{(}   \langle  \Theta  \rangle^{ \LamDeltant{A} }_{ \LamDeltant{A'} }  \LamDeltant{t''_{{\mathrm{2}}}}   \LamDeltasym{)}  \LamDeltasym{)} & \redto^* & \LamDeltamv{z} \, \LamDeltasym{(}   [  \LamDeltant{t}  /  \LamDeltamv{y}  ]^{ \LamDeltant{A} }  \LamDeltasym{(}   \langle  \Theta  \rangle^{ \LamDeltant{A} }_{ \LamDeltant{A'} }  \LamDeltant{t''_{{\mathrm{2}}}}   \LamDeltasym{)}   \LamDeltasym{)}\\
                                         & =        &  \langle  \Theta  \rangle^{ \LamDeltant{A} }_{ \LamDeltant{A'} }  \LamDeltant{t'} \\
          \end{array}
        \end{math}
      \end{center}

    \item[Case.] Suppose $\LamDeltant{t'_{{\mathrm{2}}}} \equiv  \Delta  \LamDeltamv{y'}  :   \neg  \LamDeltasym{(}   \LamDeltant{A}  \to  \LamDeltant{A'}   \LamDeltasym{)}   .  \LamDeltant{t''_{{\mathrm{2}}}} $, for some $\LamDeltamv{y}$ and 
      $\LamDeltant{t''_{{\mathrm{2}}}}$, $ ( \LamDeltamv{x} , \LamDeltamv{z} , \LamDeltant{t} )  \in \Theta$. 
      Now using a fresh variable $\LamDeltamv{z_{{\mathrm{2}}}}$ we know 
      \begin{center}
        \small
        \begin{math}
          \begin{array}{lll}
             \langle  \Theta  \rangle^{\uparrow^{ \LamDeltant{A} }_{ \LamDeltant{A'} } }\, \LamDeltasym{(}  \LamDeltamv{x} \, \LamDeltasym{(}   \Delta  \LamDeltamv{y'}  :   \neg  \LamDeltasym{(}   \LamDeltant{A}  \to  \LamDeltant{A'}   \LamDeltasym{)}   .  \LamDeltant{t''_{{\mathrm{2}}}}   \LamDeltasym{)}  \LamDeltasym{)}  \\
            \ \ \ \ = \LamDeltasym{(}   \lambda  \LamDeltamv{y} :  \LamDeltant{A}  \to  \LamDeltant{A'}   .  \LamDeltasym{(}  \LamDeltamv{z} \, \LamDeltasym{(}  \LamDeltamv{y} \, \LamDeltant{t}  \LamDeltasym{)}  \LamDeltasym{)}   \LamDeltasym{)} \, \LamDeltasym{(}   \Delta  \LamDeltamv{y'}  :   \neg  \LamDeltasym{(}   \LamDeltant{A}  \to  \LamDeltant{A'}   \LamDeltasym{)}   .  \LamDeltasym{(}   \langle  \Theta  \rangle^{\uparrow^{ \LamDeltant{A} }_{ \LamDeltant{A'} } }\, \LamDeltant{t''_{{\mathrm{2}}}}   \LamDeltasym{)}   \LamDeltasym{)}\\ 
            \ \ \ \ \redto \LamDeltamv{z} \, \LamDeltasym{(}  \LamDeltasym{(}   \Delta  \LamDeltamv{y'}  :   \neg  \LamDeltasym{(}   \LamDeltant{A}  \to  \LamDeltant{A'}   \LamDeltasym{)}   .  \LamDeltasym{(}   \langle  \Theta  \rangle^{\uparrow^{ \LamDeltant{A} }_{ \LamDeltant{A'} } }\, \LamDeltant{t''_{{\mathrm{2}}}}   \LamDeltasym{)}   \LamDeltasym{)} \, \LamDeltant{t}  \LamDeltasym{)} \\
            \ \ \ \ \redto \LamDeltamv{z} \, \LamDeltasym{(}   \Delta  \LamDeltamv{z_{{\mathrm{2}}}}  :   \neg  \LamDeltant{A'}   .  \LamDeltasym{(}  \LamDeltasym{[}   \lambda  \LamDeltamv{y} :  \LamDeltant{A}  \to  \LamDeltant{A'}   .  \LamDeltasym{(}  \LamDeltamv{z_{{\mathrm{2}}}} \, \LamDeltasym{(}  \LamDeltamv{y} \, \LamDeltant{t}  \LamDeltasym{)}  \LamDeltasym{)}   \LamDeltasym{/}  \LamDeltamv{y'}  \LamDeltasym{]}  \LamDeltasym{(}   \langle  \Theta  \rangle^{\uparrow^{ \LamDeltant{A} }_{ \LamDeltant{A'} } }\, \LamDeltant{t''_{{\mathrm{2}}}}   \LamDeltasym{)}  \LamDeltasym{)}   \LamDeltasym{)} \\
            \ \ \ \ = \LamDeltamv{z} \, \LamDeltasym{(}   \Delta  \LamDeltamv{z_{{\mathrm{2}}}}  :   \neg  \LamDeltant{A'}   .  \LamDeltasym{(}   \langle  \Theta  \LamDeltasym{,}   ( \LamDeltamv{y'} , \LamDeltamv{z_{{\mathrm{2}}}} , \LamDeltant{t} )   \rangle^{\uparrow^{ \LamDeltant{A} }_{ \LamDeltant{A'} } }\, \LamDeltant{t''_{{\mathrm{2}}}}   \LamDeltasym{)}   \LamDeltasym{)}\\
          \end{array}
        \end{math}
      \end{center}
      Since $(A,1,t') > (A,1,\LamDeltant{t''_{{\mathrm{2}}}})$ we know from the 
      induction hypothesis that $ \langle  \Theta  \LamDeltasym{,}   ( \LamDeltamv{y'} , \LamDeltamv{z_{{\mathrm{2}}}} , \LamDeltant{t} )   \rangle^{\uparrow^{ \LamDeltant{A} }_{ \LamDeltant{A'} } }\, \LamDeltant{t''_{{\mathrm{2}}}}   \redto^*   \langle  \Theta  \LamDeltasym{,}   ( \LamDeltamv{y'} , \LamDeltamv{z_{{\mathrm{2}}}} , \LamDeltant{t} )   \rangle^{ \LamDeltant{A} }_{ \LamDeltant{A'} }  \LamDeltant{t''_{{\mathrm{2}}}} $.
      Thus, 
      \begin{center}
        \small
        \begin{math}
          \begin{array}{lll}
            \LamDeltamv{z} \, \LamDeltasym{(}   \Delta  \LamDeltamv{z_{{\mathrm{2}}}}  :   \neg  \LamDeltant{A'}   .  \LamDeltasym{(}   \langle  \Theta  \LamDeltasym{,}   ( \LamDeltamv{y'} , \LamDeltamv{z_{{\mathrm{2}}}} , \LamDeltant{t} )   \rangle^{\uparrow^{ \LamDeltant{A} }_{ \LamDeltant{A'} } }\, \LamDeltant{t''_{{\mathrm{2}}}}   \LamDeltasym{)}   \LamDeltasym{)} & \redto^* & \LamDeltamv{z} \, \LamDeltasym{(}   \Delta  \LamDeltamv{z_{{\mathrm{2}}}}  :   \neg  \LamDeltant{A'}   .  \LamDeltasym{(}   \langle  \Theta  \LamDeltasym{,}   ( \LamDeltamv{y'} , \LamDeltamv{z_{{\mathrm{2}}}} , \LamDeltant{t} )   \rangle^{ \LamDeltant{A} }_{ \LamDeltant{A'} }  \LamDeltant{t''_{{\mathrm{2}}}}   \LamDeltasym{)}   \LamDeltasym{)}\\
            & = &  \langle  \Theta  \rangle^{ \LamDeltant{A} }_{ \LamDeltant{A'} }  \LamDeltant{t'} .
          \end{array}
        \end{math}
      \end{center}
      
    \item[Case.] Suppose $\LamDeltant{t'_{{\mathrm{2}}}}$ is not an abstraction, and $ ( \LamDeltamv{x} , \LamDeltamv{z} , \LamDeltant{t} )  \in \Theta$. 
      Since $(\LamDeltant{A},1,\LamDeltant{t'}) > (A,1,\LamDeltant{t'_{{\mathrm{2}}}})$ we know from the 
      induction hypothesis that $ \langle  \Theta  \rangle^{\uparrow^{ \LamDeltant{A} }_{ \LamDeltant{A'} } }\, \LamDeltant{t'_{{\mathrm{2}}}}   \redto^*   \langle  \Theta  \rangle^{ \LamDeltant{A} }_{ \LamDeltant{A'} }  \LamDeltant{t'_{{\mathrm{2}}}} $.  Thus,
      \begin{center}
        \begin{math}
          \begin{array}{lll}
             \langle  \Theta  \rangle^{\uparrow^{ \LamDeltant{A} }_{ \LamDeltant{A'} } }\, \LamDeltant{t'}  & = &  \langle  \Theta  \rangle^{\uparrow^{ \LamDeltant{A} }_{ \LamDeltant{A'} } }\, \LamDeltasym{(}  \LamDeltamv{x} \, \LamDeltant{t'_{{\mathrm{2}}}}  \LamDeltasym{)} \\ 
            & = & \LamDeltamv{z} \, \LamDeltasym{(}   \langle  \Theta  \rangle^{\uparrow^{ \LamDeltant{A} }_{ \LamDeltant{A'} } }\, \LamDeltant{t'_{{\mathrm{2}}}}   \LamDeltasym{)}\\ 
            & \redto^* & \LamDeltamv{z} \, \LamDeltasym{(}   \langle  \Theta  \rangle^{ \LamDeltant{A} }_{ \LamDeltant{A'} }  \LamDeltant{t'_{{\mathrm{2}}}}   \LamDeltasym{)}\\
            & = &  \langle  \Theta  \rangle^{ \LamDeltant{A} }_{ \LamDeltant{A'} }  \LamDeltasym{(}  \LamDeltamv{x} \, \LamDeltant{t'_{{\mathrm{2}}}}  \LamDeltasym{)}  =  \langle  \Theta  \rangle^{ \LamDeltant{A} }_{ \LamDeltant{A'} }  \LamDeltant{t'} .
          \end{array}
        \end{math}
      \end{center}

    \item[Case.] Suppose $ ( \LamDeltamv{x} , \LamDeltamv{z} , \LamDeltant{t''} )  \not \in \Theta$ for any term $\LamDeltant{t''}$ and $\LamDeltamv{z}$.  
      Since $(A,1,t') > (A,1,\LamDeltant{t'_{{\mathrm{2}}}})$ we know from the 
      induction hypothesis that $ \langle  \Theta  \rangle^{\uparrow^{ \LamDeltant{A} }_{ \LamDeltant{A'} } }\, \LamDeltant{t'_{{\mathrm{2}}}}   \redto^*   \langle  \Theta  \rangle^{ \LamDeltant{A} }_{ \LamDeltant{A'} }  \LamDeltant{t'_{{\mathrm{2}}}} $.  Thus,
      \begin{center}
        \begin{math}
          \begin{array}{lll}
             \langle  \Theta  \rangle^{\uparrow^{ \LamDeltant{A} }_{ \LamDeltant{A'} } }\, \LamDeltant{t'}  & = &  \langle  \Theta  \rangle^{\uparrow^{ \LamDeltant{A} }_{ \LamDeltant{A'} } }\, \LamDeltasym{(}  \LamDeltamv{x} \, \LamDeltant{t'_{{\mathrm{2}}}}  \LamDeltasym{)} \\ 
            & = & \LamDeltamv{x} \, \LamDeltasym{(}   \langle  \Theta  \rangle^{\uparrow^{ \LamDeltant{A} }_{ \LamDeltant{A'} } }\, \LamDeltant{t'_{{\mathrm{2}}}}   \LamDeltasym{)}\\ 
            & \redto^* & \LamDeltamv{x} \, \LamDeltasym{(}   \langle  \Theta  \rangle^{ \LamDeltant{A} }_{ \LamDeltant{A'} }  \LamDeltant{t'_{{\mathrm{2}}}}   \LamDeltasym{)}\\
            & = &  \langle  \Theta  \rangle^{ \LamDeltant{A} }_{ \LamDeltant{A'} }  \LamDeltasym{(}  \LamDeltamv{x} \, \LamDeltant{t'_{{\mathrm{2}}}}  \LamDeltasym{)}  =  \langle  \Theta  \rangle^{ \LamDeltant{A} }_{ \LamDeltant{A'} }  \LamDeltant{t'} .
          \end{array}
        \end{math}
      \end{center}
    \end{itemize}
  \item[Case.] Suppose $\LamDeltant{t'_{{\mathrm{1}}}}$ is not a variable.  This case follows easily from
    the induction hypothesis.
  \end{itemize}
\end{itemize}

\ \\
\noindent Part two\\
\begin{itemize}
\item[Case.] Suppose $\LamDeltant{t'}$ is a variable $\LamDeltamv{x}$ or $\LamDeltamv{y}$ distinct from $\LamDeltamv{x}$.  
  Trivial in both cases.
  
\item[Case.] Suppose $\LamDeltant{t'} \equiv  \lambda  \LamDeltamv{y} : \LamDeltant{B_{{\mathrm{1}}}}  .  \LamDeltant{s} $.  Then
  $\LamDeltasym{[}  \LamDeltant{t}  \LamDeltasym{/}  \LamDeltamv{x}  \LamDeltasym{]}  \LamDeltasym{(}   \lambda  \LamDeltamv{y} : \LamDeltant{B_{{\mathrm{1}}}}  .  \LamDeltant{s}   \LamDeltasym{)}  \LamDeltasym{=}   \lambda  \LamDeltamv{y} : \LamDeltant{B_{{\mathrm{1}}}}  .  \LamDeltasym{(}  \LamDeltasym{[}  \LamDeltant{t}  \LamDeltasym{/}  \LamDeltamv{x}  \LamDeltasym{]}  \LamDeltant{s}  \LamDeltasym{)} $. 
  Now $\LamDeltamv{s}$ is a strict subexpression of $\LamDeltant{t'}$ so we can apply the second part of the induction hypothesis to obtain 
  $\LamDeltasym{[}  \LamDeltant{t}  \LamDeltasym{/}  \LamDeltamv{x}  \LamDeltasym{]}  \LamDeltant{s}  \redto^*   [  \LamDeltant{t}  /  \LamDeltamv{x}  ]^{ \LamDeltant{A} }  \LamDeltant{s} $.  At this point we can see that since 
  $ \lambda  \LamDeltamv{y} : \LamDeltant{B_{{\mathrm{1}}}}  .  \LamDeltasym{[}  \LamDeltant{t}  \LamDeltasym{/}  \LamDeltamv{x}  \LamDeltasym{]}  \LamDeltant{s}   \equiv  \LamDeltasym{[}  \LamDeltant{t}  \LamDeltasym{/}  \LamDeltamv{x}  \LamDeltasym{]}  \LamDeltasym{(}   \lambda  \LamDeltamv{y} : \LamDeltant{B_{{\mathrm{1}}}}  .  \LamDeltant{s}   \LamDeltasym{)}$ we may
  conclude that $ \lambda  \LamDeltamv{y} : \LamDeltant{B_{{\mathrm{1}}}}  .  \LamDeltasym{[}  \LamDeltant{t}  \LamDeltasym{/}  \LamDeltamv{x}  \LamDeltasym{]}  \LamDeltant{s}   \redto^*   \lambda  \LamDeltamv{y} : \LamDeltant{B_{{\mathrm{1}}}}  .   [  \LamDeltant{t}  /  \LamDeltamv{x}  ]^{ \LamDeltant{A} }  \LamDeltant{s}  $.

\item[Case.] Suppose $\LamDeltant{t'} \equiv  \Delta  \LamDeltamv{y}  :   \neg  \LamDeltant{B}   .  \LamDeltant{s} $.  Similar to the previous case.

\item[Case.] Suppose $\LamDeltant{t'}  \equiv  \LamDeltant{t'_{{\mathrm{1}}}} \, \LamDeltant{t'_{{\mathrm{2}}}}$.  By Lemma~\ref{lemma:totality_and_type_preservation}
  there exists terms $\LamDeltant{s_{{\mathrm{1}}}}$ and $\LamDeltant{s_{{\mathrm{2}}}}$
  such that $ [  \LamDeltant{t}  /  \LamDeltamv{x}  ]^{ \LamDeltant{A} }  \LamDeltant{t'_{{\mathrm{1}}}}   \LamDeltasym{=}  \LamDeltant{s_{{\mathrm{1}}}}$ and $ [  \LamDeltant{t}  /  \LamDeltamv{x}  ]^{ \LamDeltant{A} }  \LamDeltant{t'_{{\mathrm{2}}}}   \LamDeltasym{=}  \LamDeltant{s_{{\mathrm{2}}}}$.  Since
  $\LamDeltant{t'_{{\mathrm{1}}}}$ and $\LamDeltant{t'_{{\mathrm{2}}}}$ are strict subexpressions of $\LamDeltant{t'}$ we can apply the second part of the induction hypothesis to obtain
  $\LamDeltasym{[}  \LamDeltant{t}  \LamDeltasym{/}  \LamDeltamv{x}  \LamDeltasym{]}  \LamDeltant{t'_{{\mathrm{1}}}}  \redto^*  \LamDeltant{s_{{\mathrm{1}}}}$ and $\LamDeltasym{[}  \LamDeltant{t}  \LamDeltasym{/}  \LamDeltamv{x}  \LamDeltasym{]}  \LamDeltant{t'_{{\mathrm{2}}}}  \redto^*  \LamDeltant{s_{{\mathrm{2}}}}$.  Now we case
  split on whether or not $\LamDeltant{s_{{\mathrm{1}}}}$ is a $\lambda$-abstraction and $\LamDeltant{t'_{{\mathrm{1}}}}$ is not,
  a $\Delta$-abstraction and $\LamDeltant{t'_{{\mathrm{1}}}}$ is not, or $\LamDeltant{s_{{\mathrm{1}}}}$ is not a $\lambda$-abstraction or a $\Delta$-abstraction.  If
  $\LamDeltant{s_{{\mathrm{1}}}}$ is not a $\lambda$-abstraction or a $\Delta$-abstraction then 
  $ [  \LamDeltant{t}  /  \LamDeltamv{x}  ]^{ \LamDeltant{A} }  \LamDeltant{t'}   \LamDeltasym{=}  \LamDeltasym{(}   [  \LamDeltant{t}  /  \LamDeltamv{x}  ]^{ \LamDeltant{A} }  \LamDeltant{t'_{{\mathrm{1}}}}   \LamDeltasym{)} \, \LamDeltasym{(}   [  \LamDeltant{t}  /  \LamDeltamv{x}  ]^{ \LamDeltant{A} }  \LamDeltant{t'_{{\mathrm{2}}}}   \LamDeltasym{)}  \equiv  \LamDeltant{s_{{\mathrm{1}}}} \, \LamDeltant{s_{{\mathrm{2}}}}$. Thus, by two applications of the
  induction hypothesis, $\LamDeltasym{[}  \LamDeltant{t}  \LamDeltasym{/}  \LamDeltamv{x}  \LamDeltasym{]}  \LamDeltant{t'}  \redto^*   [  \LamDeltant{t}  /  \LamDeltamv{x}  ]^{ \LamDeltant{A} }  \LamDeltant{t'} $, because $\LamDeltasym{[}  \LamDeltant{t}  \LamDeltasym{/}  \LamDeltamv{x}  \LamDeltasym{]}  \LamDeltant{t'}  \LamDeltasym{=}  \LamDeltasym{(}  \LamDeltasym{[}  \LamDeltant{t}  \LamDeltasym{/}  \LamDeltamv{x}  \LamDeltasym{]}  \LamDeltant{t'_{{\mathrm{1}}}}  \LamDeltasym{)} \, \LamDeltasym{(}  \LamDeltasym{[}  \LamDeltant{t}  \LamDeltasym{/}  \LamDeltamv{x}  \LamDeltasym{]}  \LamDeltant{t'_{{\mathrm{2}}}}  \LamDeltasym{)}$.

  \ \\
  Suppose $\LamDeltant{s_{{\mathrm{1}}}}  \equiv   \lambda  \LamDeltamv{y} : \LamDeltant{B'}  .  \LamDeltant{s'_{{\mathrm{1}}}} $ and $\LamDeltant{t'_{{\mathrm{1}}}}$ is not a $\lambda$-abstraction.  
  By Lemma~\ref{lemma:ctype_props} there exists a type $\LamDeltant{B''}$ such that
  $ \textsf{ctype}_{ \LamDeltant{A} }( \LamDeltamv{x} , \LamDeltant{t'_{{\mathrm{1}}}} )  = \LamDeltant{B''}$, $\LamDeltant{B''}  \equiv   \LamDeltant{B'}  \to  \LamDeltant{B} $, and $\LamDeltant{B''}$ is a subexpression
  of $\LamDeltant{A}$.  Then by the definition of the hereditary substitution function 
  $ [  \LamDeltant{t}  /  \LamDeltamv{x}  ]^{ \LamDeltant{A} }  \LamDeltasym{(}  \LamDeltant{t'_{{\mathrm{1}}}} \, \LamDeltant{t'_{{\mathrm{2}}}}  \LamDeltasym{)}   \LamDeltasym{=}   [  \LamDeltant{s_{{\mathrm{2}}}}  /  \LamDeltamv{y}  ]^{ \LamDeltant{B'} }  \LamDeltant{s'_{{\mathrm{1}}}} $.
  Now we know $\LamDeltant{A}  \LamDeltasym{>}  \LamDeltant{B'}$ so we can apply the second part of the induction hypothesis to obtain 
  $\LamDeltasym{[}  \LamDeltant{s_{{\mathrm{2}}}}  \LamDeltasym{/}  \LamDeltamv{y}  \LamDeltasym{]}  \LamDeltant{s'_{{\mathrm{1}}}}  \redto^*   [  \LamDeltant{s_{{\mathrm{2}}}}  /  \LamDeltamv{y}  ]^{ \LamDeltant{B'} }  \LamDeltant{s'_{{\mathrm{1}}}} $. By knowing that 
  $ \LamDeltasym{(}  \LamDeltasym{(}   \lambda  \LamDeltamv{y} : \LamDeltant{B'}  .  \LamDeltant{s'_{{\mathrm{1}}}}   \LamDeltasym{)} \, \LamDeltant{s_{{\mathrm{2}}}}  \LamDeltasym{)}  \redto  \LamDeltasym{(}  \LamDeltasym{[}  \LamDeltant{s_{{\mathrm{2}}}}  \LamDeltasym{/}  \LamDeltamv{y}  \LamDeltasym{]}  \LamDeltant{s'_{{\mathrm{1}}}}  \LamDeltasym{)} $ and
  by the previous fact we know $\LamDeltasym{(}   \lambda  \LamDeltamv{y} : \LamDeltant{B'}  .  \LamDeltant{s'_{{\mathrm{1}}}}   \LamDeltasym{)} \, \LamDeltant{s_{{\mathrm{2}}}}  \redto^*   [  \LamDeltant{s_{{\mathrm{2}}}}  /  \LamDeltamv{y}  ]^{ \LamDeltant{B'} }  \LamDeltant{s'_{{\mathrm{1}}}} $.
  We now make use of the well known result of full $\beta$-reduction.  The
  result is stated as
  \begin{center}
    \begin{math}
      $$\mprset{flushleft}
      \inferrule* [right=] {
        \LamDeltant{a}  \redto^*  \LamDeltant{a'}
        \\\\
        \LamDeltant{b}  \redto^*  \LamDeltant{b'}
        \\
        \LamDeltant{a'} \, \LamDeltant{b'}  \redto^*  \LamDeltant{c}
      }{\LamDeltant{a} \, \LamDeltant{b}  \redto^*  \LamDeltant{c}}
    \end{math}
  \end{center}
  where $\LamDeltant{a}$, $\LamDeltant{a'}$, $\LamDeltant{b}$, $\LamDeltant{b'}$, and $\LamDeltant{c}$ are all terms.  We apply this
  result by instantiating $\LamDeltant{a}$, $\LamDeltant{a'}$, $\LamDeltant{b}$, $\LamDeltant{b'}$, and $\LamDeltant{c}$ with
  $\LamDeltasym{[}  \LamDeltant{t}  \LamDeltasym{/}  \LamDeltamv{x}  \LamDeltasym{]}  \LamDeltant{t'_{{\mathrm{1}}}}$, $\LamDeltant{s_{{\mathrm{1}}}}$, $\LamDeltasym{[}  \LamDeltant{t}  \LamDeltasym{/}  \LamDeltamv{x}  \LamDeltasym{]}  \LamDeltant{t'_{{\mathrm{2}}}}$, $\LamDeltant{s_{{\mathrm{2}}}}$, and $ [  \LamDeltant{s_{{\mathrm{2}}}}  /  \LamDeltamv{y}  ]^{ \LamDeltant{B'} }  \LamDeltant{s'_{{\mathrm{1}}}} $ 
  respectively.  Therefore, $\LamDeltasym{[}  \LamDeltant{t}  \LamDeltasym{/}  \LamDeltamv{x}  \LamDeltasym{]}  \LamDeltasym{(}  \LamDeltant{t'_{{\mathrm{1}}}} \, \LamDeltant{t'_{{\mathrm{2}}}}  \LamDeltasym{)}  \redto^*   [  \LamDeltant{s_{{\mathrm{2}}}}  /  \LamDeltamv{y}  ]^{ \LamDeltant{B'} }  \LamDeltant{s'_{{\mathrm{1}}}} $.        

  \ \\
  Suppose $\LamDeltant{s_{{\mathrm{1}}}}  \equiv   \Delta  \LamDeltamv{y}  :   \neg  \LamDeltasym{(}   \LamDeltant{B'}  \to  \LamDeltant{B}   \LamDeltasym{)}   .  \LamDeltant{s'_{{\mathrm{1}}}} $ and $\LamDeltant{t'_{{\mathrm{1}}}}$ is not a $\Delta$-abstraction.
  By Lemma~\ref{lemma:ctype_props} there exists a type $\LamDeltant{B''}$ such that
  $ \textsf{ctype}_{ \LamDeltant{A} }( \LamDeltamv{x} , \LamDeltant{t'_{{\mathrm{1}}}} )  = \LamDeltant{B''}$, $\LamDeltant{B''}  \equiv   \LamDeltant{B'}  \to  \LamDeltant{B} $, and $\LamDeltant{B''}$ is a subexpression
  of $\LamDeltant{A}$.  Then by the definition of the hereditary substitution function
  $ [  \LamDeltant{t}  /  \LamDeltamv{x}  ]^{ \LamDeltant{A} }  \LamDeltasym{(}  \LamDeltant{t'_{{\mathrm{1}}}} \, \LamDeltant{t'_{{\mathrm{2}}}}  \LamDeltasym{)}   \LamDeltasym{=}   \Delta  \LamDeltamv{z}  :   \neg  \LamDeltant{B}   .   \langle   ( \LamDeltamv{y} , \LamDeltamv{z} , \LamDeltant{s_{{\mathrm{2}}}} )   \rangle^{ \LamDeltant{B'} }_{ \LamDeltant{B} }  \LamDeltant{s'_{{\mathrm{1}}}}  $, where $\LamDeltamv{z}$ is fresh variable.
  Now
  \begin{center}
    \begin{math}
      \begin{array}{lll}
        \LamDeltasym{[}  \LamDeltant{t}  \LamDeltasym{/}  \LamDeltamv{x}  \LamDeltasym{]}  \LamDeltasym{(}  \LamDeltant{t'_{{\mathrm{1}}}} \, \LamDeltant{t'_{{\mathrm{2}}}}  \LamDeltasym{)} & = & \LamDeltasym{(}  \LamDeltasym{[}  \LamDeltant{t}  \LamDeltasym{/}  \LamDeltamv{x}  \LamDeltasym{]}  \LamDeltant{t'_{{\mathrm{1}}}}  \LamDeltasym{)} \, \LamDeltasym{(}  \LamDeltasym{[}  \LamDeltant{t}  \LamDeltasym{/}  \LamDeltamv{x}  \LamDeltasym{]}  \LamDeltant{t'_{{\mathrm{2}}}}  \LamDeltasym{)}\\
        & \redto^* & \LamDeltant{s_{{\mathrm{1}}}} \, \LamDeltant{s_{{\mathrm{2}}}}\\
        & \equiv & \LamDeltasym{(}   \Delta  \LamDeltamv{y}  :   \neg  \LamDeltasym{(}   \LamDeltant{B'}  \to  \LamDeltant{B}   \LamDeltasym{)}   .  \LamDeltant{s'_{{\mathrm{1}}}}   \LamDeltasym{)} \, \LamDeltant{s_{{\mathrm{2}}}}\\
        & \redto &  \Delta  \LamDeltamv{z}  :   \neg  \LamDeltant{B}   .  \LamDeltasym{[}   \lambda  \LamDeltamv{y'} :  \LamDeltant{B'}  \to  \LamDeltant{B}   .  \LamDeltasym{(}  \LamDeltamv{z} \, \LamDeltasym{(}  \LamDeltamv{y'} \, \LamDeltant{s_{{\mathrm{2}}}}  \LamDeltasym{)}  \LamDeltasym{)}   \LamDeltasym{/}  \LamDeltamv{y}  \LamDeltasym{]}  \LamDeltant{s'_{{\mathrm{1}}}} \\
        & = &  \Delta  \LamDeltamv{z}  :   \neg  \LamDeltant{B}   .  \LamDeltasym{(}   \langle   ( \LamDeltamv{y} , \LamDeltamv{z} , \LamDeltant{s_{{\mathrm{2}}}} )   \rangle^{\uparrow^{ \LamDeltant{B'} }_{ \LamDeltant{B} } }\, \LamDeltant{s'_{{\mathrm{1}}}}   \LamDeltasym{)} 
      \end{array}
    \end{math}
  \end{center}
  It suffices to show that $ \Delta  \LamDeltamv{z}  :   \neg  \LamDeltant{B}   .  \LamDeltasym{(}   \langle   ( \LamDeltamv{y} , \LamDeltamv{z} , \LamDeltant{s_{{\mathrm{2}}}} )   \rangle^{\uparrow^{ \LamDeltant{B'} }_{ \LamDeltant{B} } }\, \LamDeltant{s'_{{\mathrm{1}}}}   \LamDeltasym{)}   \redto^*   \Delta  \LamDeltamv{z}  :   \neg  \LamDeltant{B}   .   \langle   ( \LamDeltamv{y} , \LamDeltamv{z} , \LamDeltant{s_{{\mathrm{2}}}} )   \rangle^{ \LamDeltant{B'} }_{ \LamDeltant{B} }  \LamDeltant{s'_{{\mathrm{1}}}}  $,
  but this follows from the induction hypothesis, because $(\LamDeltant{A},0,\LamDeltant{t'}) > (\LamDeltant{B'},1,\LamDeltant{s'_{{\mathrm{1}}}})$.
\end{itemize}
% subsection soundness_with_respect_to_reduction (end)

\subsection{Proof of Type Soundness}
\label{subsec:proof_of_type_soundness}
This is a proof by induction on the assumed typing derivation.
\begin{itemize}
\item[Case.] \ \\
  \begin{center}
    $\LamDeltadruleAx{}$
  \end{center}
  Trivial.
  
\item[Case.] \ \\
  \begin{center}
    $\LamDeltadruleLam{}$
  \end{center}
  By the induction hypothesis $\LamDeltant{t} \in \interp{\LamDeltant{B}}_{\Gamma  \LamDeltasym{,}  \LamDeltamv{x}  \LamDeltasym{:}  \LamDeltant{A}}$.  By the definition of the
  interpretation of types $\LamDeltant{t} \normto \LamDeltant{n} \in \interp{\LamDeltant{B}}_{\Gamma  \LamDeltasym{,}  \LamDeltamv{x}  \LamDeltasym{:}  \LamDeltant{A}}$ and 
  $ \Gamma  \LamDeltasym{,}  \LamDeltamv{x}  \LamDeltasym{:}  \LamDeltant{A}  \vdash  \LamDeltant{n}  :  \LamDeltant{B} $.  Thus, by applying the $\lambda$-abstraction type-checking
  rule, $ \Gamma  \vdash   \lambda  \LamDeltamv{x} : \LamDeltant{A}  .  \LamDeltant{n}   :   \LamDeltant{A}  \to  \LamDeltant{B}  $, hence by the definition of the 
  interpretation of types $ \lambda  \LamDeltamv{x} : \LamDeltant{A}  .  \LamDeltant{n}  \in \interp{ \LamDeltant{A}  \to  \LamDeltant{B} }_\Gamma$.  Therefore,
  $ \lambda  \LamDeltamv{x} : \LamDeltant{A}  .  \LamDeltant{t}  \normto  \lambda  \LamDeltamv{x} : \LamDeltant{A}  .  \LamDeltant{n}  \in \interp{ \LamDeltant{A}  \to  \LamDeltant{B} }_\Gamma$.
  
\item[Case.] \ \\
  \begin{center}
    $\LamDeltadruleDelta{}$
  \end{center}
  Similar to the previous case.
  
\item[Case.] \ \\
  \begin{center}
    $\LamDeltadruleApp{}$
  \end{center}
  
  By the induction hypothesis we know $\LamDeltant{t_{{\mathrm{1}}}} \in \interp{ \LamDeltant{A}  \to  \LamDeltant{B} }_\Gamma$ and $\LamDeltant{t_{{\mathrm{2}}}} \in \interp{\LamDeltant{A}}_\Gamma$.
  So by the definition of the interpretation of types we know there exists normal forms $\LamDeltant{n_{{\mathrm{1}}}}$ and $\LamDeltant{n_{{\mathrm{2}}}}$
  such that $\LamDeltant{t_{{\mathrm{1}}}}  \redto^*  \LamDeltant{n_{{\mathrm{1}}}} \in \interp{ \LamDeltant{A}  \to  \LamDeltant{B} }_\Gamma$ and $\LamDeltant{t_{{\mathrm{2}}}}  \redto^*  \LamDeltant{n_{{\mathrm{2}}}} \in \interp{\LamDeltant{A}}_\Gamma$. Assume $\LamDeltamv{y}$ is a fresh
  variable in $\LamDeltant{n_{{\mathrm{1}}}}$ and $\LamDeltant{n_{{\mathrm{2}}}}$ of type $\LamDeltant{A}$.    
  Then by hereditary 
  substitution for the interpretation of types (Lemma~\ref{lemma:substitution_for_the_interpretation_of_types}) 
  $ [  \LamDeltant{n_{{\mathrm{1}}}}  /  \LamDeltamv{y}  ]^{ \LamDeltant{A} }  \LamDeltasym{(}  \LamDeltamv{y} \, \LamDeltant{n_{{\mathrm{2}}}}  \LamDeltasym{)}  \in \interp{\LamDeltant{B}}_\Gamma$.  
  It suffices to show that $\LamDeltant{t_{{\mathrm{1}}}} \, \LamDeltant{t_{{\mathrm{2}}}}  \redto^*   [  \LamDeltant{n_{{\mathrm{1}}}}  /  \LamDeltamv{y}  ]^{ \LamDeltant{A} }  \LamDeltasym{(}  \LamDeltamv{y} \, \LamDeltant{n_{{\mathrm{2}}}}  \LamDeltasym{)} $.  This is an easy consequence of soundness with respect
  to reduction (Lemma~\ref{lemma:soundness_reduction}), that is, $\LamDeltant{t_{{\mathrm{1}}}} \, \LamDeltant{t_{{\mathrm{2}}}}  \redto^*  \LamDeltant{n_{{\mathrm{1}}}} \, \LamDeltant{n_{{\mathrm{2}}}} = \LamDeltasym{[}  \LamDeltant{n_{{\mathrm{1}}}}  \LamDeltasym{/}  \LamDeltamv{y}  \LamDeltasym{]}  \LamDeltasym{(}  \LamDeltamv{y} \, \LamDeltant{n_{{\mathrm{2}}}}  \LamDeltasym{)}$ 
  and by soundness with respect to reduction $\LamDeltasym{[}  \LamDeltant{n_{{\mathrm{1}}}}  \LamDeltasym{/}  \LamDeltamv{y}  \LamDeltasym{]}  \LamDeltasym{(}  \LamDeltamv{y} \, \LamDeltant{n_{{\mathrm{2}}}}  \LamDeltasym{)}  \redto^*   [  \LamDeltant{n_{{\mathrm{1}}}}  /  \LamDeltamv{y}  ]^{ \LamDeltant{A} }  \LamDeltasym{(}  \LamDeltamv{y} \, \LamDeltant{n_{{\mathrm{2}}}}  \LamDeltasym{)} $.  Therefore, 
  $\LamDeltant{t_{{\mathrm{1}}}} \, \LamDeltant{t_{{\mathrm{2}}}} \in \interp{\LamDeltant{B}}_\Gamma$.  
\end{itemize}
% subsection proof_of_type_soundness (end)
% section proofs (end)

%%% Local Variables: 
%%% mode: latex
%%% TeX-master: "paper"
%%% End:

\end{document}